\documentclass[12pt,oneside]{amsart}

\usepackage{amsmath, amssymb, amsthm, amscd,
}

\usepackage[matrix,arrow,curve]{xy}

\usepackage{amsaddr}

\allowdisplaybreaks[4]

\newcommand{\eval}[2][\right]{\relax
  \ifx#1\right\relax \left.\fi#2#1\rvert}

\newcommand{\pd}{{\partial}}

\newcommand{\frl}{\mathfrak{F}}
\newcommand{\frid}{\mathfrak{I}}
\newcommand{\fla}{\mathbf{A}}
\newcommand{\flb}{\mathbf{B}}
\newcommand{\flz}{\mathbf{Z}}

\newcommand{\wga}{\mathrm{A}}
\newcommand{\wgb}{\mathrm{B}}
\newcommand{\mR}{\mathfrak{R}}

\newcommand{\agn}{\mathfrak{A}}

\newcommand{\ml}{q}

\newcommand{\al}{{\alpha}}
\newcommand{\be}{{\beta}}
\newcommand{\la}{{\lambda}}
\newcommand{\za}{{\mathbf{V}}}
\newcommand{\ua}{U_a}

\newcommand{\beq}{\begin{equation}}

\newcommand{\zf}{\mathbf{F}}
\newcommand{\xtf}{\mathbb{L}}

\newcommand{\wea}{\mathfrak{W}_a}
\newcommand{\swe}{\mathfrak{R}}

\newcommand{\qalg}{\mathcal{Q}}
\newcommand{\bG}{\mathbf{G}}

\newcommand{\ipol}{\mathcal{I}}

\newcommand{\zp}{\mathbb{Z}_{\ge 0}}
\newcommand{\zsp}{\mathbb{Z}_{>0}}

\newcommand{\er}{\eqref}
\newcommand{\cl}{\colon}

\newcommand{\ee}{\end{equation}}

\newcommand{\bmu}{\begin{multline}}
\newcommand{\emul}{\end{multline}}

\newcommand{\CE}{\mathcal{E}}
\newcommand{\ce}{\mathcal{E}}
\newcommand{\CA}{\mathcal{A}}
\newcommand{\CB}{\mathcal{B}}

\newcommand{\anA}{\mathsf{A}}
\newcommand{\anB}{\mathsf{B}}

\newcommand{\mat}{{\mathcal{M}}}

\newcommand{\zd}{\mathrm{D}}

\newcommand{\ad}{{\rm ad\,}}

\renewcommand{\sl}{\mathfrak{sl}}
\newcommand{\gl}{\mathfrak{gl}}
\newcommand{\mg}{\mathfrak{g}}

\newcommand{\bl}{\mathfrak{L}}

\newcommand{\ga}{\mathbb{A}}
\newcommand{\gb}{\mathbb{B}}

\newcommand{\hga}{\hat{\mathbb{A}}}
\newcommand{\hgb}{\hat{\mathbb{B}}}

\newcommand{\hal}{{\hat{\alpha}}}
\newcommand{\hbeta}{{\hat{\beta}}}

\newcommand{\mh}{\mathfrak{H}}

\newcommand{\so}{\mathfrak{so}}
\newcommand{\mic}{\mathfrak{I}}

\newcommand{\lb}{\label}

\newcommand{\vf}{\varphi}

\newcommand{\Com}{\mathbb{C}}
\newcommand{\fik}{\mathbb{K}}

\newcommand{\un}{\mathrm{U}}

\DeclareMathOperator{\fd}{\mathbb{F}}
\DeclareMathOperator{\fds}{\mathbb{F}}

\newcommand{\hmm}{\sigma}
\newcommand{\hrf}{\mu}

\newcommand{\ai}{p}
\newcommand{\bi}{q}

\newcommand{\oc}{p}
\newcommand{\ocn}{p}

\newcommand{\dc}{q}
\newcommand{\nv}{m}
\newcommand{\eo}{d}
\newcommand{\sm}{N}

\newcommand{\ost}{\mathbb{U}}

\newtheorem{theorem}{Theorem}

\newtheorem{lemma}{Lemma}

\theoremstyle{definition}

\newtheorem{example}{Example}
\newtheorem{remark}{Remark}

\begin{document}




\title[On Lie algebras responsible for zero-curvature representations]{On Lie algebras responsible for zero-curvature representations 
of multicomponent $(1+1)$-dimensional evolution PDEs}
\date{}

\author{Sergei Igonin\qquad\qquad Gianni Manno} 
\address{Department of Mathematical Sciences, 
Politecnico di Torino, \\ 
Corso Duca degli Abruzzi 24, 10129 Torino, Italy}

\author{}
\email{s-igonin@yandex.ru, d027258@polito.it}

\begin{abstract}

Zero-curvature representations (ZCRs) are one of the main 
tools in the theory of integrable $(1+1)$-dimensional PDEs.
According to the preprint arXiv:1212.2199, 
for any given $(1+1)$-dimensional evolution PDE one can define 
a sequence of Lie algebras $\fd^\oc$, $\oc=0,1,2,3,\dots$, 
such that representations of these algebras classify all ZCRs of the PDE 
up to local gauge equivalence. 
ZCRs depending on derivatives of arbitrary finite order are allowed.
Furthermore, these algebras provide necessary conditions for existence of 
B\"acklund transformations between two given PDEs.
The algebras $\fd^\oc$ are defined in arXiv:1212.2199 
in terms of generators and relations.

In the present paper, 
we describe some methods to study the structure of the algebras $\fd^\oc$ 
for multicomponent $(1+1)$-dimensional evolution PDEs.
Using these methods, we compute the explicit structure (up to non-essential 
nilpotent ideals) of the Lie algebras $\fd^\oc$ for 
the Landau-Lifshitz, nonlinear Schr\"odinger equations, 
and for the $n$-component
Landau-Lifshitz system of Golubchik and Sokolov for any $n>3$.

In particular, this means that for the $n$-component Landau-Lifshitz system  
we classify all ZCRs (depending on derivatives of arbitrary finite order), 
up to local gauge equivalence and up to killing nilpotent ideals 
in the corresponding Lie algebras. 
As a result of the classification, one obtains two main non-equivalent ZCRs:
a well-known ZCR with values in the infinite-dimensional Lie algebra 
of certain $(n+1)\times (n+1)$ matrix-valued functions on some algebraic curve and 
a very different ZCR with values in the Lie algebra $\so_{n-1}$.
We prove that any other ZCR of the $n$-component Landau-Lifshitz system is equivalent 
to a reduction of these two ZCRs.

The presented methods to classify ZCRs can be applied also 
to other $(1+1)$-dimensional evolution PDEs.
Furthermore, the obtained results can be used for proving non-existence of B\"acklund transformations between some PDEs, which will be described in forthcoming publications.

\end{abstract}


\maketitle

\tableofcontents

\section{Introduction}

\subsection{The main ideas}
\lb{secintr}



Let $\nv$ be a positive integer.
An $\nv$-component $(1+1)$-dimensional evolution PDE 
for functions $u^1(x,t),\dots,u^\nv(x,t)$ is a PDE of the form
\begin{gather}
\label{sys_intr}
\frac{\pd u^i}{\pd t}
=F^i(x,t,u^1,\dots,u^\nv,\,u^1_1,\dots,u^\nv_1,\dots,u^1_{\eo},\dots,u^\nv_{\eo}),\\
\notag
u^i=u^i(x,t),\quad\qquad u^i_k=\frac{\pd^k u^i}{\pd x^k},\quad\qquad  
i=1,\dots,\nv,\quad\qquad k\in\zsp. 
\end{gather}
Here the number $\eo\ge 1$ is such that the functions $F^i$ may depend only 
on the variables $x$, $t$, $u^j$, $u^j_k$ for $k\le\eo$.

A large part of the 
theory of integrable systems is devoted to the study of such PDEs.
This class of PDEs includes many celebrated equations of mathematical physics 
(e.g., the KdV, Krichever-Novikov, 
Landau-Lifshitz, nonlinear Schr\"odinger equations). 
Many more PDEs can be written in the evolution form~\er{sys_intr} 
after a suitable change of variables.

To simplify notation, we set $u^j_0=u^j$ and $u^j_t=\pd u^j/\pd t$ for $j=1,\dots,\nv$.
Then the right-hand side of~\er{sys_intr} can be written as 
$F^i(x,t,u^j_0,u^j_1,\dots,u^j_{\eo})$, and the PDE~\er{sys_intr} reads 
\beq
\lb{uitfi}
u^i_t=F^i(x,t,u^j_0,u^j_1,\dots,u^j_{\eo}),\qquad\quad
i=1,\dots,\nv.
\ee
So~\er{uitfi} is a compact form of~\er{sys_intr}.

In this paper, integrability of PDEs is understood 
in the sense of soliton theory and the inverse scattering method. 
This is sometimes called $S$-integrability. 

It is well known that, in order to investigate possible 
integrability properties of~\er{sys_intr}, 
one needs to consider zero-curvature representations (ZCRs).

Let $\mg$ be a finite-dimensional Lie algebra.  
For a PDE of the form~\er{sys_intr}, 
a \emph{zero-curvature representation \textup{(}ZCR\textup{)} 
with values in~$\mg$} is given by $\mg$-valued functions 
\beq
\lb{mnoc}
A=A(x,t,u^j_0,u^j_1,\dots,u^j_\oc),\qquad\quad 
B=B(x,t,u^j_0,u^j_1,\dots,u^j_{\oc+\eo-1})
\ee
satisfying
\beq
\lb{mnzcr}
D_x(B)-D_t(A)+[A,B]=0.
\ee

The \emph{total derivative operators} $D_x$, $D_t$ in~\er{mnzcr} are 
\beq
\lb{evdxdt}
D_x=\frac{\pd}{\pd x}+\sum_{\substack{i=1,\dots,\nv,\\ k\ge 0}} 
u^i_{k+1}\frac{\pd}{\pd u^i_k},\qquad\qquad
D_t=\frac{\pd}{\pd t}+\sum_{\substack{i=1,\dots,\nv,\\ k\ge 0}} 
D_x^k\big(F^i(x,t,u^j_0,u^j_1,\dots,u^j_{\eo})\big)\frac{\pd}{\pd u^i_k}.
\ee

The number $\oc$ in~\er{mnoc} is such that  
the function $A$ may depend only on 
the variables $x$, $t$, $u^j_{k}$ for $k\le\oc$ and $j=1,\dots,\nv$.
Then equation~\er{mnzcr} implies that 
the function $B$ may depend only on $x$, $t$, $u^j_{k'}$ for $k'\le\oc+\eo-1$.

Such ZCRs are said to be \emph{of order~$\le\oc$}. 
In other words, a ZCR given by $A$, $B$ is of order~$\le\oc$ iff 
$\dfrac{\pd A}{\pd u^i_l}=0$ for all $l>\oc$. 

\begin{remark}
The right-hand side $F^i(x,t,u^j_0,u^j_1,\dots,u^j_{\eo})$ 
of~\er{sys_intr} appears in condition~\er{mnzcr}, 
because $F^i$ appears 
in the formula for the operator $D_t$ in~\er{evdxdt}.

Note that~\er{mnzcr} can be written as $[D_x+A,\,D_t+B]=0$.
Condition~\er{mnzcr} is equivalent to the fact that 
the auxiliary linear system 
\beq
\lb{auxls}
\pd_x(W)=-AW,\qquad\quad
\pd_t(W)=-BW
\ee
is compatible modulo~\er{sys_intr}.
Here $W=W(x,t)$ is an invertible $\sm\times\sm$ matrix-function.
\end{remark}

\begin{remark}
When we consider a function $Q=Q(x,t,u^j_0,u^j_1,\dots,u^j_l)$ 
for some $l\in\zp$, we always assume that this function is analytic 
on an open subset of the manifold with the coordinates 
$x,t,u^j_0,u^j_1,\dots,u^j_l$ for $j=1,\dots,\nv$.
For example, $Q$ may be a meromorphic function,  
because a meromorphic function is analytic on some open subset of the manifold.
In particular, this applies to the functions~\er{mnoc}.
\end{remark}

We study the following problem. 
How to describe all ZCRs~\er{mnoc},~\er{mnzcr} for a given PDE~\er{sys_intr}? 

In the case when $\oc=0$ and the functions $F^i$, $A$, $B$ do not depend on $x$, $t$, 
a partial answer to this question is provided by the Wahlquist-Estabrook prolongation
method (WE method for short). 
Namely, for a given PDE of the form 
$u^i_t=F^i(u^j_0,u^j_1,\dots,u^j_{\eo})$, $i=1,\dots,\nv$, 
the WE method constructs a Lie algebra so that ZCRs of the form
\beq
\lb{wecov}
A=A(u^j_0),\qquad B=B(u^j_0,u^j_1,\dots,u^j_{\eo-1}),\qquad
D_x(B)-D_t(A)+[A,B]=0
\ee
correspond to representations of this 
algebra (see, e.g.,~\cite{dodd,Prol,mll-2012}). 
It is called the \emph{Wahlquist-Estabrook prolongation algebra}.
Note that in~\er{wecov} the function $A(u^j_0)$ depends only 
on $u^j_0$, $j=1,\dots,\nv$.

To study the general case of ZCRs~\er{mnoc},~\er{mnzcr} 
with arbitrary $\oc$ for any equation~\er{sys_intr}, 
we need to consider gauge transformations. 

Without loss of generality, one can assume that $\mg$ is a Lie subalgebra 
of $\gl_\sm$ for some $\sm\in\zsp$, where $\gl_\sm$ is the algebra of 
$\sm\times\sm$ matrices with entries from $\mathbb{R}$ or $\mathbb{C}$. 
So our considerations are applicable to both cases $\gl_\sm=\gl_\sm(\mathbb{R})$ 
and $\gl_\sm=\gl_\sm(\mathbb{C})$.
And we denote by $\mathrm{GL}_\sm$ the group of invertible $\sm\times\sm$ matrices.

Let $\fik$ be either $\Com$ or $\mathbb{R}$.
Then $\gl_\sm=\gl_\sm(\fik)$ and $\mathrm{GL}_\sm=\mathrm{GL}_\sm(\fik)$.
In this paper, all algebras are supposed to be over the field~$\fik$.

Let $\mathcal{G}\subset\mathrm{GL}_\sm$ be the connected matrix Lie group 
corresponding to the Lie algebra $\mg\subset\gl_\sm$.
(That is, $\mathcal{G}$ is the connected 
immersed Lie subgroup of $\mathrm{GL}_\sm$ 
corresponding to the Lie subalgebra $\mg\subset\gl_\sm$.)
A \emph{gauge transformation} is given by an invertible matrix-function 
$G=G(x,t,u^j_0,u^j_1,\dots,u^j_l)$ with values in~$\mathcal{G}$.

For any ZCR~\er{mnoc},~\er{mnzcr} and any gauge 
transformation $G=G(x,t,u^j_0,u^j_1,\dots,u^j_l)$, the functions 
\beq
\lb{mnprint}
\tilde{A}=GAG^{-1}-D_x(G)\cdot G^{-1},\qquad\qquad
\tilde{B}=GBG^{-1}-D_t(G)\cdot G^{-1}
\ee
satisfy $D_x(\tilde{B})-D_t(\tilde{A})+[\tilde{A},\tilde{B}]=0$ 
and, therefore, form a ZCR. (This is explained in Remark~\ref{remzcrgt} below.)

Moreover, since $A$, $B$ take values in~$\mg$ and $G$ 
takes values in~$\mathcal{G}$, 
the functions $\tilde{A}$, $\tilde{B}$ take values in~$\mg$.
(This is well known, but for completeness we prove this in Lemma~\ref{lemgt} 
in Section~\ref{csev}.)

The ZCR~\er{mnprint} is said to be \emph{gauge equivalent} 
to the ZCR~\er{mnoc},~\er{mnzcr}. 
For a given PDE~\er{sys_intr}, formulas~\er{mnprint} determine 
an action of the group of gauge transformations on the set of ZCRs of this PDE.

The WE method does not use gauge transformations in a systematic way. 
In the classification of ZCRs~\er{wecov} this is acceptable, 
because the class of ZCRs~\er{wecov} is relatively small.  

The class of ZCRs~\er{mnoc},~\er{mnzcr} is much larger than that of~\er{wecov}.
As is shown in the present paper, 
gauge transformations play a very important role in the classification 
of ZCRs~\er{mnoc},~\er{mnzcr}. 
Because of this, the classical WE method does not produce satisfactory results 
for~\er{mnoc},~\er{mnzcr}, especially in the case~$\oc>0$.  


To overcome this problem,
we use the approach which we developed in~\cite{prol12,scal13}.
Namely, using some ideas from~\cite{prol12,scal13}, 
we find a normal form for ZCRs~\er{mnoc},~\er{mnzcr}
with respect to the action of the group of gauge transformations.
Using the normal form of ZCRs, for any given equation~\er{sys_intr},
we define a Lie algebra $\fds^{\oc}$ for each $\oc\in\zp$ so  
that the following property holds. 

For every finite-dimensional Lie algebra $\mg$, 
any $\mg$-valued ZCR~\er{mnoc},~\er{mnzcr} of order~$\le\oc$  
is locally gauge equivalent to the ZCR arising from a homomorphism 
$\fds^{\oc}\to\mg$. 

More precisely, as is discussed below, 
we define a Lie algebra $\fds^{\oc}$ for each $\oc\in\zp$ 
and each point~$a$ of the infinite prolongation~$\CE$ of the PDE~\er{sys_intr}. 
So the full notation for the algebra is $\fds^{\oc}(\CE,a)$.  

Recall that the \emph{infinite prolongation} $\CE$ of~\er{sys_intr} 
is the infinite-dimensional manifold with the coordinates 
\beq
\lb{ixtuik}
x,\qquad t,\qquad u^i_k,\qquad i=1,\dots,\nv,\qquad k\in\zp.
\ee 
(A more precise definition of the manifold $\CE$ is given in Section~\ref{subsip}.)
The precise definition of $\fds^{\oc}(\CE,a)$ 
for any evolution PDE~\er{sys_intr} is presented in Section~\ref{csev}. 
In this definition, the algebra $\fds^{\oc}(\CE,a)$ is given in terms of generators and relations.

For every finite-dimensional Lie algebra $\mg$, 
homomorphisms $\fds^{\oc}(\CE,a)\to\mg$ classify (up to gauge equivalence) 
all $\mg$-valued ZCRs~\er{mnoc},~\er{mnzcr} of order~$\le\oc$, 
where functions $A$, $B$ are defined on a neighborhood of 
the point $a\in\CE$. See Section~\ref{csev} for details.

The algebras $\fds^{\oc}(\CE,a)$ are responsible also for parameter-dependent ZCRs, 
see Remark~\ref{rpzcr} below.

For scalar evolution equations with $\nv=1$ this approach 
was developed in~\cite{scal13}.
In the present paper we study the algebras $\fds^{\oc}(\CE,a)$ 
for multicomponent evolution PDEs~\er{sys_intr} with arbitrary $\nv$.

Note that the same Lie algebras $\fds^{\oc}(\CE,a)$ 
were used in~\cite{prol12} for a different purpose, see Remark~\ref{rnzcr} below.

\begin{remark}
\lb{rpzcr}
In the theory of integrable $(1+1)$-dimensional PDEs, 
one is especially interested in ZCRs depending on a parameter $\la$.
So consider a $\mg$-valued ZCR of the form
\begin{gather}
\lb{abla}
A=A(\la,x,t,u^j_0,u^j_1,\dots,u^j_\oc),\qquad\quad 
B=B(\la,x,t,u^j_0,u^j_1,\dots,u^j_{\oc+\eo-1}),\\
\notag 
D_x(B)-D_t(A)+[A,B]=0,
\end{gather}
where $\mg$-valued functions $A$, $B$ depend on $x$, $t$, $u^i_k$ 
and a parameter $\la$.

Let $\tilde{\mg}$ be the infinite-dimensional Lie algebra 
of functions $h(\la)$ with values in $\mg$. 
(Depending on the problem under study, one can consider analytic or meromorphic functions $h(\la)$.
Or one can assume that $\la$ runs through an open subset of some algebraic curve 
and consider $\mg$-valued functions $h(\la)$ on this algebraic curve.)

Then~\er{abla} can be regarded as a ZCR with values in $\tilde{\mg}$.  
After a suitable (parameter-dependent) gauge transformation, 
each ZCR of the form~\er{abla} corresponds to a homomorphism  
$\fds^{\oc}(\CE,a)\to\tilde{\mg}$.
So the Lie algebras $\fds^{\oc}(\CE,a)$ are responsible 
also for parameter-dependent ZCRs.
\end{remark}

Applications of $\fds^{\oc}(\CE,a)$ to the theory of 
B\"acklund transformations are described in~\cite{prol12}.
In the scalar case $\nv=1$, 
applications of $\fds^{\oc}(\CE,a)$ to obtaining 
necessary conditions for integrability of scalar evolution equations 
are discussed in~\cite{scal13}. 

According to Section~\ref{csev}, 
the algebras $\fds^{\oc}(\CE,a)$ for $\oc\in\zp$ are arranged in a sequence of surjective homomorphisms 
\beq
\lb{intfdoc1}
\dots\to\fds^{\oc}(\CE,a)\to
\fds^{\oc-1}(\CE,a)\to\dots\to\fds^1(\CE,a)\to\fds^0(\CE,a).
\ee

Recall that $\fik$ is either $\Com$ or $\mathbb{R}$.
We suppose that the variables $x$, $t$, $u^i_k$ take values in $\fik$. 
A point $a\in\CE$ is determined by the values of the coordinates 
$x$, $t$, $u^i_k$ at $a$. Let
\begin{equation}
\lb{iaxtaik}
a=(x=x_a,\,t=t_a,\,u^i_k=a^i_k)\,\in\,\CE,\qquad\quad 
x_a,t_a,a^i_k\in\fik,\qquad 
i=1,\dots,\nv,\qquad 
k\in\zp,
\end{equation}
be a point of $\CE$.
In other words, the constants $x_a$, $t_a$, $a^i_k$ are the coordinates 
of the point $a\in\CE$ in the coordinate system $x$, $t$, $u^i_k$.
\begin{example}

To clarify the definition of $\fds^{\oc}(\CE,a)$, 
let us consider the case $\nv=\oc=1$. 
To this end, we fix an evolution PDE~\er{sys_intr} with $\nv=1$
and study ZCRs of order~$\le 1$ of this PDE. 

Since we assume $\nv=1$, any ZCR of order~$\le 1$ is written as follows
\beq
\lb{zcru1}
A=A(x,t,u^1_0,u^1_1),\qquad B=B(x,t,u^1_0,u^1_1,\dots,u^1_{\eo}),\qquad
D_x(B)-D_t(A)+[A,B]=0.
\ee

According to Theorem~\ref{thzcro1} in Section~\ref{csev}, any such ZCR~\er{zcru1}
on a neighborhood of $a\in\CE$ is gauge equivalent to a ZCR of the form 
\begin{gather}
\lb{nfzcr}
\tilde{A}=\tilde{A}(x,t,u^1_0,u^1_1),\qquad \tilde{B}=
\tilde{B}(x,t,u^1_0,u^1_1,\dots,u^1_{\eo}),\\
\lb{nfzcreq}
D_x(\tilde{B})-D_t(\tilde{A})+[\tilde{A},\tilde{B}]=0,\\
\lb{nfab}
\frac{\pd\tilde{A}}{\pd u^1_1}(x,t,u^1_0,a^1_1)=0,\qquad 
\tilde{A}(x,t,a^1_0,a^1_1)=0,\qquad\tilde{B}(x_a,t,a^1_0,a^1_1,\dots,a^1_{\eo})=0,
\end{gather}
where $x_a$, $a^1_k$ are the constants determined 
by the point $a\in\CE$ given by~\er{iaxtaik}.

In other words, properties~\er{nfab} 
determine a normal form for ZCRs~\er{zcru1} 
with respect to the action of the group of gauge transformations 
on a neighborhood of $a\in\CE$.

A similar normal form for ZCRs~\er{mnoc},~\er{mnzcr} 
with arbitrary $\nv$, $\oc$ is described 
in Theorem~\ref{evcov} in Section~\ref{csev}.

Since the functions $\tilde{A}$, $\tilde{B}$ from~\er{nfzcr},~\er{nfab} 
are analytic on a neighborhood of $a\in\CE$, these functions
are represented as absolutely convergent power series
\begin{gather}
\label{aser1}
\tilde{A}=\sum_{l_1,l_2,i_0,i_1\ge 0} 
(x-x_a)^{l_1} (t-t_a)^{l_2}(u^1_0-a^1_0)^{i_0}(u^1_1-a^1_1)^{i_1}\cdot
\tilde{A}^{l_1,l_2}_{i_0,i_1},\\
\lb{bser1}
\tilde{B}=\sum_{l_1,l_2,j_0,\dots,j_{\eo}\ge 0} 
(x-x_a)^{l_1} (t-t_a)^{l_2}(u^1_0-a^1_0)^{j_0}\dots(u^1_{\eo}-a^1_{\eo})^{j_{\eo}}\cdot
\tilde{B}^{l_1,l_2}_{j_0\dots j_{\eo}}.
\end{gather}
Here $\tilde{A}^{l_1,l_2}_{i_0,i_1}$ and $\tilde{B}^{l_1,l_2}_{j_0\dots j_{\eo}}$ 
are elements of a Lie algebra, which we do not specify yet. 

Using formulas~\er{aser1},~\er{bser1}, we see that properties~\er{nfab} are equivalent to 
\beq
\lb{ab000int}
\tilde{A}^{l_1,l_2}_{i_0,1}=
\tilde{A}^{l_1,l_2}_{0,0}=
\tilde{B}^{0,l_2}_{0\dots 0}=0
\qquad\qquad\forall\,l_1,l_2,i_0\in\zp.
\ee
To define $\fds^1(\CE,a)$, we regard $\tilde{A}^{l_1,l_2}_{i_0,i_1}$, 
$\tilde{B}^{l_1,l_2}_{j_0\dots j_{\eo}}$ from~\er{aser1},~\er{bser1} 
as abstract symbols.  
By definition, the algebra $\fds^1(\CE,a)$ is generated by the symbols 
$\tilde{A}^{l_1,l_2}_{i_0,i_1}$, $\tilde{B}^{l_1,l_2}_{j_0\dots j_{\eo}}$
for $l_1,l_2,i_0,i_1,j_0,\dots,j_{\eo}\in\zp$.
Relations for these generators are provided by equations~\er{nfzcreq},~\er{ab000int}. 
A more detailed description of this construction 
is given in Section~\ref{csev} for arbitrary $\nv$, $\oc$.
The scalar case $\nv=1$ is studied also in~\cite{scal13}.
\end{example}

\begin{remark}
\lb{rnzcr}

Let $\dc\in\zsp$. Consider a system of the form
\begin{gather}
\lb{wlx}
w^l_x=
\al^l(w^1,\dots,w^\dc,x,t,u^j_0,u^j_1,\dots,u^j_\oc),\\
\lb{wlt}
w^l_t=
\beta^l(w^1,\dots,w^\dc,x,t,u^j_0,u^j_1,\dots,u^j_{\ocn+\eo-1}),\\
\notag
w^l=w^l(x,t),\qquad\qquad l=1,\dots,\dc.
\end{gather}

We assume that system \er{wlx}, \er{wlt} is compatible modulo~\eqref{sys_intr}, 
which means the following. 
Differentiating equation \er{wlx} with respect to $t$ 
and equation \er{wlt} with respect to $x$, one obtains some expressions for 
$w^l_{xt}$ and $w^l_{tx}$.
The expressions for $w^l_{xt}$ and $w^l_{tx}$ must coincide, modulo~\eqref{sys_intr}.

For example, the linear system~\er{auxls} corresponding to the ZCR \er{mnoc}, \er{mnzcr} 
is compatible modulo~\eqref{sys_intr}. 
In general, in \er{wlx}, \er{wlt} the functions $\al^l$, $\beta^l$ may depend on 
$w^1,\dots,w^\dc$ nonlinearly.
It is well known that such compatible systems play an important role in the theory 
of B\"acklund transformations (see, e.g., \cite{backl82}).

Recall that $\CE$ is the manifold with the coordinates \er{ixtuik}.
We assume that the functions $\al^l$, $\beta^l$ from \er{wlx}, \er{wlt}
are defined on the manifold $\mathbf{W}\times\mathbf{U}$, 
where $\mathbf{W}$ is a manifold with the coordinates $w^1,\dots,w^\dc$ 
and $\mathbf{U}$ is an open subset of $\CE$.

Let $\bl$ be the Lie algebra of vector fields on $\mathbf{W}$. 
Since $\mathbf{W}$ is a manifold with the coordinates $w^1,\dots,w^\dc$, 
the Lie algebra $\bl$ consists of vector fields of the form 
$\sum_{l=1}^\dc f^l(w^1,\dots,w^\dc)\dfrac{\pd}{\pd w^l}$.

Then 
\begin{gather}
\lb{vfaw}
\mathcal{A}=\sum_{l=1}^\dc\al^l(w^1,\dots,w^\dc,x,t,u^j_0,u^j_1,\dots,u^j_\oc)
\frac{\partial}{\pd w^l},\\
\lb{vfbw}
\mathcal{B}=\sum_{l=1}^\dc 
\beta^l(w^1,\dots,w^\dc,x,t,u^j_0,u^j_1,\dots,u^j_{\ocn+\eo-1})\frac{\partial}{\pd w^l}
\end{gather}
can be regarded as functions on $\mathbf{U}\subset\CE$ with values in $\bl$.

It is well known (and is explained in~\cite{prol12}) that 
system \er{wlx}, \er{wlt} is compatible modulo~\eqref{sys_intr} iff 
$$
D_x(\mathcal{B})-D_t(\mathcal{A})+[\mathcal{A},\mathcal{B}]=0.
$$
Therefore, \er{vfaw}, \er{vfbw} can be viewed as a ZCR with values in $\bl$.

The preprint~\cite{prol12} shows that, up to gauge equivalence, 
compatible systems of the form \er{wlx}, \er{wlt}
can be described in terms of homomorphisms $\fds^{\oc}(\CE,a)\to\bl$.
(The notion of gauge equivalence for such systems is discussed in~\cite{prol12}.)
\end{remark}

The main goal of this paper is to demonstrate techniques for computation 
of the algebras $\fds^{\oc}(\CE,a)$ for multicomponent evolution PDEs.
Since the algebras $\fds^{\oc}(\CE,a)$ are responsible for all ZCRs, 
computation of $\fds^{\oc}(\CE,a)$ leads to classification of ZCRs up to gauge equivalence.
The results of the paper are described in the next subsection.

\subsection{The main results}
\label{subsmr}

As has been discussed above, for each $\oc\in\zp$ the algebra~$\fds^{\oc}(\CE,a)$ is defined 
by a certain set of generators and relations arising from 
a normal form of ZCRs. 
In Theorem~\ref{lemgenfdq} in Section~\ref{secrgen} 
we describe a smaller subset of generators for~$\fds^{\oc}(\CE,a)$. 
This result is very helpful in the computation of $\fds^{\oc}(\CE,a)$
for concrete PDEs, which is demonstrated in Sections~\ref{secfd1},~\ref{secfdk}.

In Theorem~\ref{thmhfd0} in Section~\ref{fd0we} we describe a relation between
the algebra $\fds^0(\CE,a)$ and the Wahlquist-Estabrook prolongation algebra.
This is helpful in the computation of $\fds^0(\CE,a)$ for PDEs whose 
Wahlquist-Estabrook prolongation algebra is known.

The main example of a PDE considered in this paper 
is the multicomponent generalization of the Landau-Lifshitz equation 
from~\cite{mll,skr-jmp}.  
To describe this PDE, we need some notation.

For any $k\in\zsp$ and any $k$-dimensional vectors 
$V={(v^1,\dots,v^k)}$ and $W={(w^1,\dots,w^k)}$, we set 
$\langle V,W\rangle=\sum_{i=1}^kv^iw^i$.

Recall that $\fik$ is either $\Com$ or $\mathbb{R}$.
Fix an integer $n\ge 2$.
Let $r_1,\dots,r_n\in\fik$ be such that $r_i\neq r_j$ for all $i\neq j$. 
Denote by $R=\mathrm{diag}\,(r_1,\dots,r_n)$ the diagonal $n\times n$ 
matrix with entries $r_i$. Consider the PDE
\begin{equation}
\label{main}
S_t=\Big(S_{xx}+\frac32\langle S_x,S_x\rangle S\Big)_x+\frac32\langle S,RS\rangle S_x,
\qquad\quad \langle S,S\rangle=1,\qquad\quad
R=\mathrm{diag}\,(r_1,\dots,r_n),
\end{equation} 
where $S=\big(s^1(x,t),\dots,s^n(x,t)\big)$ 
is an $n$-dimensional vector-function, and $s^i(x,t)$ take values in $\fik$.

System~\eqref{main} was introduced in~\cite{mll}. 
According to~\cite{mll}, 
for $n=3$ it coincides with the higher symmetry (the commuting flow) 
of third order for the Landau-Lifshitz equation. 
Thus~\eqref{main} can be regarded as an $n$-component generalization 
of the Landau-Lifshitz equation. 
For this reason, we call~\er{main} the $n$-component Landau-Lifshitz system.

The paper~\cite{mll} considers also the following algebraic curve 
\begin{equation}
\label{curve}
\la_i^2-\la_j^2=r_j-r_i,\qquad\qquad i,j=1,\dots,n,
\end{equation}
in the space $\fik^n$ with coordinates $\la_1,\dots,\la_n$.
According to~\cite{mll}, this curve is of genus ${1+(n-3)2^{n-2}}$, 
and system~\eqref{main} possesses a ZCR parametrized by points of this curve.
(The ZCR is described in Remark~\ref{mllzcr} below.) 

System~\eqref{main} has an infinite number of symmetries, 
conservation laws~\cite{mll}, 
and an auto-B\"acklund transformation with a parameter~\cite{ll-backl}. 
Soliton-like solutions of~\eqref{main} can be found in~\cite{ll-backl}. 
In~\cite{skr-jmp} system~\eqref{main}
and its symmetries are constructed by means of the Kostant--Adler scheme.

\begin{remark}
\lb{mllzcr}
For $i,j=1,\dots,n+1$, let 
$E_{i,j}\in\mathfrak{gl}_{n+1}$ be the $(n+1)\times(n+1)$ matrix 
with $(i,j)$-th entry equal to 1 and all other entries equal to zero. 

From the results of~\cite{mll,skr-jmp} one can obtain the following 
$\mathfrak{gl}_{n+1}$-valued ZCR for the PDE~\eqref{main}
\begin{gather}
\label{M}
\CA=\sum_{i=1}^ns^i\la_i(E_{i,n+1}+E_{n+1,i}),\\
\label{N}
\CB=D_x^2(\CA)+[D_x(\CA),\CA]
+\Big(r_1+\la_1^2+\frac12\langle S,RS\rangle
+\frac32\langle D_x(S),D_x(S)\rangle\Big)\CA,\\
\notag
D_x(\CB)-D_t(\CA)+[\CA,\CB]=0.
\end{gather}
Here $\la_1,\dots,\la_n\in\fik$ are parameters satisfying~\eqref{curve}. 

We regard $\la_i(E_{i,n+1}+E_{n+1,i})$ as $\mathfrak{gl}_{n+1}$-valued functions
on the curve~\er{curve}. 
Let $\mathfrak{L}$ be the infinite-dimensional Lie algebra of all polynomial 
$\mathfrak{gl}_{n+1}$-valued functions $M(\la_1,\dots,\la_n)$ on the curve~\er{curve}. 

Let $L(n)\subset\mathfrak{L}$ be the Lie subalgebra
generated by the functions $\la_i(E_{i,n+1}+E_{n+1,i})$, $i=1,\dots,n$.
Using relations~\er{curve}, one can easily show that $L(n)$ 
consists of linear combinations of the functions
\begin{gather}
\lb{lae}
((\la_1)^2+r_1)^l\la_i(E_{i,n+1}+E_{n+1,i}),\qquad\quad
((\la_1)^2+r_1)^l\la_i\la_j(E_{i,j}-E_{j,i}),\\
\notag
i,j=1,\dots,n,\qquad i<j,\qquad 
l\in\zp.
\end{gather}
According to \cite{mll-2012}, the Lie algebra $L(n)$ is infinite-dimensional,
and the functions~\er{lae} form a basis for it. 
We describe $L(n)$ in more detail in Section~\ref{seclnfd0}.

Note that the algebra $L(n)$ is very similar to infinite-dimensional 
Lie algebras that were studied in~\cite{skr,skr-jmp}. 
According to~\cite{mll-2012}, the Lie algebra $L(n)$ appears in the description 
of the Wahlquist-Estabrook prolongation algebra for the PDE~\er{main}.
The paper~\cite{mll-2012} gives also a presentation for the algebra $L(n)$
in terms of a finite number of generators and relations. 

Note that \er{M}, \er{N} can be regarded as functions with values in $L(n)$.
So \er{M}, \er{N} can be viewed as a ZCR with values in $L(n)$.
\end{remark}

\begin{remark}
\lb{son1}
Let $\mathfrak{so}_{n,1}\subset\mathfrak{gl}_{n+1}$ 
be the Lie algebra of the matrix Lie group 
$\mathrm{O}(n,1)\subset\mathrm{GL}_{n+1}$, 
which consists of invertible linear transformations 
that preserve the standard bilinear form of signature $(n,1)$.

The algebra $\mathfrak{so}_{n,1}$ has the following basis
$$
E_{i,j}-E_{j,i},\qquad i<j\le n,\qquad\quad 
E_{l,n+1}+E_{n+1,l},\qquad l=1,\dots,n.
$$
Note that the functions \er{M}, \er{N}, \er{lae} take values in 
$\mathfrak{so}_{n,1}$.
\end{remark}

\begin{remark}
\lb{izcrso}
In order to describe the algebras $\fd^\oc(\CE,a)$ for system~\er{main},
we need to resolve the constraint $\langle S,S\rangle=1$
for the vector-function $S=\big(s^1(x,t),\dots,s^n(x,t)\big)$.
Following~\cite{mll}, we do this as follows
\begin{equation}
\label{insp}
s^j=\frac{2u^j}{1+\langle u,u\rangle},\qquad\qquad
j=1,\dots,n-1,\qquad\qquad
s^n=\frac{1-\langle u,u\rangle}{1+\langle u,u\rangle},
\end{equation}
where $u=\big(u^1(x,t),\dots,u^{n-1}(x,t)\big)$ 
is an $(n-1)$-dimensional vector-function.

Then system~\er{main} can be written as a PDE of the form 
$u^i_t=u^i_{3}+G^i(u^j,u^j_1,u^j_{2})$, $i=1,\dots,n-1$.
The explicit formula for this PDE is \er{pt} in Section~\ref{seclnfd0}.
Substituting~\er{insp} in \er{M}, \er{N}, 
we see that \er{M}, \er{N} is a ZCR of order~$\le 0$ for this PDE.
(That is, the function~\er{M} does not depend on $u^j_l$ for $l>0$.)

In Theorem~\ref{zcrson} in Section~\ref{secszcr} 
we construct for this PDE an $\mathfrak{so}_{n-1}$-valued ZCR of 
the form 
\beq
\lb{fzcrso}
A=A(u^j,u^j_1),\qquad B=B(u^j,u^j_1,u^j_2,u^j_3),\qquad
D_x(B)-D_t(A)+[A,B]=0,
\ee
where $\mathfrak{so}_{n-1}$ is the Lie algebra of skew-symmetric $(n-1)\times(n-1)$ matrices with entries from $\fik$. 
According to our notation, $u^j=u^j_0$, hence the ZCR~\er{fzcrso} is of order~$\le 1$.
Note that the ZCR~\er{fzcrso} does not depend on any parameters.

So we have two very different ZCRs for the same PDE~\er{main},
which can be transformed to the PDE~\er{pt} by the transformation~\er{insp}. 
Namely, we have the $\mathfrak{gl}_{n+1}$-valued ZCR \er{M}, \er{N} 
and the $\mathfrak{so}_{n-1}$-valued ZCR~\er{fzcrso} described in Theorem~\ref{zcrson}.

One can embed the Lie algebras $\gl_{n+1}$ and $\so_{n-1}$ 
into the Lie algebra $\gl_\sm$ for some $\sm\ge n+1$, and then 
one can regard these ZCRs as $\gl_\sm$-valued ZCRs.
One can ask whether these ZCRs can become gauge equivalent 
after suitable embeddings $\gl_{n+1}\hookrightarrow\gl_\sm$ and  
$\so_{n-1}\hookrightarrow\gl_\sm$.
In Remark~\ref{rzcrnge} in Section~\ref{sidyi} 
we show that these ZCRs cannot become gauge equivalent.
\end{remark}

Using the theory described in Sections \ref{csev}, \ref{fd0we}, 
we compute the algebras $\fds^{\oc}(\CE,a)$ for the PDE~\er{main} 
in Sections~\ref{seclnfd0},~\ref{secfd1},~\ref{secfdk}.

The PDE~\er{main} is imposed on a vector-function 
$S=\big(s^1(x,t),\dots,s^n(x,t)\big)$ satisfying $\langle S,S\rangle=1$.
We compute $\fds^{\oc}(\CE,a)$ for this PDE in the case $n\ge 4$.
(The cases $n=2,3$ are less interesting and will described elsewhere.) 

In Section~\ref{seclnfd0} we compute the algebra $\fds^0(\CE,a)$ for this PDE, using 
its Wahlquist-Estabrook prolongation algebra described in~\cite{mll-2012}.

In Section~\ref{secfd1} we compute $\fds^1(\CE,a)$, and in Section~\ref{secfdk}  
the algebras $\fds^k(\CE,a)$ for all $k\ge 2$ are computed for this PDE.
(We describe the structure of these Lie algebras up to some non-essential 
nilpotent ideals.)

The results are summarized in the following theorem, which is proved 
in Section~\ref{subsfd}.

\begin{theorem}[Section~\ref{subsfd}]
\lb{ithfamll}
Let $n\ge 4$. The Lie algebras $\fds^{\oc}(\CE,a)$ for the PDE~\eqref{main}  
have the following structure.

The algebra $\fd^0(\CE,a)$ is isomorphic to 
the algebra $L(n)$ defined in Remark~\ref{mllzcr}.

There is an abelian ideal $\mathfrak{S}$ of $\fd^{1}(\CE,a)$ 
such that $\fd^{1}(\CE,a)/\mathfrak{S}\cong L(n)\oplus\mathfrak{so}_{n-1}$, 
where $\mathfrak{so}_{n-1}$ is the Lie algebra 
of skew-symmetric $(n-1)\times(n-1)$ matrices. 
The surjective homomorphism $\fd^{1}(\CE,a)\to\fd^0(\CE,a)$ from~\er{intfdoc1} coincides 
with the composition of the homomorphisms
\beq
\notag
\fd^{1}(\CE,a)\to\fd^{1}(\CE,a)/\mathfrak{S}\cong L(n)\oplus\mathfrak{so}_{n-1}
\to L(n)\cong \fd^0(\CE,a). 
\ee

Let $\tau_k\cl\fd^{k}(\CE,a)\to\fd^{k-1}(\CE,a)$ 
be the surjective homomorphism from~\er{intfdoc1}.
Then for any $k\ge 2$ we have
\beq
\notag
[v_1,[v_2,v_3]]=0\qquad\quad
\forall\,v_1,v_2,v_3\in\ker\tau_k.
\ee
In particular, the kernel of $\tau_k$ is nilpotent.

For each $k\ge 1$, 
let $\vf_k\colon\fd^k(\CE,a)\,\to\,L(n)\oplus\mathfrak{so}_{n-1}$ 
be the composition of the surjective homomorphisms 
\beq
\notag
\fd^{k}(\CE,a)\to\fd^{1}(\CE,a)\to\fd^{1}(\CE,a)/\mathfrak{S}\cong L(n)\oplus\mathfrak{so}_{n-1},
\ee
where $\fd^k(\CE,a)\to\fd^1(\CE,a)$ arises from~\er{intfdoc1}.
Then 
\beq
\notag
[h_1,[h_2,\dots,[h_{2k-2},[h_{2k-1},h_{2k}]]\dots]]=0\qquad\qquad
\forall\,h_1,h_2,\dots,h_{2k}\in\ker\vf_k.
\ee
In particular, the kernel of $\vf_k$ is nilpotent.
\end{theorem}
\begin{remark}
Nilpotent ideals of the Lie algebras $\fds^{\oc}(\CE,a)$ are not important 
for the main applications to the theory of 
B\"acklund transformations and integrability.

Theorem~\ref{ithfamll} says that, for the PDE~\eqref{main} in the case $n\ge 4$, 
for each $k\ge 1$ there is a surjective homomorphism 
$\vf_k\colon\fd^k(\CE,a)\,\to\,L(n)\oplus\mathfrak{so}_{n-1}$
such that the kernel of $\vf_k$ is nilpotent.
Since nilpotent ideals of $\fd^k(\CE,a)$ are not important, 
one can say that ``the main part'' of $\fd^k(\CE,a)$ is $L(n)\oplus\mathfrak{so}_{n-1}$.

And one has also the isomorphism $\fd^0(\CE,a)\cong L(n)$
for this PDE, according to Theorem~\ref{ithfamll}.
\end{remark}

Consider a ZCR \er{mnoc}, a gauge transformation $G=G(x,t,u^j_0,u^j_1,\dots,u^j_l)$, 
and the corresponding gauge equivalent ZCR \er{mnprint}.
Then, essentially, \er{mnoc} and \er{mnprint} carry the same information, 
because one can switch from \er{mnoc} to \er{mnprint} and vice versa, 
using $G$ and $G^{-1}$.

Furthermore, the following fact about ZCRs with values in a Lie algebra $\mg$ is well known. 
If one is interested in applications to the theory of integrable systems, 
then one can ignore nilpotent ideals of the Lie algebra $\mg$.

Therefore, it makes sense to classify ZCRs
up to gauge equivalence and up to killing nilpotent ideals 
in the corresponding Lie algebras.
As we show below, the algebras $\fds^{\oc}(\CE,a)$ are helpful in this respect.

As has been said above, for every finite-dimensional Lie algebra $\mg$, 
homomorphisms $\fds^{\oc}(\CE,a)\to\mg$ classify (up to gauge equivalence) 
all $\mg$-valued ZCRs~\er{mnoc},~\er{mnzcr} of order~$\le\oc$, 
where functions $A$, $B$ are defined on a neighborhood of the point $a\in\CE$.

According to Theorem~\ref{ithfamll}, 
for the PDE~\er{main} in the case $n\ge 4$, 
we have $\fd^0(\CE,a)\cong L(n)$, and for each $k\ge 1$ there is a surjective 
homomorphism $\vf_k\colon\fd^k(\CE,a)\,\to\,L(n)\oplus\mathfrak{so}_{n-1}$
such that the kernel of $\vf_k$ is nilpotent.

This allows us to classify all ZCRs (up to gauge equivalence and 
up to killing nilpotent ideals) for the PDE~\er{main} in the case $n\ge 4$ as follows.
In Section~\ref{secczcr} we prove that, after suitable gauge transformations 
and after killing some nilpotent ideals, any ZCR becomes isomorphic 
to a reduction of the direct sum of 
the $L(n)$-valued ZCR described in Remark~\ref{mllzcr} 
and the $\mathfrak{so}_{n-1}$-valued ZCR described in Remark~\ref{izcrso}.
(The notions of direct sums and reductions of ZCRs 
are explained in Section~\ref{sbdszcr}.)

In other words, as a result of the classification of 
all ZCRs (depending on derivatives of arbitrary finite order)
for the PDE~\er{main} in the case $n\ge 4$,
we obtain two main non-equivalent ZCRs: the $L(n)$-valued ZCR 
and the $\mathfrak{so}_{n-1}$-valued ZCR described above.
In Section~\ref{secczcr} we prove that 
any other ZCR for the considered PDE is essentially equivalent 
(up to killing nilpotent ideals) to a reduction of the direct sum of these two main ZCRs. 

In our opinion, it is interesting to see that, 
for a given multicomponent evolution PDE, 
one can classify all ZCRs (depending on derivatives of arbitrary order), 
and, as a result of the classification,
one obtains several non-equivalent ZCRs depending on different derivatives 
and taking values in different Lie algebras.
As has been said above, we classify ZCRs 
up to gauge equivalence and up to killing nilpotent ideals 
in the corresponding Lie algebras.

In the present paper we do this for the multicomponent Landau-Lifshitz system, 
but the described computational techniques can be applied to many more evolution PDEs.
(Although for some PDEs the computations may be very difficult.)

For some scalar evolution PDEs of orders $3$, $5$, $7$ 
(including the KdV and Krichever-Novikov equations), 
a similar approach to classification of ZCRs was described in~\cite{cfa,scal13}.
Note that the multicomponent case considered in the present paper is much more 
sophisticated than the scalar case considered in~\cite{cfa,scal13}.


Using the methods developed in Sections~\ref{csev}--\ref{secfdk}, 
in Section~\ref{sllnls} we describe the structure of the Lie algebras 
$\fds^\oc(\CE,a)$, $\oc\in\zp$, for the 
classical Landau-Lifshitz and nonlinear Schr\"odinger equations 
in the case $\fik=\Com$.

As has been said above, 
in Theorem~\ref{thmhfd0} in Section~\ref{fd0we} we describe a relation between
the algebra $\fds^0(\CE,a)$ and the Wahlquist-Estabrook prolongation algebra.
To compute $\fds^0(\CE,a)$ for the classical Landau-Lifshitz 
and nonlinear Schr\"odinger equations in Section~\ref{sllnls},
we use this relation and the corresponding 
Wahlquist-Estabrook prolongation algebras computed in~\cite{ll,we-nls}.

For the classical Landau-Lifshitz equation, we show that $\fds^0(\CE,a)$ 
is isomorphic to an infinite-dimensional Lie algebra of certain 
$\mathfrak{so}_3(\Com)$-valued functions on an elliptic curve.
According to~\cite{ll}, this algebra arises from the well-known elliptic 
$\mathfrak{so}_3(\Com)$-valued ZCR of the Landau-Lifshitz 
equation~\cite{sklyanin,ft,ll}.
For this equation, we show also that the Lie algebras 
$\fd^q(\CE,a)$ with $q\in\zsp$ are obtained from the Lie algebra $\fd^0(\CE,a)$
by applying several times the operation of central extension.

For the nonlinear Schr\"odinger equation, we show that $\fd^0(\CE,a)$
is isomorphic to the direct sum of the infinite-dimensional Lie algebra
$$
\sl_2(\Com[\la])=\sl_2(\Com)\otimes_\Com\Com[\la]
$$
and a one-dimensional abelian Lie algebra. 
Here $\Com[\la]$ is the algebra of polynomials in the variable $\la$.
The Lie algebra $\sl_2(\Com[\la])$ arises from the well-known 
parameter-dependent 
$\sl_2(\Com)$-valued ZCR of the nonlinear Schr\"odinger equation, 
and $\la$ corresponds to the parameter in the ZCR.
For this equation, we show also that the Lie algebras 
$\fd^q(\CE,a)$ with $q\in\zsp$ are obtained from the Lie algebra $\fd^0(\CE,a)$
by applying several times the operation of central extension.

Using the algebras $\fds^{\oc}(\CE,a)$ for $(1+1)$-dimensional evolution PDEs, 
the preprint \cite{prol12} describes a necessary condition for existence 
of B\"acklund transformations between two given PDEs.
This necessary condition is given in \cite{prol12}
in terms of the algebras $\fds^{\oc}(\CE,a)$.

Using this necessary condition from \cite{prol12} 
and knowing the structure of the algebras $\fds^{\oc}(\CE,a)$, 
one can sometimes prove non-existence of 
B\"acklund transformations between two given PDEs.
Examples of such results for scalar evolution equations are given in~\cite{prol12} 
and references therein.
The methods to compute $\fds^{\oc}(\CE,a)$ described in the present paper 
allow one to obtain similar results for some multicomponent PDEs, 
which will be discussed in forthcoming publications.

\begin{remark}
As has been said above, the ZCR~\er{M},~\er{N} for the PDE~\er{main}
is parametrized by points of the curve~\er{curve}.
Some other integrable PDEs with ZCRs parametrized by the 
curve~\eqref{curve} were introduced in~\cite{mll,naukma,skr}. 
It was noticed in~\cite{skr} that the formulas
$\la=\la_i^2+r_i,\,\ y=\prod_{i=1}^n\la_i$ 
provide a map from the curve~\eqref{curve}
to the hyperelliptic curve $y^2=\prod_{i=1}^n(\la-r_i)$.
According to~\cite{mll}, for $n>3$ 
the curve~\eqref{curve} itself is not hyperelliptic.
\end{remark}

\begin{remark}
\lb{remzcrgt}
It is well known that equation~\er{mnzcr} implies 
$D_x(\tilde{B})-D_t(\tilde{A})+[\tilde{A},\tilde{B}]=0$ 
for $\tilde{A}$, $\tilde{B}$ given by~\er{mnprint}.
Indeed, one has $D_x+\tilde{A}=G(D_x+A)G^{-1}$ and $D_x+\tilde{B}=G(D_t+B)G^{-1}$. 
Therefore, 
\begin{multline*}
D_x(\tilde{B})-D_t(\tilde{A})+[\tilde{A},\tilde{B}]=
[D_x+\tilde{A},D_t+\tilde{B}]=[G(D_x+A)G^{-1},G(D_t+B)G^{-1}]=\\
=G[D_x+A,D_t+B]G^{-1}=G(D_x(B)-D_t(A)+[A,B])G^{-1}.
\end{multline*}
Hence the equation $D_x(B)-D_t(A)+[A,B]=0$ implies 
$D_x(\tilde{B})-D_t(\tilde{A})+[\tilde{A},\tilde{B}]=0$.
\end{remark}

\begin{remark}
Some other approaches to the study of 
the action of gauge transformations on ZCRs can be found 
in~\cite{marvan93,marvan97,marvan2010,sakov95,sakov2004,sebest2008} 
and references therein.
For a given ZCR with values in a matrix Lie algebra $\mg$, 
the papers~\cite{marvan93,marvan97,sakov95} define 
certain $\mg$-valued functions that transform by conjugation 
when the ZCR transforms by gauge. 
Applications of these functions to construction and classification of 
some types of ZCRs are described  
in~\cite{marvan93,marvan97,marvan2010,sakov95,sakov2004,sebest2008}.

To our knowledge, 
the theory of~\cite{marvan93,marvan97,marvan2010,sakov95,sakov2004,sebest2008} 
does not produce any infinite-dimensional Lie algebras responsible for ZCRs. 
So this theory does not contain the algebras $\fds^\oc(\CE,a)$.
\end{remark}

\subsection{Abbreviations, conventions, and notation}
\lb{subs-conv}

The following abbreviations, conventions, and notation are used in the paper.

ZCR = zero-curvature representation, WE = Wahlquist-Estabrook.

The symbols $\zsp$ and $\zp$ denote the sets of positive and nonnegative 
integers respectively.

$\fik$ is either $\Com$ or $\mathbb{R}$.
All vector spaces and algebras are supposed to be over the field~$\fik$.
We denote by $\gl_\sm$ the algebra of 
$\sm\times\sm$ matrices with entries from $\fik$
and by $\mathrm{GL}_\sm$ the group of invertible $\sm\times\sm$ matrices.

Also, we use the following notation for partial derivatives of 
functions $u^i=u^i(x,t)$, $i=1,\dots,\nv$,
$$
u^i_0=u^i,\qquad
u^i_k=\frac{\pd^k u^i}{\pd x^k},\qquad\quad k\in\zp.
$$

\section{Zero-curvature representations, 
gauge transformations, and the algebras $\fds^\oc(\CE,a)$}
\lb{csev}

In this section we study the algebras $\fds^\oc(\CE,a)$ introduced in~\cite{prol12}.
For completeness, we give detailed definitions of $\CE$, $a\in\CE$, and 
$\fds^\oc(\CE,a)$.

\subsection{The infinite prolongation of an evolution PDE}
\lb{subsip}

As has been said in Section~\ref{secintr}, 
we suppose that $x$, $t$, $u^i_k$ take values in $\fik$, 
where $\fik$ is either $\Com$ or $\mathbb{R}$.
Let $\fik^\infty$ be the infinite-dimensional space  
with the coordinates 
\beq
\lb{xtuik}
x,\qquad t,\qquad u^i_k,\qquad i=1,\dots,\nv,\qquad k\in\zp.
\ee 
The topology on $\fik^\infty$ is defined as follows. 

For each $l\in\zp$, consider the space $\fik^{\nv(l+1)+2}$ 
with the coordinates $x$, $t$, $u^i_k$ for $k=0,1,\dots,l$ and $i=1,\dots,\nv$. 
One has the natural projection $\pi_l\cl\fik^\infty\to\fik^{\nv(l+1)+2}$ 
that ``forgets'' the coordinates $u^i_{k'}$ for $k'>l$. 

Since $\fik^{\nv(l+1)+2}$ is a finite-dimensional vector space,
we have the standard topology on~$\fik^{\nv(l+1)+2}$. 
For any $l\in\zp$ and any open subset $V\subset\fik^{\nv(l+1)+2}$, 
the subset~$\pi_l^{-1}(V)\subset\fik^\infty$ is, by definition, open in $\fik^\infty$. Such subsets form a base of the topology on~$\fik^\infty$. 
In other words, we consider the smallest topology on~$\fik^\infty$ such that 
the maps $\pi_l\cl\fik^\infty\to\fik^{\nv(l+1)+2}$, $l\in\zp$, are continuous. 

According to our notation, the PDE~\er{sys_intr} can be written as~\er{uitfi}.
Let $\fik^{\nv(\eo+1)+2}$ be the space 
with the coordinates $x$, $t$, $u^i_k$ for $k=0,1,\dots,\eo$ and $i=1,\dots,\nv$.
Let $\ost\subset\fik^{\nv(\eo+1)+2}$ be an open subset such that the functions 
$F^i(x,t,u^j_0,u^j_1,\dots,u^j_{\eo})$ from~\er{uitfi} are defined on~$\ost$. 

The \emph{infinite prolongation} $\CE$ of the evolution PDE~\er{sys_intr} 
can be defined as follows 
$$
\CE=\pi_{\eo}^{-1}(\ost)\subset\fik^\infty.
$$
So $\CE$ is an open subset of the space $\fik^\infty$  
with the coordinates~\er{xtuik}. 
The topology on~$\CE$ is induced by the embedding $\CE\subset\fik^\infty$. 

\begin{example}
For any constants $e_1,e_2,e_3\in\fik$, 
consider the Krichever-Novikov equation~\cite{krich80,svin-sok83} 
\begin{equation}
\label{knedef}
u_t=u_{xxx}-\frac32\frac{(u_{xx})^2}{u_x}+
\frac{(u-e_1)(u-e_2)(u-e_3)}{u_x},\quad\qquad u=u(x,t).
\end{equation}
Since this is a scalar equation of order $3$, we have here 
$\nv=1$ and $\eo=3$.

In our notation, we set $u^1=u(x,t)$ and rewrite equation~\er{knedef} as follows 
\begin{gather}
\label{u1kn}
u^1_t=F^1(x,t,u^1_0,u^1_1,u^1_2,u^1_3),\\
\lb{f1kn}
F^1(x,t,u^1_0,u^1_1,u^1_2,u^1_3)=
u^1_3-\frac32\frac{(u^1_2)^2}{u^1_1}+
\frac{(u^1_0-e_1)(u^1_0-e_2)(u^1_0-e_3)}{u^1_1},
\end{gather}
where $u^1_0=u^1$ and $u^1_k=\dfrac{\pd^k u^1}{\pd x^k}$ for $k\in\zp$.

Let $\fik^6$ be the space with the coordinates 
$x$, $t$, $u^1_0$, $u^1_1$, $u^1_2$, $u^1_3$.
According to~\er{f1kn}, 
the function $F^1$ is defined on the open subset $\ost\subset\fik^6$ determined 
by the condition $u^1_1\neq 0$. 

In this example, $\fik^\infty$ is the space with the coordinates
$x$, $t$, $u^1_k$ for $k\in\zp$. 
We have the map $\pi_3\cl\fik^\infty\to\fik^6$
that ``forgets'' the coordinates $u^1_{k'}$ for $k'>3$. 
The infinite prolongation $\CE$ of equation~\er{u1kn} is the following open subset 
of $\fik^\infty$
$$
\CE=\pi_3^{-1}(\ost)=
\big\{(x,t,u^1_0,u^1_1,u^1_2,\dots)\in\fik^\infty\,\big|\,
u^1_1\neq 0\big\}.
$$
\end{example}

\subsection{A normal form for ZCRs with respect to the action of gauge transformations}

Consider again an evolution PDE~\er{sys_intr} with arbitrary $\nv,\eo\in\zsp$.
As has been said above, the infinite prolongation $\CE$ 
of~\er{sys_intr} is an open subset of the space $\fik^\infty$  
with the coordinates~\er{xtuik}. 

A point $a\in\CE$ is determined by the values of the coordinates 
$x$, $t$, $u^i_k$ at $a$. Let
\begin{equation}
\lb{axtaik}
a=(x=x_a,\,t=t_a,\,u^i_k=a^i_k)\,\in\,\CE,\qquad\quad 
x_a,t_a,a^i_k\in\fik,\qquad 
i=1,\dots,\nv,\qquad 
k\in\zp,
\end{equation}
be a point of $\CE$.
In other words, the constants $x_a$, $t_a$, $a^i_k$ are the coordinates 
of the point $a\in\CE$ in the coordinate system $x$, $t$, $u^i_k$.


Recall that, for every $\sm\in\zsp$,
we denote by $\gl_\sm$ the algebra of 
$\sm\times\sm$ matrices with entries from $\fik$ 
and by $\mathrm{GL}_\sm$ the group of invertible $\sm\times\sm$ matrices.
Let $\mathrm{Id}\in\mathrm{GL}_\sm$ be the identity matrix.

By the standard Lie group -- Lie algebra correspondence, 
for every Lie subalgebra $\mg\subset\gl_\sm$ 
there is a unique connected immersed Lie subgroup 
$\mathcal{G}\subset\mathrm{GL}_\sm$ whose Lie algebra is $\mg$.
We call $\mathcal{G}$ the connected matrix Lie group 
corresponding to the matrix Lie algebra $\mg\subset\gl_\sm$.

For any $l\in\zp$, a matrix-function $G=G(x,t,u^j_0,u^j_1,\dots,u^j_l)$ 
with values in~$\mathcal{G}$ is called a gauge transformation.
Equivalently, one can say that a gauge transformation is given 
by a $\mathcal{G}$-valued function $G=G(x,t,u^j_0,u^j_1,\dots,u^j_l)$.

The following lemma is known, but for completeness we present a proof of it.
\begin{lemma}
\lb{lemgt}
Let $\sm\in\zsp$ and $\oc\in\zp$. 
Let $\mg\subset\gl_\sm$ be a matrix Lie algebra and  
$\mathcal{G}\subset\mathrm{GL}_\sm$ be 
the connected matrix Lie group corresponding to $\mg\subset\gl_\sm$. 

Let 
\beq
\lb{lemab}
A=A(x,t,u^j_0,u^j_1,\dots,u^j_\oc),\qquad 
B=B(x,t,u^j_0,u^j_1,\dots,u^j_{\oc+\eo-1}),\qquad
D_x(B)-D_t(A)+[A,B]=0
\ee
be a ZCR of order $\le\oc$ such that 
the functions $A$, $B$ take values in $\mg$.
Here $D_x$ and $D_t$ are given by~\er{evdxdt}.

Then for any $\mathcal{G}$-valued function 
\beq
\lb{ggocs1}
G=G(x,t,u^j_0,u^j_1,\dots,u^j_{\oc-1})
\ee
depending on $x$, $t$, $u^j_0,\dots,u^j_{\oc-1}$, $j=1,\dots,\nv$, the functions 
\beq
\lb{lemtab}
\tilde{A}=GAG^{-1}-D_x(G)\cdot G^{-1},\qquad\qquad
\tilde{B}=GBG^{-1}-D_t(G)\cdot G^{-1}
\ee
form a $\mg$-valued ZCR of order $\le\oc$. 
That is, 
\beq
\lb{lemtzcr}
\tilde{A}=\tilde{A}(x,t,u^j_0,u^j_1,\dots,u^j_\oc),\quad 
\tilde{B}=\tilde{B}(x,t,u^j_0,u^j_1,\dots,u^j_{\oc+\eo-1}),\quad
D_x(\tilde{B})-D_t(\tilde{A})+[\tilde{A},\tilde{B}]=0,
\ee
and $\tilde{A}$, $\tilde{B}$ take values in $\mg$.

Formulas~\er{lemtab} determine an action of the group 
of $\mathcal{G}$-valued gauge transformations~\er{ggocs1}
on the set of $\mg$-valued ZCRs of order~$\le\oc$.
\end{lemma}
\begin{proof}
Since $A$, $B$ take values in $\mg$ and $G$ takes values in 
the connected matrix Lie group $\mathcal{G}$ 
corresponding to the Lie algebra $\mg\subset\gl_\sm$,
the functions 
\beq
\lb{gagb}
GAG^{-1},\quad GBG^{-1},\quad\frac{\pd}{\pd x}(G)\cdot G^{-1},\quad
\frac{\pd}{\pd t}(G)\cdot G^{-1},\quad\frac{\pd}{\pd u^i_k}(G)\cdot G^{-1}\qquad
\forall\,i,k
\ee
take values in $\mg$. 
Hence the functions $\tilde{A}$, $\tilde{B}$ given by~\er{lemtab} 
take values in $\mg$ as well.

Using formulas~\er{evdxdt}, \er{lemab}, \er{lemtab} 
and the fact that $G$ may depend only on $x$, $t$, $u^j_0,\dots,u^j_{\oc-1}$, 
we easily get~\er{lemtzcr}. 

One has $D_x+\tilde{A}=G(D_x+A)G^{-1}$ and $D_x+\tilde{B}=G(D_t+B)G^{-1}$, 
which implies that formulas~\er{lemtab} determine an action of the group 
of $\mathcal{G}$-valued gauge transformations~\er{ggocs1}
on the set of $\mg$-valued ZCRs of order~$\le\oc$.
\end{proof}

\begin{remark}
\lb{anmer}
For any $l\in\zp$ and $a\in\CE$, 
when we consider a function $Q=Q(x,t,u^j_0,u^j_1,\dots,u^j_l)$ 
defined on a neighborhood of $a\in\CE$, 
we always assume that the function is analytic on this neighborhood.
For example, $Q$ may be a meromorphic function defined on an open subset of~$\CE$
such that $Q$ is analytic on a neighborhood of $a\in\CE$.

In particular, this applies to the functions 
$A$, $B$, $G$, $\tilde{A}$, $\tilde{B}$ 
considered in Theorems~\ref{thzcro1},~\ref{evcov} below.
\end{remark}

Theorems~\ref{thzcro1} and~\ref{evcov} below describe 
a normal form for ZCRs with respect to the action of the group of gauge transformations.
To clarify the construction, we first consider the case of ZCRs of order~$\le 1$ 
in Theorem~\ref{thzcro1}.
The general case of ZCRs of order~$\le\oc$ for any $\oc\in\zp$ is described in 
Theorem~\ref{evcov}.

\begin{theorem}
\lb{thzcro1}
Let $\nv,\sm\in\zsp$. 
Let $\mg\subset\gl_\sm$ be a matrix Lie algebra. 
Denote by $\mathcal{G}\subset\mathrm{GL}_\sm$ the connected matrix Lie group 
corresponding to $\mg\subset\gl_\sm$. 

Let $\CE$ be the infinite prolongation 
of an $\nv$-component evolution PDE~\er{sys_intr}.
Consider a point $a\in\CE$ given by~\er{axtaik}. 
According to~\er{axtaik}, the point $a$ is determined by constants $x_a$, $t_a$, $a^i_k$.

Let 
\beq
\lb{thmnoc}
A=A(x,t,u^j_0,u^j_1),\qquad B=B(x,t,u^j_0,u^j_1,\dots,u^j_{\eo}),\qquad
D_x(B)-D_t(A)+[A,B]=0
\ee
be a ZCR of order $\le 1$ such that 
the functions $A$, $B$ are defined on a neighborhood of $a\in\CE$ and take values in $\mg$.

Then, on a neighborhood of $a\in\CE$, there is a $\mathcal{G}$-valued function 
\beq
\lb{ggxtuj0}
G=G(x,t,u^j_0)
\ee
depending on $x$, $t$, $u^j_0$, $j=1,\dots,\nv$, such that the functions 
\beq
\lb{mnprth}
\tilde{A}=GAG^{-1}-D_x(G)\cdot G^{-1},\qquad\qquad
\tilde{B}=GBG^{-1}-D_t(G)\cdot G^{-1}
\ee
satisfy
\begin{gather}
\label{d=0o1}
\forall\,i_0=1,\dots,\nv,\qquad\quad
\frac{\pd \tilde{A}}{\pd u^{i_0}_1}
\,\,\bigg|_{u^j_1=a^j_1\ \forall\,j,\ u^i_0=a^i_0\ \forall\,i>i_0}=0,\\
\lb{aukako1}
\tilde{A}\,\Big|_{u^j_0=a^j_0,\ u^j_1=a^j_1\ \forall\,j}=0,\\
\lb{bxx0o1}
\tilde{B}\,\Big|_{x=x_a,\ u^j_k=a^j_k\ \forall\,j,\ \forall\,k\ge 0}=0,
\end{gather}
and
\beq
\lb{gxt0ua}
G\,\Big|_{x=x_a,\ t=t_a,\ u^j_0=a^j_0\ \forall\,j}=\mathrm{Id}.
\ee

Note that, according to Lemma~\ref{lemgt}, the functions~\er{mnprth}
form a $\mg$-valued ZCR of order $\le 1$. That is, 
\beq
\lb{tthmnoc}
\tilde{A}=\tilde{A}(x,t,u^j_0,u^j_1),\qquad 
\tilde{B}=\tilde{B}(x,t,u^j_0,u^j_1,\dots,u^j_{\eo}),\qquad
D_x(\tilde{B})-D_t(\tilde{A})+[\tilde{A},\tilde{B}]=0,
\ee
and $\tilde{A}$, $\tilde{B}$ take values in $\mg$.
\end{theorem}
\begin{remark}
The notation $\tilde{B}\,\Big|_{x=x_a,\ u^j_k=a^j_k\ \forall\,j,\ \forall\,k\ge 0}$ 
in~\er{bxx0o1} 
means that we substitute $x=x_a$ and $u^j_k=a^j_k$ for all $j=1,\dots,\nv$ and 
all $k\ge 0$ in the function 
$\tilde{B}=\tilde{B}(x,t,u^j_0,u^j_1,\dots,u^j_{\eo})$. That is, 
$$
\tilde{B}\,\Big|_{x=x_a,\ u^j_k=a^j_k\ \forall\,j,\ \forall\,k\ge 0}=
\tilde{B}(x_a,t,a^j_0,a^j_1,\dots,a^j_{\eo}).
$$
The notation in \er{d=0o1}, \er{aukako1}, \er{gxt0ua} can be understood in a similar way.
\end{remark}
\begin{proof}
To clarify the main idea,
let us consider first the case $\nv=2$. Then formulas \er{thmnoc}, \er{tthmnoc} become
\begin{gather}
A=A(x,t,u^1_0,u^2_0,u^1_1,u^2_1),\quad\qquad 
B=B(x,t,u^1_0,u^2_0,u^1_1,u^2_1,\dots,u^1_{\eo},u^2_{\eo}),\\
\tilde{A}=\tilde{A}(x,t,u^1_0,u^2_0,u^1_1,u^2_1),\quad\qquad 
\tilde{B}=\tilde{B}(x,t,u^1_0,u^2_0,u^1_1,u^2_1,\dots,u^1_{\eo},u^2_{\eo}).
\end{gather}
For $\nv=2$ in~\er{d=0o1} we have $i_0=1,2$, so condition~\er{d=0o1} is equivalent 
to the following two equations
\begin{gather}
\lb{tau11}
\frac{\pd\tilde{A}}{\pd u^1_1}(x,t,u^1_0,a^2_0,a^1_1,a^2_1)=0,\\
\lb{tau21}
\frac{\pd\tilde{A}}{\pd u^2_1}(x,t,u^1_0,u^2_0,a^1_1,a^2_1)=0.
\end{gather}
Conditions~\er{aukako1},~\er{bxx0o1} in the case $\nv=2$ can be written as
\begin{gather}
\lb{axta}
\tilde{A}(x,t,a^1_0,a^2_0,a^1_1,a^2_1)=0,\\
\lb{bx0ta}
\tilde{B}(x_a,t,a^1_0,a^2_0,a^1_1,a^2_1,\dots,a^1_{\eo},a^2_{\eo})=0.
\end{gather}

According to Lemma~\ref{lemgt} in the case $\oc=1$, 
formulas~\er{mnprth} determine an action of the group 
of $\mathcal{G}$-valued gauge transformations~\er{ggxtuj0}
on the set of $\mg$-valued ZCRs of order~$\le 1$.

To prove the statement of the theorem in the case $\nv=2$, 
we need to find a $\mathcal{G}$-valued gauge transformation $G=G(x,t,u^1_0,u^2_0)$ 
such that the transformed ZCR~\er{mnprth} 
satisfies \er{tau11}, \er{tau21}, \er{axta}, \er{bx0ta}, 
and $G(x_a,t_a,a^1_0,a^2_0)=\mathrm{Id}$.

We are going to construct the required gauge transformation in several steps.
First, we will construct a transformation to achieve property~\er{tau21},
then another transformation to get properties~\er{tau11},~\er{tau21},
then another transformation to get properties \er{tau11}, \er{tau21}, \er{axta},
and finally another transformation to obtain 
all properties \er{tau11}, \er{tau21}, \er{axta}, \er{bx0ta}.

Consider the ordinary differential equation (ODE)
\beq
\lb{pdgquq}
\frac{\pd G_2}{\pd u^2_0}=G_2\cdot
\bigg(\frac{\pd A}{\pd u^2_1}(x,t,u^1_0,u^2_0,a^1_1,a^2_1)\bigg)
\ee
with respect to the variable $u^2_0$ and an unknown function   
$G_2=G_2(x,t,u^1_0,u^2_0)$. The variables $x,\,t,\,u^1_0$ 
are regarded as parameters in this ODE. 

Let $G_2(x,t,u^1_0,u^2_0)$ be a local 
solution of the ODE~\er{pdgquq} with the initial condition
$G_2(x,t,u^1_0,a^2_0)=\mathrm{Id}$. 
Since ${\pd A}/{\pd u^2_1}$ takes values in $\mg$, 
the function $G_2$ takes values in~$\mathcal{G}$.

Set 
\beq
\lb{a11}
\hat{A}=G_2AG_2^{-1}-D_x(G_2)\cdot G_2^{-1},\qquad\qquad
\hat{B}=G_2BG_2^{-1}-D_t(G_2)\cdot G_2^{-1}.
\ee
Since $G_2$ takes values in $\mathcal{G}$, 
the functions $\hat{A}$, $\hat{B}$ take values in $\mg$.
Using~\er{a11} and~\er{pdgquq}, we get  
\begin{multline}
\lb{multgag}
\frac{\pd\hat{A}}{\pd u^2_1}(x,t,u^1_0,u^2_0,a^1_1,a^2_1)=
G_2\bigg(
\frac{\pd A}{\pd u^2_1}(x,t,u^1_0,u^2_0,a^1_1,a^2_1)\bigg)G_2^{-1}
-\bigg(
\frac{\pd}{\pd u^2_1}\big(D_x(G_2)\big)\bigg)G_2^{-1}=\\
=G_2\bigg(
\frac{\pd A}{\pd u^2_1}(x,t,u^1_0,u^2_0,a^1_1,a^2_1)\bigg)G_2^{-1}
-\frac{\pd G_2}{\pd u^2_0}G_2^{-1}=\\
=G_2\bigg(\frac{\pd A}{\pd u^2_1}(x,t,u^1_0,u^2_0,a^1_1,a^2_1)\bigg)G_2^{-1}
-G_2\bigg(\frac{\pd A}{\pd u^2_1}(x,t,u^1_0,u^2_0,a^1_1,a^2_1)\bigg)G_2^{-1}=0.
\end{multline}

Now consider the ODE 
\beq
\lb{pdgquq0}
\frac{\pd G_1}{\pd u^1_0}=G_1\cdot
\bigg(
\frac{\pd\hat{A}}{\pd u^1_1}(x,t,u^1_0,a^2_0,a^1_1,a^2_1)\bigg)
\ee
with respect to the variable $u^1_{0}$ and an unknown function   
$G_1=G_1(x,t,u^1_0)$, where $x,\,t$ are regarded as parameters. 

Let $G_1(x,t,u^1_0)$ be a local 
solution of the ODE~\er{pdgquq0} with the initial condition
$G_1(x,t,a^1_0)=\mathrm{Id}$. 
Since ${\pd\hat{A}}/{\pd u^1_1}$ takes values in $\mg$, 
the function $G_1$ takes values in~$\mathcal{G}$.

Set 
\beq
\lb{a22}
\bar{A}=G_1\hat{A}G_1^{-1}-D_x(G_1)\cdot G_1^{-1},\qquad\qquad
\bar{B}=G_1\hat{B}G_1^{-1}-D_t(G_1)\cdot G_1^{-1}.
\ee
Then~\er{multgag},~\er{pdgquq0},~\er{a22} imply that $\bar{A}$ satisfies 
properties \er{tau11}, \er{tau21}, if we replace $\tilde{A}$ by $\bar{A}$ 
in \er{tau11}, \er{tau21}.
Furthermore, since $G_1$ takes values in~$\mathcal{G}$, 
the functions $\bar{A}$, $\bar{B}$ take values in~$\mg$.

Let $\tilde G=\tilde G(x,t)$ be a local solution of the ODE 
\beq
\notag
\frac{\pd \tilde G}{\pd x}=\tilde G\cdot
\bar{A}(x,t,a^1_0,a^2_0,a^1_1,a^2_1)
\ee
with the initial condition $\tilde G(x_a,t)=\mathrm{Id}$, 
where $t$ is viewed as a parameter. 
Set 
\beq
\lb{a33}
\check{A}=\tilde{G}\bar{A}\tilde{G}^{-1}-D_x(\tilde{G})\cdot \tilde{G}^{-1},
\qquad\qquad
\check{B}=\tilde{G}\bar{B}\tilde{G}^{-1}-D_t(\tilde{G})\cdot \tilde{G}^{-1}.
\ee
Then $\check{A}$ satisfies 
properties \er{tau11}, \er{tau21}, \er{axta}, 
if we replace $\tilde{A}$ by $\check{A}$ in \er{tau11}, \er{tau21}, \er{axta}.

Finally, let $\hat G=\hat G(t)$ be a local solution of the ODE 
\beq
\lb{pdhatgt}
\frac{\pd \hat G}{\pd t}=\hat G\cdot
\check{B}(x_a,t,a^1_0,a^2_0,a^1_1,a^2_1,\dots,a^1_{\eo},a^2_{\eo})
\ee
with the initial condition $\hat G(t_a)=\mathrm{Id}$.
Set 
\beq
\lb{a44}
\tilde{A}=\hat{G}\check{A}\hat{G}^{-1}-D_x(\hat{G})\cdot \hat{G}^{-1},
\qquad\qquad
\tilde{B}=\hat{G}\check{B}\hat{G}^{-1}-D_t(\hat{G})\cdot \hat{G}^{-1}.
\ee
Then $\tilde{A}$, $\tilde{B}$ obey
\er{tau11}, \er{tau21}, \er{axta}, \er{bx0ta}. 

Let $G=\hat G\cdot\tilde G\cdot G_1\cdot G_2$. Then 
equations~\er{a11}, \er{a22}, \er{a33}, \er{a44} imply 
\beq
\notag
\tilde{A}=GAG^{-1}-D_x(G)\cdot G^{-1},\qquad\quad
\tilde{B}=GBG^{-1}-D_t(G)\cdot G^{-1}.
\ee
Furthermore, 
since $G_2(x,t,u^1_0,a^2_0)=G_1(x,t,a^1_0)=\tilde G(x_a,t)=\hat G(t_a)=\mathrm{Id}$, 
we have $G(x_a,t_a,a^1_0,a^2_0)=\mathrm{Id}$.
Thus $G=\hat G\cdot\tilde G\cdot G_1\cdot G_2$ satisfies 
all the required properties in the case $\nv=2$. 

This construction can be easily generalized to the case of arbitrary~$\nv$. 
One can define $G$ as the product 
$G=\hat G\cdot\tilde G\cdot G_1\cdot G_2\dots G_{\nv}$, 
where the $\mathcal{G}$-valued functions
\begin{gather*}
G_q=G_q(x,t,u^1_0,\dots,u^q_0),\qquad q=1,\dots,\nv,\qquad\quad 
\tilde G=\tilde G(x,t),\qquad\quad\hat G=\hat G(t)
\end{gather*}
are defined as solutions of certain ODEs similar to the ODEs considered above. 
\end{proof}

The set $\{1,\dots,\nv\}\times\zp$ consists of pairs
$(i,k)$, where $i\in\{1,\dots,\nv\}$ and $k\in\zp$.
Consider the following ordering $\preceq$ of the set 
$\{1,\dots,\nv\}\times\zp$ 
\begin{gather}
\notag
i,i'\in\{1,\dots,\nv\},\qquad\quad k,k'\in\zp,\qquad\quad k\neq k',\\
\lb{ev_ord}
(i,k)\prec(i',k')\ \text{ iff }\ k<k',\qquad\qquad 
(i,k)\prec(i',k)\ \text{ iff }\ i<i'.
\end{gather}
That is, $(1,0)\prec(2,0)\prec\dots\prec(\nv,0)\prec(1,1)\prec(2,1)\prec\dots$.

As usual, the notation $(i_1,k_1)\succeq(i_2,k_2)$ means 
that either $(i_1,k_1)\succ(i_2,k_2)$ or $(i_1,k_1)=(i_2,k_2)$. 

\begin{remark}
Let $F=F(x,t,u^i_k)$ be a function of the variables $x$, $t$, $u^i_k$. 
Let $i'\in\{1,\dots,\nv\}$ and $k'\in\zp$. 
Then the notation 
$F\,\Big|_{u^i_k=a^i_k\ \forall\,(i,k)\succ(i',k')}$
says that 
we substitute $u^i_k=a^i_k$ for all $(i,k)\succ(i',k')$ in the function $F$. 

Similarly, the notation 
$F\,\Big|_{x=x_a,\ u^i_k=a^i_k\ \forall\,(i,k)\succeq(i',k')}$
means that we substitute $x=x_a$ and $u^i_k=a^i_k$ for all $(i,k)\succeq(i',k')$ in $F$. 
\end{remark}

\begin{theorem}
\lb{evcov}
Let $\nv,\sm\in\zsp$ and $\oc\in\zp$. 
Let $\mg\subset\gl_\sm$ be a matrix Lie algebra. 
Denote by $\mathcal{G}\subset\mathrm{GL}_\sm$ the connected matrix Lie group 
corresponding to $\mg\subset\gl_\sm$. 

Let $\CE$ be the infinite prolongation of 
an $\nv$-component evolution PDE~\er{sys_intr}.
Consider a point $a\in\CE$ given by~\er{axtaik}. 
According to~\er{axtaik}, the point $a$ is determined by constants $x_a$, $t_a$, $a^i_k$.

Let 
\beq
\lb{thzcroc}
A=A(x,t,u^j_0,u^j_1,\dots,u^j_\oc),\qquad 
B=B(x,t,u^j_0,u^j_1,\dots,u^j_{\oc+\eo-1}),\qquad
D_x(B)-D_t(A)+[A,B]=0
\ee
be a ZCR of order $\le\oc$ such that 
the functions $A$, $B$ are defined on a neighborhood of $a\in\CE$ and take values in $\mg$.

Then, on a neighborhood of $a\in\CE$, 
there is a $\mathcal{G}$-valued function $G=G(x,t,u^j_0,u^j_1,\dots,u^j_{\oc-1})$ 
depending on $x$, $t$, $u^j_0,\dots,u^j_{\oc-1}$, $j=1,\dots,\nv$, 
such that the functions 
\beq
\notag
\tilde{A}=GAG^{-1}-D_x(G)\cdot G^{-1},\qquad\qquad
\tilde{B}=GBG^{-1}-D_t(G)\cdot G^{-1}
\ee
satisfy
\begin{gather}
\label{gd=0}
\forall\,i_0=1,\dots,\nv,\qquad\forall\,k_0\ge 1,\qquad\quad
\frac{\pd\tilde{A}}{\pd u^{i_0}_{k_0}}
\,\,\bigg|_{u^i_k=a^i_k\ \forall\,(i,k)\succ(i_0,k_0-1)}=0,\\
\lb{gaukak}
 \tilde{A}\,\Big|_{u^i_k=a^i_k\ \forall\,(i,k)}=0,\\
\lb{gbxx0}
\tilde{B}\,\Big|_{x=x_a,\ u^i_k=a^i_k\ \forall\,(i,k)}=0,
\end{gather}
and
\beq
\lb{pgxt0ua}
G\,\Big|_{x=x_a,\ t=t_a,\ u^i_k=a^i_k\ \forall\,(i,k)}=\mathrm{Id}.
\ee

Note that, according to Lemma~\ref{lemgt}, the functions~\er{mnprth}
form a $\mg$-valued ZCR of order~$\le\oc$.
That is, $\tilde{A}$, $\tilde{B}$ take values in $\mg$ and satisfy~\er{lemtzcr}.
\end{theorem}
\begin{proof}
This theorem can be proved similarly to Theorem~\ref{thzcro1}.
One can define $G$ as the product of several gauge transformations, which 
are defined as solutions of certain ODEs similar to the ODEs considered 
in the proof of Theorem~\ref{thzcro1}. 
\end{proof}

Fix a point $a\in\CE$ given by~\er{axtaik}, 
which is determined by constants $x_a$, $t_a$, $a^i_k$.

A ZCR 
\beq
\lb{anzcr}
\anA=\anA(x,t,u^j_0,u^j_1,\dots,u^j_\oc),\qquad 
\anB=\anB(x,t,u^j_0,u^j_1,\dots,u^j_{\oc+\eo-1}),\qquad
D_x(\anB)-D_t(\anA)+[\anA,\anB]=0
\ee
is said to be \emph{$a$-normal} if $\anA$, $\anB$ satisfy the following equations
\begin{gather}
\label{agd=0}
\forall\,i_0=1,\dots,\nv,\qquad\forall\,k_0\ge 1,\qquad\quad
\frac{\pd \anA}{\pd u^{i_0}_{k_0}}
\,\,\bigg|_{u^i_k=a^i_k\ \forall\,(i,k)\succ(i_0,k_0-1)}=0,\\
\lb{agaukak}
 \anA\,\Big|_{u^i_k=a^i_k\ \forall\,(i,k)}=0,\\
\lb{agbxx0}
\anB\,\Big|_{x=x_a,\ u^i_k=a^i_k\ \forall\,(i,k)}=0.
\end{gather}

\begin{remark}
\lb{ranorm}
For example, the ZCR $\tilde{A}$, $\tilde{B}$ described in Theorem~\ref{evcov} is 
$a$-normal, because $\tilde{A}$, $\tilde{B}$ obey \er{gd=0}, 
\er{gaukak}, \er{gbxx0}.
Theorem~\ref{evcov} implies that any ZCR on a neighborhood of $a\in\CE$
is gauge equivalent to an $a$-normal ZCR.
\end{remark}

\begin{theorem}
\lb{thaca}
Let $\sm\in\zsp$.
Let 
\beq
\lb{aban}
\anA=\anA(x,t,u^j_0,u^j_1,\dots,u^j_\oc),\quad 
\anB=\anB(x,t,u^j_0,u^j_1,\dots,u^j_{\oc+\eo-1}),\quad
D_x(\anB)-D_t(\anA)+[\anA,\anB]=0
\ee
be an $a$-normal ZCR with values in $\gl_\sm$.
\textup{(}So $\anA$, $\anB$ take values in $\gl_\sm$ and 
satisfy \er{agd=0}, \er{agaukak}, \er{agbxx0}.\textup{)}

Consider another $a$-normal ZCR with values in $\gl_\sm$
\begin{gather}
\lb{cabzcr}
\CA=\CA(x,t,u^j_0,u^j_1,\dots,u^j_q),\quad 
\CB=\CB(x,t,u^j_0,u^j_1,\dots,u^j_{q+\eo-1}),\quad
D_x(\CB)-D_t(\CA)+[\CA,\CB]=0,\\
\label{cd=0}
\forall\,i_0=1,\dots,\nv,\qquad\forall\,k_0\ge 1,\qquad\quad
\frac{\pd \CA}{\pd u^{i_0}_{k_0}}
\,\,\bigg|_{u^i_k=a^i_k\ \forall\,(i,k)\succ(i_0,k_0-1)}=0,\\
\lb{caukak}
\CA\,\Big|_{u^i_k=a^i_k\ \forall\,(i,k)}=0,\\
\lb{cbxx0}
\CB\,\Big|_{x=x_a,\ u^i_k=a^i_k\ \forall\,(i,k)}=0.
\end{gather}

Suppose that there is a function $\bG=\bG(x,t,u^j_0,u^j_1,\dots,u^j_l)$ 
with values in $\mathrm{GL}_\sm$ such that 
\begin{gather}
\lb{cabg}
\CA=\bG \anA\bG^{-1}-D_x(\bG)\cdot\bG^{-1},\\
\lb{cbbg}
\CB=\bG \anB\bG^{-1}-D_t(\bG)\cdot\bG^{-1}.
\end{gather}
In other words, we suppose that the $a$-normal ZCR $\anA,\,\anB$ is 
gauge equivalent to the $a$-normal ZCR $\CA,\,\CB$ with respect to 
a gauge transformation $\bG=\bG(x,t,u^j_0,u^j_1,\dots,u^j_l)$.

Then the function $\bG$ is actually a constant element of $\mathrm{GL}_\sm$
\textup{(}that is, $\bG$ does not depend on $x$, $t$, $u^j_k$\textup{)}, and we have
\beq
\lb{cabgcb}
\CA=\bG \anA\bG^{-1},\qquad\qquad
\CB=\bG \anB\bG^{-1}.
\ee
\end{theorem}
\begin{proof}
Using \er{agd=0}, \er{cd=0}, \er{cabg}, one can prove 
\beq
\lb{pdbgu}
\frac{\pd\bG}{\pd u^j_k}=0\qquad\quad\forall\,j,k
\ee
by descending induction on $(j,k)$ with respect to the ordering $\prec$.
Equation~\er{pdbgu} means that the function $\bG$ may depend only on $x$, $t$.

Now, taking into account~\er{pdbgu} and~\er{evdxdt}, 
we can rewrite \er{cabg}, \er{cbbg} as
\begin{gather}
\lb{adxbg}
\CA=\bG \anA\bG^{-1}-\frac{\pd\bG}{\pd x}\cdot\bG^{-1},\\
\lb{bdtbg}
\CB=\bG \anB\bG^{-1}-\frac{\pd\bG}{\pd t}\cdot\bG^{-1}.
\end{gather}
Substituting $u^i_k=a^i_k$ for all $i$, $k$ in~\er{adxbg} 
and using \er{agaukak}, \er{caukak}, we get
${\pd\bG}/{\pd x}=0$. Hence $\bG$ may depend only on $t$.
Substituting $x=x_a$ and $u^i_k=a^i_k$ for all $i$, $k$ in~\er{bdtbg} 
and using \er{agbxx0}, \er{cbxx0}, we get
${\pd\bG}/{\pd t}=0$.

Thus $\bG$ does not depend on $x$, $t$, $u^j_k$, 
so $\bG$ is a constant element of $\mathrm{GL}_\sm$.
Then $D_x(\bG)=D_t(\bG)=0$, and relations \er{cabg}, \er{cbbg} imply~\er{cabgcb}.
\end{proof}
\begin{remark}
\lb{remaca}
In the situation described in Theorem~\ref{thaca},
since $\bG$ is a constant element of $\mathrm{GL}_\sm$,
the equation $\CA=\bG \anA\bG^{-1}$ implies that the functions 
$\anA$ and $\CA$ depend on the same variables $x$, $t$, $u^j_k$.
\end{remark}

\subsection{The algebras $\fd^\oc(\CE,a)$}
\lb{deffdoc}

Recall that $\CE$ is the infinite prolongation of 
an $\nv$-component evolution PDE~\er{sys_intr}.
The number $\nv\in\zsp$ is fixed throughout this section.
Consider a point $a\in\CE$ given by~\er{axtaik}. 
According to~\er{axtaik}, the point $a$ 
is determined by constants $x_a$, $t_a$, $a^i_k$.

For each $\ml\in\zp$, let $\mat_{\ml}$ be the set of matrices of 
size~$\nv\times (\ml+1)$ with nonnegative integer entries. 
For a matrix $\gamma\in\mat_\ml$, 
its entries are denoted by $\gamma_{i,k}\in\zp$, 
where $i=1,\dots,\nv$ and $k=0,\dots,\ml$. 
Let $\ua^\gamma$ be the following product 
\begin{equation}
\label{ugamma}
\ua^\gamma=\prod_{\substack{i=1,\dots,\nv,\\ k=0,\dots,\ml}}
\big(u^i_k-a^i_k\big)^{\gamma_{i,k}}.
\end{equation}

\begin{remark}
\lb{remm}
For each $\ml\in\zsp$, $i_0\in\{1,\dots,\nv\}$, and $k_0\in\{1,\dots,\ml\}$, 
denote by $M_{i_0,k_0}^\ml\subset\mat_\ml$
the subset of matrices $\al$ satisfying the following conditions
\begin{gather}
\lb{ali0k0}
\al_{i_0,k_0}=1,\quad\forall\,k>k_0\quad\forall\,i\quad\al_{i,k}=0,
\quad\forall\,i_1\neq i_0\quad\al_{i_1,k_0}=0,\quad
\forall\,i_2>i_0\quad\al_{i_2,k_0-1}=0.
\end{gather}
In other words, for each $k>k_0$ the $k$-th column 
of any matrix $\al\in M_{i_0,k_0}^\ml$ is zero,
the $k_0$-th column contains only one nonzero entry $\al_{i_0,k_0}=1$, 
and in the $(k_0-1)$-th column
one has $\al_{i_2,k_0-1}=0$ for all $i_2>i_0$.

Set also $M_{i_0,k_0}^0=\varnothing$ for all $i_0$, $k_0$.
So the set $M_{i_0,k_0}^0$ is empty.
\end{remark}

Let $\sm\in\zsp$ and $\oc\in\zp$. 
Consider again a matrix Lie algebra $\mg\subset\gl_\sm$. 
According to Theorem~\ref{evcov}, any $\mg$-valued ZCR~\er{thzcroc} 
of order~$\le\oc$ defined on a neighborhood of $a\in\CE$ 
is gauge equivalent to a $\mg$-valued ZCR
\begin{gather}
\lb{tatbf}
\tilde{A}=\tilde{A}(x,t,u^j_0,u^j_1,\dots,u^j_\oc),\qquad 
\tilde{B}=\tilde{B}(x,t,u^j_0,u^j_1,\dots,u^j_{\oc+\eo-1}),\\
\lb{tatbzcr}
D_x(\tilde{B})-D_t(\tilde{A})+[\tilde{A},\tilde{B}]=0
\end{gather}
satisfying \er{gaukak}, \er{gbxx0}, \er{pgxt0ua}.

According to Remark~\ref{anmer}, the $\mg$-valued functions $\tilde{A}$, $\tilde{B}$
are analytic on a neighborhood of $a\in\CE$.
Hence, in some neighborhood of $a\in\CE$, the functions $\tilde{A}$, $\tilde{B}$
can be represented as absolutely convergent power series
\begin{gather}
\label{aser}
\tilde{A}=\sum_{\al\in \mat_\oc,\ l_1,l_2\in\zp}
(x-x_a)^{l_1}(t-t_a)^{l_2}\cdot \ua^\al\cdot \tilde{A}^{l_1,l_2}_\al,\\
\lb{bser}
\tilde{B}=\sum_{\beta\in \mat_{\oc+\eo-1},\ l_1,l_2\in\zp}
(x-x_a)^{l_1}(t-t_a)^{l_2}\cdot \ua^\beta\cdot\tilde{B}^{l_1,l_2}_\beta,\\
\notag
\tilde{A}^{l_1,l_2}_\al,\tilde{B}^{l_1,l_2}_\beta\in\mg.
\end{gather}

\begin{remark}
\lb{abcoef0}
Using formulas \er{aser}, \er{bser}, 
we see that properties \er{gaukak}, \er{gbxx0}, \er{pgxt0ua} are equivalent to
\beq
\lb{ab000}
\tilde{A}^{l_1,l_2}_0=\tilde{B}^{0,l_2}_0=0,\quad
\tilde{A}^{l_1,l_2}_{\hat{\al}}=0,\quad\hat{\al}\in M^\oc_{i_0,k_0},
\quad i_0=1,\dots,\nv,\quad k_0=1,\dots,\oc,
\quad l_1,l_2\in\zp,
\ee
where $M^\oc_{i_0,k_0}\subset\mat_\oc$ 
is the set of matrices defined in Remark~\ref{remm}.
\end{remark}

\begin{remark}  
\label{inform}
The main idea of the definition of the Lie algebra $\fds^\oc(\CE,a)$  
can be informally outlined as follows. 
According to Theorem~\ref{evcov} and Remark~\ref{abcoef0}, 
any ZCR~\er{thzcroc} of order~$\le\oc$ is gauge equivalent 
to a ZCR given by functions $\tilde{A}$, $\tilde{B}$ 
that are of the form~\er{aser},~\er{bser} 
and satisfy \er{tatbzcr}, \er{ab000}.

To define $\fd^\oc(\CE,a)$, 
we regard $A^{l_1,l_2}_\al$, $B^{l_1,l_2}_\beta$ from~\er{aser},~\er{bser} 
as abstract symbols. 
By definition, the Lie algebra $\fd^\oc(\CE,a)$ 
is generated by the symbols $A^{l_1,l_2}_\al$, $B^{l_1,l_2}_\beta$ 
for $\al\in\mat_\oc$, 
$\be\in\mat_{\oc+\eo-1}$, 
$l_1,l_2\in\zp$.
Relations for these generators are provided by equations \er{tatbzcr}, \er{ab000}.
The details of this construction are presented below.
\end{remark}
Let $\frl$ be the free Lie algebra generated 
by the symbols $\fla^{l_1,l_2}_\al$, $\flb^{l_1,l_2}_\beta$ for 
${\al\in\mat_\oc}$, ${\be\in\mat_{\oc+\eo-1}}$, $l_1,l_2\in\zp$.
In particular, we have
$$
\fla^{l_1,l_2}_\al\in\frl,\quad\ 
\flb^{l_1,l_2}_\beta\in\frl,\quad\ 
\big[\fla^{l_1,l_2}_\al,\flb^{l_1,l_2}_\beta\big]\in\frl\qquad 
\forall\,\al\in\mat_\oc,\qquad \forall\,\be\in\mat_{\oc+\eo-1},\qquad 
\forall\,l_1,l_2\in\zp.
$$
Consider the following formal power series with coefficients in~$\frl$
\begin{gather*}
\notag
\fla=\sum_{\al\in \mat_\oc,\ l_1,l_2\in\zp}(x-x_a)^{l_1}(t-t_a)^{l_2}\cdot \ua^\al
\cdot \fla^{l_1,l_2}_\al,\\
\flb=\sum_{\beta\in \mat_{\oc+\eo-1},\ l_1,l_2\in\zp}
(x-x_a)^{l_1}(t-t_a)^{l_2}\cdot \ua^\beta\cdot\flb^{l_1,l_2}_\beta.
\end{gather*}

Set 
\begin{gather}
\lb{dxflb}
D_x(\flb)=\sum_{\beta\in \mat_{\oc+\eo-1},\ l_1,l_2\in\zp}
D_x\big((x-x_a)^{l_1}(t-t_a)^{l_2}\ua^\beta\big)\cdot\flb^{l_1,l_2}_\beta,\\
\lb{dtfla}
D_t(\fla)=\sum_{\al\in \mat_\oc,\ l_1,l_2\in\zp}D_t\big((x-x_a)^{l_1}(t-t_a)^{l_2}\ua^\al\big)
\cdot \fla^{l_1,l_2}_\al,\\
\label{lieab}
[\fla,\flb]=\sum_{\substack{\al\in\mat_\oc,\ \beta\in \mat_{\oc+\eo-1},\\
l_1,l_2,l'_1,l'_2\in\zp}} 
(x-x_a)^{l_1+l'_1}(t-t_a)^{l_2+l'_2}\cdot 
\ua^\al\cdot \ua^{\beta}\cdot\big[\fla^{l_1,l_2}_\al,\flb^{l'_1,l'_2}_\beta\big].
\end{gather}
For any $\al\in\mat_\oc$, $\beta\in \mat_{\oc+\eo-1}$, $l_1,l_2\in\zp$, 
the expressions $D_x\big((x-x_a)^{l_1}(t-t_a)^{l_2}\ua^\beta\big)$ 
and $D_t\big((x-x_a)^{l_1}(t-t_a)^{l_2}\ua^\al\big)$ 
are functions of the variables $x$, $t$, $u^i_k$. 
Taking the corresponding Taylor series at the point~\eqref{axtaik}, 
we regard these expressions as power series. 

Then~\er{dxflb},~\er{dtfla},~\er{lieab} are 
formal power series with coefficients in~$\frl$, and we have 
\begin{equation*}
D_x(\flb)-D_t(\fla)+[\fla,\flb]=
\sum_{\gamma\in \mat_{\oc+\eo},\ l_1,l_2\in\zp}
(x-x_a)^{l_1}(t-t_a)^{l_2}\cdot \ua^\gamma\cdot\flz^{l_1,l_2}_\gamma
\end{equation*}
for some elements $\flz^{l_1,l_2}_\gamma\in\frl$. 

Let $\frid\subset\frl$ be the ideal generated by the elements
\begin{gather*}
\flz^{l_1,l_2}_\gamma,\qquad\fla^{l_1,l_2}_0,\qquad 
\flb^{0,l_2}_0,\qquad\gamma\in\mat_{\oc+\eo},\qquad l_1,l_2\in\zp,\\
\fla^{l_1,l_2}_{\hat{\al}},\qquad{\hat{\al}}\in M^\oc_{i_0,k_0},\qquad i_0=1,\dots,\nv,\qquad k_0=1,\dots,\oc,
\qquad l_1,l_2\in\zp.
\end{gather*}
Set $\fd^\oc(\CE,a)=\frl/\frid$. 
Consider the natural homomorphism  
$\rho\cl\frl\to\frl/\frid=\fd^\oc(\CE,a)$ and set 
$$
\ga^{l_1,l_2}_\al=\rho\big(\fla^{l_1,l_2}_\al\big),\qquad\qquad 
\gb^{l_1,l_2}_\beta=\rho\big(\flb^{l_1,l_2}_\beta\big).
$$
The definition of~$\frid$ implies that the power series 
\begin{gather}
\label{gasumxt}
\ga=\sum_{\al\in \mat_\oc,\ l_1,l_2\in\zp}(x-x_a)^{l_1}(t-t_a)^{l_2}\cdot \ua^\al
\cdot\ga^{l_1,l_2}_\al,\\
\label{gbsumxt}
\gb=\sum_{\beta\in \mat_{\oc+\eo-1},\ l_1,l_2\in\zp}
(x-x_a)^{l_1}(t-t_a)^{l_2}\cdot \ua^\beta\cdot\gb^{l_1,l_2}_\beta
\end{gather}
satisfy 
\beq
\lb{xgbtga}
D_x(\gb)-D_t(\ga)+[\ga,\gb]=0.
\ee

\begin{remark}
\lb{rem_fdpgen}
The Lie algebra $\fd^\oc(\CE,a)$ can be described in terms 
of generators and relations as follows. 

Equation~\er{xgbtga} is equivalent to some Lie algebraic relations 
for $\ga^{l_1,l_2}_\al$, $\gb^{l_1,l_2}_\beta$.

The algebra $\fd^\oc(\CE,a)$ is given by the generators 
$\ga^{l_1,l_2}_\al$, $\gb^{l_1,l_2}_\beta$,  
the relations arising from~\er{xgbtga}, and the following relations 
\beq
\lb{gagb00}
\ga^{l_1,l_2}_0=\gb^{0,l_2}_0=0,\quad
\ga^{l_1,l_2}_{\hat{\al}}=0,\quad\hat{\al}\in M^\oc_{i_0,k_0},\quad 
i_0=1,\dots,\nv,\quad k_0=1,\dots,\oc,
\quad l_1,l_2\in\zp.
\ee

Note that condition~\er{gagb00} is equivalent to the following equations
\begin{gather}
\lb{pdgau}
\forall\,i_0=1,\dots,\nv,\qquad\forall\,k_0\ge 1,\qquad\quad
\frac{\pd\ga}{\pd u^{i_0}_{k_0}}
\,\,\bigg|_{u^i_k=a^i_k\ \forall\,(i,k)\succ(i_0,k_0-1)}=0,\\
\lb{gaua0}
\ga\,\Big|_{u^i_k=a^i_k\ \forall\,(i,k)}=0,\\
\lb{gbxua0}
\gb\,\Big|_{x=x_a,\ u^i_k=a^i_k\ \forall\,(i,k)}=0.
\end{gather}
\end{remark}


\begin{remark}
\lb{rfzcr}
Let $\bl$ be a Lie algebra. 
If $\anA$, $\anB$ are functions with values in $\bl$ and satisfy~\er{anzcr} 
then $\anA$, $\anB$ form a ZCR of order~$\le\oc$ with values in $\bl$.

Instead of functions with values in $\bl$, 
one can consider formal power series with coefficients in $\bl$.
Then one gets the notion of \emph{formal ZCRs with coefficients in $\bl$}.

More precisely, a \emph{formal ZCR of order~$\le\oc$ with coefficients in $\bl$}
is given by power series 
\begin{gather}
\label{anasum}
\anA=\sum_{\al\in \mat_\oc,\ l_1,l_2\in\zp}(x-x_a)^{l_1}(t-t_a)^{l_2}\cdot\ua^\al
\cdot\anA^{l_1,l_2}_\al,\\
\label{anbsum}
\anB=\sum_{\beta\in \mat_{\oc+\eo-1},\ l_1,l_2\in\zp}
(x-x_a)^{l_1}(t-t_a)^{l_2}\cdot\ua^\beta\cdot\anB^{l_1,l_2}_\beta
\end{gather}
such that $\anA^{l_1,l_2}_\al,\anB^{l_1,l_2}_\beta\in\bl$ and 
$D_x(\anB)-D_t(\anA)+[\anA,\anB]=0$.

If the power series~\er{anasum},~\er{anbsum} 
satisfy \er{agd=0}, \er{agaukak}, \er{agbxx0} 
then this formal ZCR is said to be \emph{$a$-normal}.

For example, since \er{gasumxt}, \er{gbsumxt} 
obey \er{xgbtga}, \er{pdgau}, \er{gaua0}, \er{gbxua0} and 
$\ga^{l_1,l_2}_\al,\gb^{l_1,l_2}_\beta\in\fd^\oc(\CE,a)$, 
the power series \er{gasumxt}, \er{gbsumxt} constitute an $a$-normal 
formal ZCR of order~$\le\oc$ with coefficients in $\fd^\oc(\CE,a)$.
\end{remark}

\begin{remark}
\lb{rzfz}
Let $\bl$ be a Lie algebra. 
Consider an $\bl$-valued ZCR of order~$\le\oc$ given by 
$\bl$-valued functions $\anA$, $\anB$ satisfying~\er{anzcr}.
Then the Taylor series of the functions $\anA$, $\anB$ are of the form 
\er{anasum}, \er{anbsum} and constitute a formal ZCR with coefficients in $\bl$.

Thus any $\bl$-valued ZCR can be regarded as a formal ZCR with coefficients 
in $\bl$, if we replace the $\bl$-valued functions by the corresponding Taylor 
series with coefficients in $\bl$.
\end{remark}

The definition of the Lie algebra $\fds^\oc(\CE,a)$ implies the following result.
\begin{theorem}
\lb{thhfzcr}
Let \er{anasum}, \er{anbsum} be an $a$-normal
formal ZCR of order~$\le\oc$ with coefficients in a Lie algebra~$\bl$.
Then one has a homomorphism $\fds^\oc(\CE,a)\to\bl$ given by the formulas 
$\ga^{l_1,l_2}_\al\mapsto\anA^{l_1,l_2}_\al$, 
$\gb^{l_1,l_2}_\beta\mapsto\anB^{l_1,l_2}_\beta$ for all $\al$, $\beta$, $l_1$, $l_2$.
\end{theorem}
\begin{remark}
\lb{homfdhmg}
Let $\hat{A}=\hat{A}(x,t,u^j_0,u^j_1,\dots,u^j_\oc)$,
$\hat{B}=\hat{B}(x,t,u^j_0,u^j_1,\dots,u^j_{\oc+\eo-1})$ be an $a$-normal ZCR 
of order~$\le\oc$ with values in a Lie algebra $\hat{\mg}$.
We suppose that the $\hat{\mg}$-valued functions $\hat{A}$, $\hat{B}$ 
are defined on a neighborhood of $a\in\CE$.

Then the Taylor series of the functions $\hat{A}$, $\hat{B}$ at the point $a\in\CE$ 
constitute an $a$-normal formal ZCR of order~$\le\oc$ with coefficients in~$\hat{\mg}$.
Therefore, by Theorem~\ref{thhfzcr}, 
we get a homomorphism $\fds^\oc(\CE,a)\to\hat{\mg}$ which maps the coefficients 
of the power series \er{gasumxt}, \er{gbsumxt} to the corresponding coefficients 
of the Taylor series of the functions $\hat{A}$, $\hat{B}$.
\end{remark}

Let $\mg$ be a finite-dimensional Lie algebra.
A homomorphism $\hrf\cl\fds^\oc(\CE,a)\to\mg$ 
is said to be \emph{regular} if the power series 
\begin{gather}
\label{taser}
\tilde{A}=\sum_{\al\in \mat_\oc,\ l_1,l_2\in\zp}
(x-x_a)^{l_1}(t-t_a)^{l_2}\cdot \ua^\al\cdot
\hrf\big(\ga^{l_1,l_2}_\al\big),\\
\lb{tbser}
\tilde{B}=\sum_{\beta\in \mat_{\oc+\eo-1},\ l_1,l_2\in\zp}
(x-x_a)^{l_1}(t-t_a)^{l_2}\cdot\ua^\beta\cdot
\hrf\big(\gb^{l_1,l_2}_\beta\big)
\end{gather}
are absolutely convergent in a neighborhood of $a\in\CE$. 
In other words, $\hrf$ is regular iff \er{taser},~\er{tbser} are 
analytic functions on a neighborhood of $a\in\CE$.

Since \er{gasumxt},~\er{gbsumxt} obey~\er{xgbtga}, 
the power series~\er{taser},~\er{tbser} satisfy~\er{tatbzcr} 
for any homomorphism $\hrf\cl\fds^\oc(\CE,a)\to\mg$.  
Therefore, if $\hrf$ is regular, 
the analytic functions~\er{taser},~\er{tbser} form a ZCR with values in $\mg$.
Denote this ZCR by $\mathbf{Z}(\CE,a,\oc,\hrf)$.

Combining this construction with Theorem~\ref{evcov} and Remark~\ref{abcoef0}, 
we obtain the following result.
\begin{theorem}
\lb{thzcrfd}
Let $\mg$ be a finite-dimensional matrix Lie algebra.
For any $\mg$-valued ZCR~\er{thzcroc} of order~$\le\oc$ 
on a neighborhood of $a\in\CE$, 
there is a regular homomorphism $\hrf\cl\fds^\oc(\CE,a)\to\mg$ 
such that the ZCR~\er{thzcroc} is gauge equivalent to 
the ZCR $\mathbf{Z}(\CE,a,\oc,\hrf)$ given by \er{taser}, \er{tbser}.

The ZCR $\mathbf{Z}(\CE,a,\oc,\hrf)$ takes values in the Lie algebra 
$\hrf\big(\fds^\oc(\CE,a)\big)\subset\mg$.

Recall that we have the power series $\ga$, $\gb$ with coefficients in 
$\fds^\oc(\CE,a)$ given by formulas~\er{gasumxt},~\er{gbsumxt}.
Formulas \er{taser}, \er{tbser} say that $\tilde{A}=\hrf(\ga)$ and 
$\tilde{B}=\hrf(\gb)$, in the sense that $\hrf$ maps the coefficients 
of the power series $\ga$, $\gb$ to the corresponding 
coefficients of the Taylor series of the functions $\tilde{A}$, $\tilde{B}$.

So the ZCR~\er{thzcroc} is gauge equivalent to the ZCR $\mathbf{Z}(\CE,a,\oc,\hrf)$
given by the functions $\tilde{A}=\hrf(\ga)$, $\tilde{B}=\hrf(\gb)$.
\end{theorem}

\subsection{The homomorphism $\fd^{\oc}(\CE,a)\to\fd^{\oc-1}(\CE,a)$}
\lb{hfdpp1}

According to Remarks~\ref{rem_fdpgen},~\ref{rfzcr}, the Lie algebra $\fd^{\oc}(\CE,a)$ 
is generated by $\ga^{l_1,l_2}_\al$, $\gb^{l_1,l_2}_\beta$, 
and the power series \er{gasumxt}, \er{gbsumxt} constitute an $a$-normal 
formal ZCR of order~$\le\oc$ with coefficients in $\fd^\oc(\CE,a)$.

In this subsection we suppose that $\oc\ge 1$. 
Constructing the Lie algebra $\fd^{\oc-1}(\CE,a)$ in the same way, 
we get power series 
\begin{gather}
\label{hgasum}
\hga=\sum_{\hal\in\mat_{\oc-1},\ l_1,l_2\in\zp}(x-x_a)^{l_1}(t-t_a)^{l_2}\cdot\ua^\hal
\cdot\hga^{l_1,l_2}_\hal,\\
\label{hgbsum}
\hgb=\sum_{\hbeta\in\mat_{\oc+\eo-2},\ l_1,l_2\in\zp}
(x-x_a)^{l_1}(t-t_a)^{l_2}\cdot\ua^\hbeta\cdot\hgb^{l_1,l_2}_\hbeta
\end{gather}
such that the Lie algebra $\fd^{\oc-1}(\CE,a)$ 
is generated by $\hga^{l_1,l_2}_\hal$, $\hgb^{l_1,l_2}_\hbeta$
for all $\hal\in\mat_{\oc-1}$, $\hbeta\in\mat_{\oc+\eo-2}$, $l_1,l_2\in\zp$,
and the power series \er{hgasum}, \er{hgbsum} constitute an $a$-normal 
formal ZCR of order~$\le\oc-1$ with coefficients in $\fd^{\oc-1}(\CE,a)$.

We are going to construct a surjective homomorphism 
$\tau_\oc\cl\fd^{\oc}(\CE,a)\to\fd^{\oc-1}(\CE,a)$. 
Since the algebra $\fd^{\oc}(\CE,a)$ 
is generated by $\ga^{l_1,l_2}_\al$, $\gb^{l_1,l_2}_\beta$, 
it is sufficient to define 
$\tau_\oc(\ga^{l_1,l_2}_\al)$, $\tau_\oc(\gb^{l_1,l_2}_\beta)$. 

To do this, we need to introduce some extra notation.
Recall that, for each $\ml\in\zp$, we denote by $\mat_{\ml}$ 
the set of matrices of size~$\nv\times (\ml+1)$ with nonnegative integer entries. 
For a matrix $\gamma\in\mat_\ml$, 
its entries are denoted by $\gamma_{i,k}\in\zp$, 
where $i=1,\dots,\nv$ and $k=0,\dots,\ml$. 

For each $\ml\ge 1$ and each $\nv\times (\ml+1)$ matrix $\gamma\in\mat_\ml$, 
we denote by $\mathbf{r}(\gamma)$ the $\nv\times\ml$ matrix with the entries 
$\mathbf{r}(\gamma)_{i,k}=\gamma_{i,k}$ for $i=1,\dots,\nv$ and $k=0,\dots,\ml-1$.
In other words, the matrix $\mathbf{r}(\gamma)\in\mat_{\ml-1}$ is obtained 
from the matrix $\gamma$ by erasing the last column.

Let $l_1,l_2\in\zp$.
For $\al\in\mat_\oc$, $\beta\in \mat_{\oc+\eo-1}$, 
we can consider the matrices $\mathbf{r}(\al)\in\mat_{\oc-1}$, 
$\mathbf{r}(\beta)\in\mat_{\oc+\eo-2}$
and the elements 
$\hga^{l_1,l_2}_{\mathbf{r}(\al)},\hgb^{l_1,l_2}_{\mathbf{r}(\beta)}\in\fd^{\oc-1}(\CE,a)$.

For all $l_1,l_2\in\zp$, $\al\in\mat_{\oc}$, $\beta\in\mat_{\oc+\eo-1}$,
we set 
\begin{gather}
\lb{tocga}
\tau_\oc(\ga^{l_1,l_2}_\al)=
\left\{
\begin{aligned}
& 0,\text{ if there is $i\in\{1,\dots,\nv\}$ such that }\al_{i,\oc}\neq 0,\\
& \hga^{l_1,l_2}_{\mathbf{r}(\al)},\text{ if $\al_{i,\oc}=0$ for all $i$},
\end{aligned}\right.\\
\lb{tocgb}
\tau_\oc(\gb^{l_1,l_2}_\beta)=
\left\{
\begin{aligned}
& 0,\text{ if there is $i\in\{1,\dots,\nv\}$ such that }\beta_{i,\oc+\eo-1}\neq 0,\\
& \hgb^{l_1,l_2}_{\mathbf{r}(\beta)},\text{ if $\beta_{i,\oc+\eo-1}=0$ for all $i$}.
\end{aligned}\right.
\end{gather}
The definition of $\fd^{\oc}(\CE,a)$ and $\fd^{\oc-1}(\CE,a)$ implies that
$\tau_\oc\cl\fd^{\oc}(\CE,a)\to\fd^{\oc-1}(\CE,a)$ defined by~\er{tocga},~\er{tocgb}
is indeed a surjective homomorphism. The meaning of this homomorphism 
is explained in Remark~\ref{rtoc1} below.

According to~\er{tocga}, one has $\tau_\oc(\ga^{l_1,l_2}_\al)=0$ if there is a 
nonzero entry in the last column of the matrix $\al\in\mat_{\oc}$. 
According to~\er{tocgb}, one has $\tau_\oc(\gb^{l_1,l_2}_\beta)=0$ if there is a 
nonzero entry in the last column of the matrix $\beta\in\mat_{\oc+\eo-1}$. 

Recall that the power series $\ga$ and $\gb$ are given by \er{gasumxt}, \er{gbsumxt}.
Taking into account formula~\er{ugamma}, 
we see that formulas \er{tocga}, \er{tocgb} say the following.
Applying $\tau_\oc$ to the coefficients of the power series 
$\ga\,\Big|_{u^i_\oc=a^i_\oc\ \forall\,i}$, we get the power series~\er{hgasum}.
Applying $\tau_\oc$ to the coefficients of the power series 
$\gb\,\Big|_{u^i_{\oc+\eo-1}=a^i_{\oc+\eo-1}\ \forall\,i}$, 
we get the power series~\er{hgbsum}.

\begin{remark}
\lb{rtoc1}
Any ZCR of order~$\le\oc-1$ is at the same time of order~$\le\oc$.
Therefore, \er{hgasum}, \er{hgbsum} can be regarded as 
an $a$-normal formal ZCR of order~$\le\oc$ with coefficients in $\fd^{\oc-1}(\CE,a)$.
The homomorphism $\tau_\oc\cl\fd^{\oc}(\CE,a)\to\fd^{\oc-1}(\CE,a)$ is the homomorphism 
which corresponds to this ZCR by Theorem~\ref{thhfzcr}.
\end{remark}

Thus we obtain the following sequence of surjective homomorphisms of Lie algebras
\beq
\lb{fdnn-1}
\dots\xrightarrow{\tau_{\oc+1}}
\fd^{\oc}(\CE,a)\xrightarrow{\tau_{\oc}}\fd^{\oc-1}(\CE,a)\xrightarrow{\tau_{\oc-1}}
\dots\xrightarrow{\tau_{2}}\fd^1(\CE,a)\xrightarrow{\tau_{1}}\fd^0(\CE,a).
\ee

\subsection{Some results on generators of $\fds^{\oc}(\CE,a)$}
\lb{secrgen}

According to Remark~\ref{rem_fdpgen}, 
the algebra $\fds^\oc(\CE,a)$ is given by the generators 
$\ga^{l_1,l_2}_\al$, $\gb^{l_1,l_2}_\beta$ 
and the relations arising from~\er{xgbtga},~\er{gagb00}.
Using~\er{evdxdt}, we can rewrite equation~\er{xgbtga} as
\begin{multline}
\lb{zcrdet}
\frac{\pd}{\pd x}(\gb)+
\sum_{\substack{i=1,\dots,\nv,\\ k=0,1,\dots,\oc+\eo-1}}
u^i_{k+1}\frac{\pd}{\pd u^i_k}(\gb)
-\frac{\pd}{\pd t}(\ga)\\
-\sum_{\substack{i=1,\dots,\nv,\\ k=0,1,\dots,\oc}}
D_x^k\big(F^i(x,t,u^j_0,u^j_1,\dots,u^j_{\eo})\big)\frac{\pd}{\pd u^i_k}(\ga)
+[\ga,\gb]=0.
\end{multline}
Here we regard $F^i=F^i(x,t,u^j_0,u^j_1,\dots,u^j_{\eo})$ as a power series, 
using the Taylor series of the function~$F^i$ at the point~\er{axtaik}.

\begin{theorem}
\lb{lemgenfdq}
The elements 
\beq
\lb{gal1alprop}
\ga^{l_1,0}_\al,\qquad\qquad l_1\in\zp,\qquad\al\in\mat_\oc,
\ee
generate the algebra $\fds^\oc(\ce,a)$.
\end{theorem}
\begin{proof}
For each $l\in\zp$, 
denote by $\agn_l\subset \fds^{\oc}(\CE,a)$ the subalgebra generated by 
all the elements $\ga^{l_1,l_2}_\al$ with $l_2\le l$.
To prove Theorem~\ref{lemgenfdq}, we need several lemmas. 

\begin{lemma}
\label{gbllalgl}
Let $l_1,l_2\in\zp$.
Let $\beta\in\mat_{\oc+\eo-1}$ be such that not all entries of the matrix
$\beta$ are zero 
\textup{(}i.e., in the matrix $\beta$ there is a nonzero entry\textup{)}.
Then $\gb^{l_1,l_2}_\beta\in\agn_{l_2}$.
\end{lemma}
\begin{proof}
One can prove this lemma, analyzing equation~\er{zcrdet} and 
properties~\er{gagb00},~\er{pdgau}.
\end{proof}

\begin{lemma}
\label{gb00mgl}
For all $l_1,l_2\in\zp$, one has $\gb^{l_1,l_2}_{0}\in\agn_{l_2}$.
Here $0\in\mat_{\oc+\eo-1}$ is the matrix with zero entries.
\end{lemma}
\begin{proof}
According to~\er{gagb00}, we have $\gb^{0,l_2}_0=0$. 
Therefore, it is sufficient 
to prove $\gb^{l_1,l_2}_0\in\agn_{l_2}$ for $l_1>0$.


Note that property~\er{pdgau} implies 
\beq
\lb{gatga0}
\ga\,\Big|_{u^i_k=a^i_k\ \forall\,(i,k)}=0,\qquad\qquad
\frac{\pd}{\pd t}(\ga)\,\bigg|_{u^i_k=a^i_k\ \forall\,(i,k)}=0.
\ee
In view of~\er{gbsumxt}, one has  
\beq
\lb{pdxgb0}
\frac{\pd}{\pd x}(\gb)\,\bigg|_{u^i_k=a^i_k\ \forall\,(i,k)}=
\sum_{l_1>0,\ l_2\ge 0}l_1(x-x_a)^{l_1-1} (t-t_a)^{l_2}\cdot\gb^{l_1,l_2}_0.
\ee
Substituting $u^i_k=a^i_k$ for all $i$, $k$ in~\er{zcrdet} 
and using~\er{gatga0},~\er{pdxgb0}, we get 
\begin{multline}
\lb{suml1gb0}
\sum_{l_1>0,\ l_2\ge 0} 
l_1(x-x_a)^{l_1-1} (t-t_a)^{l_2}\cdot\gb^{l_1,l_2}_0=\\
=-\bigg(\sum_{i,k}
u^i_{k+1}\frac{\pd}{\pd u^i_k}(\gb)\bigg)
\,\bigg|_{u^i_k=a^i_k\ \forall\,(i,k)}
+\bigg(\sum_{i,k}
D_x^k\big(F^i\big)\frac{\pd}{\pd u^i_k}(\ga)\bigg)\,
\bigg|_{u^i_k=a^i_k\ \forall\,(i,k)}.
\end{multline}

For a matrix $\beta\in\mat_{\oc+\eo-1}$, we denote by $|\beta|$ the sum 
of all entries of $\beta$.
Combining~\er{gasumxt},~\er{gbsumxt},~\er{suml1gb0}, 
we see that for any $l_1>0$ and $l_2\ge 0$ the element $\gb^{l_1,l_2}_0$
is equal to a linear combination of elements of the form 
\beq
\label{gagabj1}
\ga^{l'_1,l_2}_\al,\qquad\gb^{l'_1,l_2}_\beta,\qquad
l'_1\in\zp,\qquad 
\al\in\mat_\oc,\qquad\beta\in\mat_{\oc+\eo-1},\qquad
|\beta|=1.
\ee
According to Lemma~\ref{gbllalgl} and the definition of $\agn_{l_2}$, 
the elements~\er{gagabj1} belong to $\agn_{l_2}$. 
Thus $\gb^{l_1,l_2}_0\in\agn_{l_2}$.
\end{proof}

\begin{lemma}
\label{gallmg}
For all $l_1,l\in\zp$ and $\al\in\mat_\oc$, we have 
$\ga^{l_1,l+1}_\al\in\agn_l$.
\end{lemma}
\begin{proof}
Using~\er{gasumxt}, we can rewrite equation~\er{zcrdet} as
\begin{multline*}
\frac{\pd}{\pd t}(\ga)=
\sum_{\al\in\mat_\oc,\ l_1,l\in\zp}(l+1) 
(x-x_a)^{l_1}(t-t_a)^{l}\cdot\ua^\al\cdot\ga^{l_1,l+1}_\al=\\
=\frac{\pd}{\pd x}(\gb)+
\sum_{i,k}u^i_{k+1}\frac{\pd}{\pd u^i_k}(\gb)
-\sum_{i,k}D_x^k\big(F^i\big)\frac{\pd}{\pd u^i_k}(\ga)+[\ga,\gb].
\end{multline*}
This implies that $\ga^{l_1,l+1}_\al$ 
is equal to a linear combination of elements of the form 
\beq
\label{elemgall}
\ga^{\hat{l}_1,\hat{l}_2}_{\hat{\al}},\quad
\gb^{\tilde{l}_1,\tilde{l}_2}_{\tilde{\beta}},\quad
\Big[\ga^{\hat{l}_1,\hat{l}_2}_{\hat{\al}},
\gb^{\tilde{l}_1,\tilde{l}_2}_{\tilde{\beta}}\Big],\quad
\hat{l}_2\le l,\quad \tilde{l}_2\le l,\quad 
\hat{l}_1,\tilde{l}_1\in\zp,\quad
\hat{\al}\in\mat_\oc,\quad \tilde{\beta}\in\mat_{\oc+\eo-1}.
\ee
Using Lemmas~\ref{gbllalgl},~\ref{gb00mgl} 
and the condition~$\tilde{l}_2\le l$, we get 
$\gb^{\tilde{l}_1,\tilde{l}_2}_{\tilde{\beta}}\in\agn_{\tilde{l}_2}\subset\agn_{l}$. 
Therefore, the elements~\er{elemgall} belong to $\agn_{l}$. 
Hence $\ga^{l_1,l+1}_\al\in\agn_{l}$.
\end{proof}

Now we return to the proof of Theorem~\ref{lemgenfdq}.
According to Lemmas~\ref{gbllalgl},~\ref{gb00mgl} and the definition of~$\agn_{l}$, 
we have 
$\ga^{l_1,l_2}_\al,\gb^{l_1,l_2}_\beta\in\agn_{l_2}$ 
for all $l_1,l_2\in\zp$, $\al\in\mat_\oc$, $\beta\in\mat_{\oc+\eo-1}$. 
Lemma~\ref{gallmg} implies that 
\beq
\notag
\agn_{l_2}\subset\agn_{l_2-1}\subset\agn_{l_2-2}\subset\dots\subset\agn_0.
\ee
Therefore, $\fds^{\oc}(\CE,a)$ is equal to $\agn_0$, which is generated by the elements~\er{gal1alprop}.
\end{proof}

\subsection{Some constructions with zero-curvature representations}
\lb{sbdszcr}

We continue to work with an evolution PDE~\er{sys_intr}, 
which can be written also as~\er{uitfi}, according to our notation.

Let $\bl^1$ and $\bl^2$ be Lie algebras.
For $i=1,2$, let 
\beq
\lb{aibi}
A^i=A^i(x,t,u^j_0,u^j_1,\dots,u^j_\oc),\quad 
B^i=B^i(x,t,u^j_0,u^j_1,\dots,u^j_{\oc+\eo-1}),\quad
D_x(B)-D_t(A)+[A,B]=0
\ee
be an $\bl^i$-valued ZCR for the PDE~\er{sys_intr}. 
So the functions $A^i$, $B^i$ take values in $\bl^i$.
The following notions will be needed in the next sections.

The \emph{direct sum} of the $\bl^1$-valued ZCR $A^1,\,B^1$ 
and the $\bl^2$-valued ZCR $A^2,\,B^2$ is the $(\bl^1\oplus\bl^2)$-valued ZCR 
given by the functions $A^1\oplus A^2$, $B^1\oplus B^2$. 
So the ZCR given by the functions $A^1\oplus A^2$, $B^1\oplus B^2$ 
takes values in the Lie algebra $\bl^1\oplus\bl^2$. 

We say that the ZCR $A^2,\,B^2$ is a \emph{reduction} of the ZCR $A^1,\,B^1$
if there is a homomorphism $\rho\cl\bl^1\to\bl^2$ such that
$A^2=\rho(A^1)$ and $B^2=\rho(B^1)$.

Similarly, one can speak also about direct sums and reductions of formal ZCRs.

\begin{remark}
\lb{remunxi}
For any (possibly infinite-dimensional) Lie algebra $\bl$, 
there is a (possibly infinite-dimensional) vector space $V$ such that 
$\bl$ is isomorphic to a Lie subalgebra of $\gl(V)$.
Here $\gl(V)$ is the algebra of linear maps $V\to V$.

For example, one can use the following construction. 
Denote by $\un(\bl)$ the universal enveloping algebra of $\bl$.
We have the injective homomorphism of Lie algebras 
\beq
\notag
\xi\cl \bl\hookrightarrow\gl(\un(\bl)),\quad\qquad \xi(v)(w)=vw,
\quad\qquad v\in \bl,\quad\qquad w\in\un(\bl).
\ee
So one can set $V=\un(\bl)$.

So we have $\bl\subset\gl(V)$.
Let $k\in\zp$. 
A \emph{formal gauge transformation of order~$\le k$}
is a formal power series of the form 
\beq
\lb{ggamma}
G=\sum_{\gamma\in \mat_k,\ l_1,l_2\in\zp}(x-x_a)^{l_1}(t-t_a)^{l_2}\cdot\ua^\gamma
\cdot G^{l_1,l_2}_\gamma,\qquad\quad 
G^{l_1,l_2}_\gamma\in\gl(V),
\ee
such that the map $G^{0,0}_0\colon V\to V$ is invertible.
(So the free term $G^{0,0}_0$ of the power series~\er{ggamma} is invertible.)
Then $G^{-1}$ is well defined and is a power series with coefficients 
in $\gl(V)$ as well. 

Recall that a formal ZCR \er{anasum}, \er{anbsum} with coefficients in $\bl$
is given by formal power series $\anA$, $\anB$ with coefficients in $\bl$ 
satisfying $D_x(\anB)-D_t(\anA)+[\anA,\anB]=0$.
Then 
\beq
\lb{anaanb}
\tilde{\anA}=G\anA G^{-1}-D_x(G)\cdot G^{-1},\qquad\qquad
\tilde{\anB}=G\anB G^{-1}-D_t(G)\cdot G^{-1}
\ee
are formal power series with coefficients in $\gl(V)$ and 
satisfy $D_x(\tilde{\anB})-D_t(\tilde{\anA})+[\tilde{\anA},\tilde{\anB}]=0$.
Therefore, \er{anaanb} is a formal ZCR with coefficients in $\gl(V)$.

The formal ZCR~\er{anaanb} is said to be gauge equivalent 
to the formal ZCR~\er{anasum},~\er{anbsum} with respect to the formal 
gauge transformation~\er{ggamma}.

Quite often, it happens that the coefficients of the power series~\er{anaanb} 
belong to $\bl\subset\gl(V)$. Then \er{anaanb} can be regarded 
as a formal ZCR with coefficients in $\bl$.

This allows us to speak about gauge equivalence for formal ZCRs 
with coefficients in infinite-dimensional Lie algebras.
\end{remark}

\section{Relations between $\fd^0(\CE,a)$ and 
the Wahlquist-Estabrook prolongation algebra}
\lb{fd0we}

Let $\nv$, $\eo$ be positive integers.
Consider an $\nv$-component evolution PDE of the form
\begin{gather}
\label{gevxt}
u^i_t=F^i(u^j_0,u^j_1,\dots,u^j_{\eo}),\\
\notag
u^i=u^i(x,t),\qquad u^i_0=u^i,\qquad
u^i_k=\frac{\pd^k u^i}{\pd x^k},\qquad  
i,j=1,\dots,\nv,\qquad k\in\zp. 
\end{gather}
Note that the functions $F^i$ in~\er{gevxt} do not depend on $x$, $t$. 

Let $\CE$ be the infinite prolongation of the PDE~\er{gevxt}.
According Section~\ref{subsip}, $\CE$ is an infinite-dimensional manifold 
with the coordinates $x$, $t$, $u^i_k$ for $i=1,\dots,\nv$ and $k\in\zp$.

Consider a point $a\in\CE$ given by~\er{axtaik}.
The constants $x_a,t_a,a^i_k\in\fik$ from~\er{axtaik} are the coordinates 
of the point $a\in\CE$ in the coordinate system $x$, $t$, $u^i_k$.

As has been said in Section~\ref{deffdoc},
for each $\ml\in\zp$, we denote by $\mat_{\ml}$ the set of matrices of 
size~$\nv\times (\ml+1)$ with nonnegative integer entries. 
For a matrix $\gamma\in\mat_\ml$, 
its entries are denoted by $\gamma_{i,k}\in\zp$, 
where $i=1,\dots,\nv$ and $k=0,\dots,\ml$. 

According to formula~\er{ugamma}, for $\al\in\mat_0$ and $\beta\in\mat_{\eo-1}$ we have
$$
\ua^\al=\prod_{i=1,\dots,\nv}\big(u^i_0-a^i_0\big)^{\al_{i,0}},\qquad\quad
\ua^\beta=\prod_{\substack{i=1,\dots,\nv,\\ k=0,\dots,\eo-1}}
\big(u^i_k-a^i_k\big)^{\beta_{i,k}}.
$$
The \emph{Wahlquist-Estabrook prolongation algebra} of the PDE~\er{gevxt} 
at the point $a\in\CE$ can be defined as follows. Consider formal power series 
\beq
\label{wgawgbser}
\wga=\sum_{\al\in\mat_0}\ua^\al\cdot\wga_\al,\qquad\qquad
\wgb=\sum_{\beta\in\mat_{\eo-1}}\ua^\beta\cdot\wgb_\beta,
\ee 
where $\wga_\al$, $\wgb_\beta$ are generators of some Lie algebra, which is described below.
The equation 
\beq
\lb{wxgbtga}
D_x(\wgb)-D_t(\wga)+[\wga,\wgb]=0
\ee
is equivalent to some Lie algebraic relations for $\wga_\al$, $\wgb_\beta$. 
The Wahlquist-Estabrook prolongation algebra (WE algebra for short) 
at the point $a\in\CE$ is the Lie algebra given by the generators 
$\wga_\al$, $\wgb_\beta$ and the relations arising from~\er{wxgbtga}. 
A more detailed definition of the WE algebra is presented in~\cite{mll-2012}.
We denote this Lie algebra by $\wea$.

Then \er{wgawgbser}, \er{wxgbtga} is called 
the \emph{formal Wahlquist-Estabrook ZCR with coefficients in~$\wea$}.

For each $\oc\in\zp$, 
the Lie algebra $\fd^\oc(\ce,a)$ has been defined in Section~\ref{deffdoc}.
In the present section we study $\fd^0(\ce,a)$. 
We are going to show that the algebra $\fd^0(\ce,a)$ for the PDE~\er{gevxt} 
is isomorphic to some Lie subalgebra of $\wea$. 

Since the PDE~\er{gevxt} is invariant 
with respect to the change of variables $x\mapsto x-x_a$, $t\mapsto t-t_a$, 
we can assume $x_a=t_a=0$ in~\er{axtaik}. 
Then \er{gasumxt}, \er{gbsumxt}, \er{xgbtga} in the case $\oc=0$ can be written as
\begin{gather}
\label{gagab0}
\ga=\sum_{\al\in \mat_0,\ l_1,l_2\in\zp}x^{l_1}t^{l_2}\cdot\ua^\al
\cdot\ga^{l_1,l_2}_\al,\qquad
\gb=\sum_{\beta\in \mat_{\eo-1},\ l_1,l_2\in\zp}
x^{l_1}t^{l_2}\cdot \ua^\beta\cdot\gb^{l_1,l_2}_\beta,\\
\lb{gagab0zcr}
D_x(\gb)-D_t(\ga)+[\ga,\gb]=0,
\end{gather}
where $\ga^{l_1,l_2}_\al,\gb^{l_1,l_2}_\beta\in\fd^0(\CE,a)$.
Note that equation~\er{gagab0zcr} is equivalent to some Lie algebraic relations 
for $\ga^{l_1,l_2}_\al$, $\gb^{l_1,l_2}_\beta$.

According to Remark~\ref{rem_fdpgen},
the Lie algebra $\fd^0(\CE,a)$ can be described in terms 
of generators and relations as follows. 
The algebra $\fd^0(\CE,a)$ is given by the generators 
$\ga^{l_1,l_2}_\al$, $\gb^{l_1,l_2}_\beta$,  
the relations arising from~\er{gagab0zcr}, and the relations 
$\ga^{l_1,l_2}_0=\gb^{0,l_2}_0=0$ for $l_1,l_2\in\zp$.

\begin{lemma}
\lb{lemgenfd0}
The elements 
\beq
\lb{gal1al}
\ga^{l_1,0}_\al,\qquad\qquad l_1\in\zp,\qquad\al\in\mat_0,\qquad\al\neq 0,
\ee
generate the algebra $\fd^0(\ce,a)$.
\end{lemma}
\begin{proof}
According to Theorem~\ref{lemgenfdq},
the elements $\ga^{l_1,0}_\al$ for $l_1\in\zp$, $\al\in\mat_0$ 
generate the algebra $\fd^0(\ce,a)$. Since $\ga^{l_1,0}_0=0$, 
we can assume $\al\neq 0$.
\end{proof}

\begin{lemma}
\label{lemfd0zcr}
Let $\bl$ be a Lie algebra.
Consider formal power series of the form  
\beq
\lb{pqsum}
P=\sum_{\al\in \mat_0,\ l_1,l_2\in\zp}x^{l_1}t^{l_2}\cdot\ua^\al
\cdot P^{l_1,l_2}_\al,\qquad\quad
Q=\sum_{\beta\in \mat_{\eo-1},\ l_1,l_2\in\zp}
x^{l_1}t^{l_2}\cdot\ua^\beta\cdot Q^{l_1,l_2}_\beta
\ee
satisfying
\begin{gather}
\lb{pqg0}
P^{l_1,l_2}_\al,\,Q^{l_1,l_2}_\beta\in\bl,\qquad\quad 
P^{l_1,l_2}_0=Q^{0,l_2}_0=0\qquad\quad\forall\,l_1,l_2\in\zp,\\
\lb{zcrqp}
D_x(Q)-D_t(P)+[P,Q]=0.
\end{gather}
Then the map $\ga^{l_1,l_2}_\al\mapsto P^{l_1,l_2}_\al,\,\ 
\gb^{l_1,l_2}_\beta\mapsto Q^{l_1,l_2}_\beta$ 
determines a homomorphism from $\fd^0(\ce,a)$ to $\bl$. 
\end{lemma}
\begin{proof}
According to Remark~\ref{rfzcr},
formulas~\er{pqsum},~\er{pqg0},~\er{zcrqp} say that the power series $P$, $Q$ 
constitute an $a$-normal formal ZCR of order~$\le 0$ with coefficients in $\bl$.

According to Theorem~\ref{thhfzcr}, this formal ZCR determines 
a homomorphism $\fds^0(\CE,a)\to\bl$ given by the formulas 
$\ga^{l_1,l_2}_\al\mapsto P^{l_1,l_2}_\al$ and  
$\gb^{l_1,l_2}_\beta\mapsto Q^{l_1,l_2}_\beta$.
\end{proof}

Let $\bl$ be a Lie algebra. 
A \emph{ZCR of Wahlquist-Estabrook type with coefficients in $\bl$} 
is given by formal power series  
\begin{equation}
\lb{pqmg}
P=\sum_{\al\in\mat_0}\ua^\al\cdot P_\al,\qquad\qquad
Q=\sum_{\beta\in\mat_{\eo-1}}\ua^\beta\cdot Q_\beta,\qquad\qquad 
P_\al,\,Q_\beta\in\bl,
\end{equation}
satisfying 
\beq
\lb{pqzcr}
D_x(Q)-D_t(P)+[P,Q]=0.
\ee

The next lemma follows from the definition of the WE algebra~$\wea$.
\begin{lemma}
\label{propwezcr}
Recall that the WE algebra $\wea$ is generated by 
$\wga_\al$, $\wgb_\beta$ for $\al\in\mat_0$, $\beta\in\mat_{\eo-1}$ 
such that one has formulas~\er{wgawgbser}, \er{wxgbtga}.

For any Lie algebra $\bl$, 
any ZCR of Wahlquist-Estabrook type~\er{pqmg},~\er{pqzcr}  
with coefficients in~$\bl$ 
determines a homomorphism $\wea\to\bl$ given by the formulas 
$\wga_\al\mapsto P_\al,\ \wgb_\beta\mapsto Q_\beta$. 
\end{lemma}

Denote by $\zf$ the vector space of formal power series 
in the variables $z_1,\,z_2$ with coefficients in $\fd^0(\ce,a)$. That is, 
an element of $\zf$ is a power series of the form 
$$
\sum_{l_1,l_2\in\zp}z_1^{l_1}z_2^{l_2}C^{l_1l_2},\qquad\qquad C^{l_1l_2}\in\fd^0(\ce,a).
$$
The space $\zf$ has the Lie algebra structure given by 
$$
\bigg[\sum_{l_1,l_2}z_1^{l_1}z_2^{l_2}C^{l_1l_2},\,
\sum_{\tilde{l}_1,\tilde{l}_2}
z_1^{\tilde{l}_1}z_2^{\tilde{l}_2}
\tilde{C}^{\tilde{l}_1\tilde{l}_2}\bigg]=
\sum z_1^{l_1+\tilde{l}_1}z_2^{l_2+\tilde{l}_2}
\Big[C^{l_1l_2},\tilde{C}^{\tilde{l}_1\tilde{l}_2}\Big],\qquad\quad
C^{l_1l_2},\tilde{C}^{\tilde{l}_1\tilde{l}_2}\in\fd^0(\ce,a).
$$
We have also the following homomorphism of Lie algebras 
\beq
\lb{zffd0}
\nu\cl\zf\to\fd^0(\ce,a),\qquad\qquad
\sum_{l_1,l_2\in\zp}z_1^{l_1}z_2^{l_2}C^{l_1l_2}\,\mapsto\,C^{00}. 
\ee

For $i=1,2$, let $\pd_{z_i}\cl \zf\to\zf$ be the linear map given by 
$$
\pd_{z_i}\Big(\sum z_1^{l_1}z_2^{l_2}C^{l_1l_2}\Big)=
\sum\frac{\pd}{\pd z_i}\Big(z_1^{l_1}z_2^{l_2}\Big) C^{l_1l_2}.
$$ 
Let $\zd$ be the linear span of $\pd_{z_1},\,\pd_{z_2}$ in the vector space of linear maps $\zf\to\zf$. 
Since the maps $\pd_{z_1},\,\pd_{z_2}$ commute, 
the space~$\zd$ is a $2$-dimensional abelian Lie algebra 
with respect to the commutator of maps. 

Denote by $\xtf$ the vector space $\zd\oplus\zf$ with the following Lie algebra structure 
$$
[X_1+f_1,\,X_2+f_2]=X_1(f_2)-X_2(f_1)+[f_1,f_2],\qquad X_1,X_2\in \zd,
\qquad f_1,f_2\in\zf.
$$   
An element of $\xtf$ can be written as a sum of the following form 
$$
\big(y_1\pd_{z_1}+y_2\pd_{z_2}\big)+
\sum z_1^{l_1}z_2^{l_2}C^{l_1l_2},\qquad\quad 
y_1,y_2\in\fik,\qquad C^{l_1l_2}\in\fd^0(\ce,a).
$$

\begin{theorem}
\lb{thmhfd0}
Recall that the WE algebra $\wea$ is generated by 
$\wga_\al$, $\wgb_\beta$ for $\al\in\mat_0$, $\beta\in\mat_{\eo-1}$ 
such that one has formulas~\er{wgawgbser}, \er{wxgbtga}.

Let $\swe\subset\wea$ be the subalgebra generated by the elements
\beq
\lb{elmh}
(\ad\wga_0)^{k}(\wga_\al),\qquad\qquad k\in\zp,
\qquad\al\in\mat_0,\qquad\al\neq 0.
\ee
\textup{(}Note that for $k=0$ we have $(\ad\wga_0)^{0}(\wga_\al)=\wga_\al$, hence 
$\wga_\al\in\swe$ for all $\al\neq 0$.\textup{)}

Then the map $(\ad\wga_0)^{k}(\wga_\al)\,\mapsto\,(k!)\ga^{k,0}_\al$
determines an isomorphism between $\swe$ and $\fd^0(\CE,a)$.
\end{theorem}
\begin{proof}
Since the functions $F^i$ in~\er{gevxt} do not depend on $x$ and $t$, 
from~\er{gagab0},~\er{gagab0zcr} it follows that the power series 
\begin{gather}
\label{zgaser}
\tilde{\wga}=\pd_{z_1}+\sum_{\al\in\mat_0,\,\ \al\neq 0} \ua^\al\cdot
\bigg(\sum_{l_1,l_2}z_1^{l_1}z_2^{l_2}\ga^{l_1,l_2}_\al\bigg),\\
\lb{zgbser}
\tilde{\wgb}=\bigg(\pd_{z_2}+\sum_{l_1,l_2}z_1^{l_1}z_2^{l_2}\gb^{l_1,l_2}_0\bigg)
+\sum_{\beta\in\mat_{\eo-1},\,\ 
\beta\neq 0} \ua^\beta\cdot\bigg(\sum_{l_1,l_2}z_1^{l_1}z_2^{l_2}
\gb^{l_1,l_2}_\beta\bigg)
\end{gather}  
form a ZCR of Wahlquist-Estabrook type with coefficients in~$\xtf$. 
Applying Lemma~\ref{propwezcr} to this ZCR, we obtain the homomorphism  
\begin{gather}
\lb{weaxtfa}
\vf\cl\wea\to\xtf,\qquad 
\vf(\wga_0)=\pd_{z_1},\qquad\vf(\wga_\al)= 
\sum_{l_1,l_2}z_1^{l_1}z_2^{l_2}\ga^{l_1,l_2}_\al,\qquad
\al\in\mat_0,\quad\al\neq 0,
\\
\notag
\vf(\wgb_0)=
\pd_{z_2}+\sum_{l_1,l_2}z_1^{l_1}z_2^{l_2}\gb^{l_1,l_2}_0,
\qquad
\vf(\wgb_\beta)=
\sum_{l_1,l_2}z_1^{l_1}z_2^{l_2}
\gb^{l_1,l_2}_\beta,\qquad \beta\in\mat_{\eo-1},\quad\beta\neq 0. 
\end{gather}

Clearly, $\zf$ is a Lie subalgebra of $\xtf=\zd\oplus\zf$. 
In view of~\er{weaxtfa}, for any $\al\in\mat_0$ and $k\in\zp$ 
such that $\al\neq 0$ we have
\beq
\label{wga0z1}
\vf\Big((\ad\wga_0)^{k}(\wga_\al)\Big)= 
\big(\ad\pd_{z_1}\big)^{k}\bigg(\sum_{l_1,l_2}z_1^{l_1}z_2^{l_2}\ga^{l_1,l_2}_\al\bigg)
=\big(\pd_{z_1}\big)^{k}\bigg(\sum_{l_1,l_2}z_1^{l_1}z_2^{l_2}\ga^{l_1,l_2}_\al\bigg)\in\zf.
\ee 
Since $\swe\subset\wea$ is generated by the elements~\er{elmh}, 
property~\er{wga0z1} implies $\vf(\swe)\subset\zf\subset\xtf$. 
Using the homomorphism~\er{zffd0} and property~\er{wga0z1}, 
we obtain 
\beq
\lb{nuvf}
\nu\circ\vf\big|_{\swe}\cl\swe\to\fd^0(\CE,a),\quad 
(\nu\circ\vf)\Big((\ad\wga_0)^{k}(\wga_\al)\Big)
=k!\cdot\ga^{k,0}_\al,\quad 
k\in\zp,\,\ 
\al\in\mat_0,\,\ 
\al\neq 0.
\ee

Using Remark~\ref{remunxi}, we can assume that $\wea$ is embedded 
into the algebra $\gl(V)$ for some vector space~$V$. 

Then the exponentials $\mathrm{e}^{t\wgb_0}$, $\mathrm{e}^{x\wga_0}$ and
the expressions~\er{wgawgbser} can be regarded as power series with coefficients 
in $\gl(V)$. If $S_1,\,S_2$ are power series with coefficients 
in $\gl(V)$, then the product $S_1 S_2$ is a well-defined power 
series as well. It is easy to check that the following formulas are valid 
\begin{multline} 
\lb{eeaee}
\mathrm{e}^{t\wgb_0}\mathrm{e}^{x\wga_0}
\Big(D_x+\sum_{\al\in\mat_0} \ua^\al\cdot\wga_\al\Big)
\mathrm{e}^{-x\wga_0}\mathrm{e}^{-t\wgb_0}=\\
=D_x-\mathrm{e}^{t\wgb_0}\wga_0\mathrm{e}^{-t\wgb_0}+
\sum_{\al\in\mat_0}\ua^\al\cdot\mathrm{e}^{t\wgb_0}\mathrm{e}^{x\wga_0}
\wga_\al\mathrm{e}^{-x\wga_0}\mathrm{e}^{-t\wgb_0}=
\\
=D_x+\sum_{\al\in\mat_0,\ \al\neq 0}  \ua^\al\cdot\sum_{l_1,l_2}\frac{1}{l_1!l_2!}x^{l_1}t^{l_2}(\ad\wgb_0)^{l_2}\Big((\ad\wga_0)^{l_1}(\wga_\al)\Big),
\end{multline}
\begin{multline} 
\lb{eebee}
\mathrm{e}^{t\wgb_0}\mathrm{e}^{x\wga_0}
\Big(D_t+\sum_{\beta\in\mat_{\eo-1}} \ua^\beta\cdot\wgb_\beta\Big)
\mathrm{e}^{-x\wga_0}\mathrm{e}^{-t\wgb_0}=\\
=D_t-B_0+\sum_{\beta\in\mat_{\eo-1}}  
\ua^\beta\cdot\sum_{l_1,l_2}\frac{1}{l_1!l_2!}x^{l_1}t^{l_2}
(\ad\wgb_0)^{l_2}\Big((\ad\wga_0)^{l_1}(\wgb_\beta)\Big).
\end{multline}
Recall that $\wga$, $\wgb$ are given by~\er{wgawgbser}.
Set
\begin{gather}
\lb{plong}
P=\sum_{\al\in\mat_0,\ \al\neq 0}  
\ua^\al\cdot\sum_{l_1,l_2}\frac{1}{l_1!l_2!}x^{l_1}t^{l_2}
(\ad\wgb_0)^{l_2}\Big((\ad\wga_0)^{l_1}(\wga_\al)\Big),\\
\lb{qlong}
Q=-B_0+\sum_{\beta\in\mat_{\eo-1}}  
\ua^\beta\cdot\sum_{l_1,l_2}\frac{1}{l_1!l_2!}x^{l_1}t^{l_2}
(\ad\wgb_0)^{l_2}\Big((\ad\wga_0)^{l_1}(\wgb_\beta)\Big).
\end{gather}
According to \er{wgawgbser}, \er{eeaee}, 
\er{eebee}, \er{plong}, \er{qlong} we have 
\beq
\lb{pqdxt}
D_x+P=\mathrm{e}^{t\wgb_0}\mathrm{e}^{x\wga_0}
(D_x+\wga)\mathrm{e}^{-x\wga_0}\mathrm{e}^{-t\wgb_0},\qquad
D_t+Q=\mathrm{e}^{t\wgb_0}\mathrm{e}^{x\wga_0}
(D_t+\wgb)\mathrm{e}^{-x\wga_0}\mathrm{e}^{-t\wgb_0}.
\ee
Note that, since $[D_x,D_t]=0$, one has
\beq
\lb{dxabpq}
[D_x+\wga,D_t+\wgb]=D_x(\wgb)-D_t(\wga)+[\wga,\wgb],\qquad
[D_x+P,D_t+Q]=D_x(Q)-D_t(P)+[P,Q].
\ee
Using \er{pqdxt}, \er{dxabpq}, \er{wxgbtga}, we get
\begin{multline*}
D_x(Q)-D_t(P)+[P,Q]=[D_x+P,D_t+Q]=\\
=\big[\mathrm{e}^{t\wgb_0}\mathrm{e}^{x\wga_0}
(D_x+\wga)\mathrm{e}^{-x\wga_0}\mathrm{e}^{-t\wgb_0},\,
\mathrm{e}^{t\wgb_0}\mathrm{e}^{x\wga_0}
(D_t+\wgb)\mathrm{e}^{-x\wga_0}\mathrm{e}^{-t\wgb_0}\big]=\\
=\mathrm{e}^{t\wgb_0}\mathrm{e}^{x\wga_0}
[D_x+\wga,D_t+\wgb]\mathrm{e}^{-x\wga_0}\mathrm{e}^{-t\wgb_0}=
\mathrm{e}^{t\wgb_0}\mathrm{e}^{x\wga_0}
(D_x(\wgb)-D_t(\wga)+[\wga,\wgb])\mathrm{e}^{-x\wga_0}\mathrm{e}^{-t\wgb_0}=0.
\end{multline*}
Therefore, the power series \er{plong}, \er{qlong} 
satisfy all conditions of Lemma~\ref{lemfd0zcr}. 
Applying Lemma~\ref{lemfd0zcr} to~\er{plong},~\er{qlong}, we obtain 
the homomorphism 
\begin{gather}
\lb{psiga}
\psi\cl\fd^0(\CE,a)\to\wea,\qquad
\psi\big(\ga^{l_1,l_2}_\al\big)=\frac{1}{l_1!l_2!}
(\ad\wgb_0)^{l_2}\Big((\ad\wga_0)^{l_1}(\wga_\al)\Big),
\qquad\al\in\mat_0,\qquad\al\neq 0,\\ 
\notag
\psi\big(\gb^{l_1,l_2}_\beta\big)=
\frac{1}{l_1!l_2!}(\ad\wgb_0)^{l_2}
\Big((\ad\wga_0)^{l_1}(\wgb_\beta)\Big),\qquad\beta\in\mat_{\eo-1},
\qquad\beta\neq 0,\\
\notag
\psi\big(\gb^{l'_1,l'_2}_0\big)=\frac{1}{l'_1!l'_2!}(\ad\wgb_0)^{l'_2}
\Big((\ad\wga_0)^{l'_1}(\wgb_0)\Big),\qquad l'_1>0.
\end{gather}
From~\er{psiga} we get 
\beq
\lb{psigal1}
\psi\big(\ga^{l_1,0}_\al\big)=\frac{1}{l_1!}
(\ad\wga_0)^{l_1}(\wga_\al)\in\swe,\qquad l_1\in\zp,
\qquad\al\in\mat_0,\qquad\al\neq 0.
\ee
Since, by Lemma~\ref{lemgenfd0}, the elements~\er{gal1al} generate 
the algebra $\fd^0(\CE,a)$, property~\er{psigal1} implies 
$\psi\big(\fd^0(\CE,a)\big)\subset\swe$. 

Then from~\er{nuvf},~\er{psigal1} it follows that the 
homomorphisms $\psi\cl\fd^0(\CE,a)\to\swe$ and 
$\nu\circ\vf\big|_{\swe}\cl\swe\to\fd^0(\CE,a)$ are inverse to each other.
\end{proof}

\section{The algebra $\fd^0(\CE,a)$ for the multicomponent  Landau-Lifshitz system}
\lb{seclnfd0}

For any $m\in\zsp$ and $m$-dimensional vectors 
$v={(v^1,\dots,v^m)}$, $w={(w^1,\dots,w^m)}$, we set 
$\langle v,w\rangle=\sum_{i=1}^mv^iw^i$.

In order to describe the algebra $\fd^0(\CE,a)$ for system~\er{main},
we need to resolve the constraint $\langle S,S\rangle=1$
for the vector-function $S=\big(s^1(x,t),\dots,s^n(x,t)\big)$.
Following~\cite{mll}, we do this as follows
\begin{equation}
\label{sp}
s^j=\frac{2u^j}{1+\langle u,u\rangle},\qquad\qquad
j=1,\dots,n-1,\qquad\qquad
s^n=\frac{1-\langle u,u\rangle}{1+\langle u,u\rangle},
\end{equation}
where $u=\big(u^1(x,t),\dots,u^{n-1}(x,t)\big)$ 
is an $(n-1)$-dimensional vector-function.

As is shown in~\cite{mll}, using~\er{sp}, 
one can rewrite system~\er{main} as follows
\begin{gather}
\label{pt}
\begin{split}
u_t&=u_{xxx}-6\langle u,u_x\rangle\Delta^{-1}u_{xx}
+\bigl(-6\langle u,u_{xx}\rangle\Delta^{-1}+24\langle u,u_{x}\rangle^2\Delta^{-2}
-6\langle u,u\rangle\langle u_x,u_{x}\rangle\Delta^{-2}\bigl)u_x+\\
&+\bigl(6\langle u_x,u_{xx}\rangle\Delta^{-1}-12\langle u,u_x\rangle\langle u_x,u_{x}\rangle\Delta^{-2}\bigl)u
+\frac32\Bigl(r_n+4\Delta^{-2}\sum_{i=1}^{n-1}(r_i-r_n)(u^i)^2\Bigl)u_x,
\end{split}\\
\notag
u=\big(u^1(x,t),\dots,u^{n-1}(x,t)\big),
\end{gather}
where $\Delta=1+\langle u,u\rangle$ and $r_1,\dots,r_n\in\fik$ 
are the numbers such that 
$R=\mathrm{diag}\,(r_1,\dots,r_n)$ in~\er{main}.
As has been said in Section~\ref{subsmr}, 
we assume $r_i\neq r_j$ for all $i\neq j$.

In this section we assume $n\ge 3$, because we will use some 
results of the paper~\cite{mll-2012}, which studied equations~\er{main},~\er{pt} 
in the case $n\ge 3$.



Let $\CE$ be the infinite prolongation of the PDE~\er{pt}.
Then $\CE$ is a manifold with the coordinates $x$, $t$, $u^i_k$ 
for $i=1,\dots,n-1$ and $k\in\zp$. 
(Recall that $u^i_0=u^i$, according to our notation.)
Let $a\in\CE$.

According to Theorem~\ref{thmhfd0}, the algebra $\fd^0(\CE,a)$ 
for~\er{pt} is isomorphic to a subalgebra of the WE algebra $\wea$ of~\er{pt}. 
The WE algebra of~\er{pt} is described in~\cite{mll-2012}. 
To present this description, 
we need to introduce some auxiliary constructions.  

Recall that $\mathfrak{gl}_{n+1}$ is the algebra of matrices of 
size $(n+1)\times(n+1)$ with entries from~$\fik$.
Let $E_{i,j}\in\mathfrak{gl}_{n+1}$ be the matrix with 
$(i,j)$-th entry equal to 1 and all other entries equal to zero. 

The Lie subalgebra $\mathfrak{so}_{n,1}\subset\mathfrak{gl}_{n+1}$ 
has been defined in Remark~\ref{son1}. 
It has the following basis
$$
E_{i,j}-E_{j,i},\qquad i<j\le n,\qquad\qquad 
E_{l,n+1}+E_{n+1,l},\qquad l=1,\dots,n.
$$

Consider the algebra $\fik[\la_1,\dots,\la_n]$ 
of polynomials in $\la_1,\dots,\la_n$. 
Let $\ipol\subset\fik[\la_1,\dots,\la_n]$ be the ideal  
generated by $\la_i^2-\la_j^2+r_i-r_j$ for $i,j=1,\dots,n$. 
 
Consider the quotient algebra $\qalg=\fik[\la_1,\dots,\la_n]/\ipol$, which 
is isomorphic to the 
algebra of polynomial functions on the algebraic curve~\er{curve}. 

The space $\mathfrak{so}_{n,1}\otimes_\fik \qalg$ 
is an infinite-dimensional Lie algebra over $\fik$ with the Lie bracket 
$$
[M_1\otimes h_1,\,M_2\otimes h_2]=[M_1,M_2]\otimes h_1h_2,
\qquad\qquad 
M_1,M_2\in \mathfrak{so}_{n,1},\qquad\qquad h_1,h_2\in \qalg.
$$

We have the natural homomorphism 
$\xi\cl\fik[\la_1,\dots,\la_n]\to\fik[\la_1,\dots,\la_n]/\ipol=\qalg$. 
Set $\hat\la_i=\xi(\la_i)\in \qalg$. 

Consider the following elements of ${\mathfrak{so}_{n,1}\otimes \qalg}$
\begin{equation}
\label{qie}
Q_i=(E_{i,n+1}+E_{n+1,i})\otimes\hat\la_i,
\qquad\qquad i=1,\dots,n.
\end{equation}
Denote by $L(n)\subset\mathfrak{so}_{n,1}\otimes \qalg$ the Lie subalgebra  
generated by~$Q_1,\dots,Q_n$. 

Since $\hat{\la}_i^2-\hat\la_j^2+r_i-r_j=0$ in $\qalg$, the element 
$\hat\la=\hat\la_i^2+r_i\in \qalg$ does not depend on $i$.  

Recall that $\zsp$ is the set of positive integers.
For $i,j\in\{1,\dots,n\}$ and $k\in\zsp$, 
consider the following elements of 
${\mathfrak{so}_{n,1}\otimes_\fik \qalg}$  
$$
Q^{2k-1}_i=(E_{i,n+1}+E_{n+1,i})\otimes\hat\la^{k-1}\hat\la_i,\qquad\qquad
Q^{2k}_{ij}=(E_{i,j}-E_{j,i})\otimes\hat\la^{k-1}\hat\la_i\hat\la_j. 
$$
For $i,j,l,m\in\{1,\dots,n\}$ and $k_1,k_2\in\zsp$ one has 
\begin{multline} 
\label{q1}
[Q^{2k_1}_{ij},\,Q^{2k_2}_{lm}]=\delta_{lj}Q^{2(k_1+k_2)}_{im} 
-\delta_{im}Q^{2(k_1+k_2)}_{lj}
+\delta_{jm}Q^{2(k_1+k_2)}_{li}-\delta_{il}Q^{2(k_1+k_2)}_{jm}+\\
+r_i\delta_{im}Q^{2(k_1+k_2-1)}_{lj}
-r_j\delta_{lj}Q^{2(k_1+k_2-1)}_{im}+r_i\delta_{il}Q^{2(k_1+k_2-1)}_{jm} 
-r_j\delta_{jm}Q^{2(k_1+k_2-1)}_{li},  
\end{multline}
\begin{equation} 
\label{q2}
[Q^{2k_1}_{ij},\,Q^{2k_2-1}_{l}]=\delta_{lj}Q^{2k_1+2k_2-1}_{i}
-\delta_{il}Q^{2k_1+2k_2-1}_{j}-r_j\delta_{lj}Q^{2k_1+2k_2-3}_{i}+r_i\delta_{il}Q^{2k_1+2k_2-3}_{j},
\end{equation}
\begin{equation} 
\label{q3}
[Q^{2k_1-1}_{i},\,Q^{2k_2-1}_{j}]=Q^{2(k_1+k_2-1)}_{ij},\qquad\qquad 
[Q^{2k_1-1}_{i},\,Q^{2k_2-1}_{i}]=0.   
\end{equation}

Since $Q^1_i=Q_i$ and $Q^{2k}_{ij}=-Q^{2k}_{ji}$, 
from~\er{q1},~\er{q2},~\er{q3} we see that the elements 
\beq
\lb{qlnelem}
Q^{2k-1}_l,\quad\qquad Q^{2k}_{ij},\quad\qquad i,j,l\in\{1,\dots,n\},\quad\qquad 
i<j,\quad\qquad k\in\zsp,
\ee
span the Lie algebra $L(n)$. 
It is shown in~\cite{mll-2012} that  
the elements~\er{qlnelem} are linearly independent 
over $\fik$ and, therefore, form a basis of $L(n)$. 

\begin{remark}
\lb{rem2dln}
In Remark~\ref{mllzcr} we have said that $L(n)$ consists of certain 
$\mathfrak{gl}_{n+1}$-valued functions on the curve~\er{curve}, 
and in Remark~\ref{son1} we have shown that these functions take values in 
the Lie subalgebra $\mathfrak{so}_{n,1}\subset\mathfrak{gl}_{n+1}$.

Here we have defined $L(n)$ as a certain Lie subalgebra 
of $\mathfrak{so}_{n,1}\otimes \qalg$.
This is in agreement with Remarks~\ref{mllzcr},~\ref{son1}, because elements 
of $\mathfrak{so}_{n,1}\otimes \qalg$ can be regarded as 
$\mathfrak{so}_{n,1}$-valued functions on the curve~\er{curve}.
\end{remark}

Note that the algebra $L(n)$ is very similar to infinite-dimensional 
Lie algebras that were studied in~\cite{skr,skr-jmp}. 

\begin{theorem}
\label{thfd0}
Let $n\ge 3$. 
For the PDE~\er{pt},  
the Lie algebra $\fd^0(\CE,a)$ defined in~\er{fdnn-1} is isomorphic to~$L(n)$. 
\end{theorem}
\begin{proof}
Let $\wea$ be the WE algebra for the PDE~\er{pt}. 
According to~\cite{mll-2012}, 
the algebra~$\wea$ is isomorphic to the direct sum of~$L(n)$ 
and the $2$-dimensional abelian Lie algebra $\fik^2$.
So $\wea\cong L(n)\oplus\fik^2$.
(Note that the algebra~$\wea$ is denoted by $\mathfrak{W}$ in~\cite{mll-2012}.)

According to~\cite{mll-2012}, 
in the formal Wahlquist-Estabrook ZCR with coefficients in~$\wea$ 
for the PDE~\er{pt}, one has 
\begin{equation}
\notag
\wga=C_0+\sum_{l=1}^{n}C_l\cdot s^l(u^1,\dots,u^{n-1}),
\qquad\quad C_0,C_1,\dots,C_n\in\wea,
\end{equation}
where the functions $s^l=s^l(u^1,\dots,u^{n-1})$ are given by~\er{sp}, 
the elements $C_1,\dots,C_n\in\wea$ generate the Lie subalgebra 
$L(n)\subset\wea\cong L(n)\oplus\fik^2$, and one has 
$[C_0,C_l]=0$ for all $l=1,\dots,n$.

This implies that the subalgebra $\swe\subset\wea$ defined in Theorem~\ref{thmhfd0} 
is equal to $L(n)\subset\wea$.
According to Theorem~\ref{thmhfd0}, one has $\fd^0(\CE,a)\cong\swe$. 
Since in the considered case we have $\swe=L(n)$, we get $\fd^0(\CE,a)\cong L(n)$.
\end{proof}

\section{The algebra $\fd^1(\CE,a)$ for the multicomponent Landau-Lifshitz system}
\lb{secfd1}


\subsection{Preliminary computations}
\label{prcomp}

We continue to use the notation $u^i_k=\dfrac{\pd^k u^i}{\pd x^k}$.

Let $\CE$ be the infinite prolongation of the PDE~\er{pt}.
Recall that $u=(u^1,\dots,u^{n-1})$ is an $(n-1)$-dimensional vector in~\er{pt}.
Then $\CE$ is an infinite-dimensional manifold with the coordinates
\beq
\lb{xtuikn1}
x,\quad t,\quad u^i_k,\qquad i=1,\dots,n-1,\quad k\in\zp,\quad u^i_0=u^i.
\ee

Consider an arbitrary point $a\in\CE$ given by
\begin{equation}
\label{pointevfd1}
a=(x=x_a,\ t=t_a,\ u^i_k=a^i_k)\,\in\,\CE,\qquad x_a,\,t_a,\,a^i_k\in\fik,\qquad
i=1,\dots,n-1,\qquad k\in\zp. 
\end{equation}

Since the PDE~\er{pt} is invariant with respect to the change of variables 
$x\mapsto x-x_a,\ t\mapsto t-t_a$, it is sufficient to consider the case 
\beq
\label{x0t0=0}
x_a=t_a=0.
\ee 
For simplicity of exposition, we assume also 
\beq
\label{aik=0}
a^i_k=0\qquad\forall\,i,k.  
\ee
(In the case $a^i_k\neq 0$, the computations change very little, 
and the final result is the same.)


According to Remark~\ref{rem_fdpgen} and assumptions~\er{x0t0=0},~\er{aik=0}, 
in order to describe the Lie algebra $\fd^1(\CE,a)$ for the PDE~\er{pt}, 
we need to study the equations
\beq
\lb{gc}
D_x(B)-D_t(A)+[A,B]=0,
\ee
\begin{gather}
\label{gd=0f1}
\forall\,i_0=1,\dots,n-1,\qquad\quad
\frac{\pd A}{\pd u^{i_0}_1}
\,\,\bigg|_{u^j_1=0\ \forall\,j,\ u^i_0=0\ \forall\,i>i_0}=0,\\
\lb{gaukakf1}
A\,\Big|_{u^j_0=0,\ u^j_1=0\ \forall\,j}=0,\\
\lb{gbxx0f1}
B\,\Big|_{x=0,\ u^j_k=0\ \forall\,j,\ \forall\,k\ge 0}=0,
\end{gather}
where 
\begin{itemize}
\item $A=A(x,t,u^j_0,u^j_1)$ is a power series in the variables 
$x$, $t$, $u^j_0$, $u^j_1$ for $j=1,\dots,n-1$,
\item $B=B(x,t,u^j_0,u^j_1,u^j_2,u^j_3)$ is a power series in the variables 
$x$, $t$, $u^j_0$, $u^j_1$, $u^j_2$, $u^j_3$ for $j=1,\dots,n-1$. 
\end{itemize}
The coefficients of the power series $A$, $B$ 
are generators of the Lie algebra $\fd^1(\CE,a)$.
Relations for these generators are provided by equations \er{gc}, \er{gd=0f1},
\er{gaukakf1}, \er{gbxx0f1}.

Since, according to our notation, $u^j_0=u^j$, below we sometimes write $u^j$ instead 
of $u^j_0$. In particular, we can write $A=A(x,t,u^j,u^j_1)$.

When we consider power series in the variables $x$, $t$, $u^j$, $u^j_k$, 
partial derivatives with respect to these variables are often denoted by subscripts.
For example,
$$
A_{u^i_1}=\frac{\pd A}{\pd u^i_1},\qquad
A_{u^i_1 x}=\frac{\pd^2 A}{\pd u^i_1\pd x},\qquad
A_{u^i_1 u^j}=\frac{\pd^2 A}{\pd u^i_1\pd u^j},\qquad
i,j=1,\dots,n-1.
$$




%
%

Differentiating equation~\er{gc} with respect to~$u^i_4$,
we obtain that $B$ is of the form
\begin{equation}
\label{eq.T.general}
B= \sum_i u^i_{3} A_{u^i_1}  + F^1(x,t,u^j,u^j_1,u^j_{2}),
\end{equation}
where $F^1$ is a power series in the variables $x$, $t$, $u^j$, $u^j_1$, $u^j_{2}$.

According to our notation, the symbols $u_t$, $u$, $u_x$, $u_{xx}$, $u_{xxx}$ 
in~\er{pt} denote the following vectors
\begin{gather*}
u_t=(u^1_t,\dots,u^{n-1}_t),\qquad
u=(u^1,\dots,u^{n-1})=(u^1_0,\dots,u^{n-1}_0),\\
u_x=(u^1_1,\dots,u^{n-1}_1),\qquad
u_{xx}=(u^1_2,\dots,u^{n-1}_2),\qquad
u_{xxx}=(u^1_3,\dots,u^{n-1}_3).
\end{gather*}
Therefore, system \er{pt} can be written as
\begin{equation}\label{eq.landau}
u^i_t=u^i_{3}+G^i(u^j,u^j_1,u^j_{2}),\qquad\qquad i=1,\dots,n-1,
\end{equation}
for some functions $G^i(u^j,u^j_1,u^j_{2})$ determined by the right-hand side of \er{pt}.
Then one has
\begin{gather*}
D_x(B)=\sum_i\Big(u^i_{4} A_{u^i_1}+u^i_{3} \big(A_{u^i_1 x} 
+\sum_j\big(u^j_1A_{u^i_1 u^j} + u^j_{2}A_{u^i_1 u^j_1}\big)\big)\Big) 
+F^1_x+\sum_i\big(u^i_1  F^1_{u^i} + u^i_{2} F^1_{u^i_1}  + u^i_{3} F^1_{u^i_{2}}\big),\\
D_t(A)=A_t+\sum_i\Big(u^i_{3} A_{u^i} + G^i A_{u^i}+ u^i_{4}A_{u^i_1} 
+\sum_j\big(u^j_1G^i_{u^j}A_{u^i_1}+u^j_{2}G^i_{u^j_1}A_{u^i_1}+
u^j_3G^i_{u^j_{2}}A_{u^i_1}\big)\Big),\\
[A,B]=\sum_iu^i_{3}[A,A_{u^i_1}] + [A,F^1].
\end{gather*}
Differentiating equation~\er{gc} with respect to~$u^i_3$, we get 
\begin{equation}
\lb{f1ui2a}
F^1_{u^i_{2}}=A_{u^i}+\sum_jG^j_{u^i_{2}}A_{u^j_1}-A_{u^i_1 x}-\sum_j u^j_1A_{u^i_1 u^j}
-\sum_j u^j_{2}A_{u^i_1 u^j_1}-[A,A_{u^i_1}],\qquad i=1,\dots,n-1.
\end{equation}
Since $G^j_{u^i_{2}u^l_{2}}=0$ for all $i$, $l$, 
from \er{f1ui2a} one obtains that $F^1$ is of the form
\begin{multline}\label{eq.F1}
F^1=-\frac12\sum_{i,j}u^i_{2}u^j_{2}A_{u^i_1u^j_1}+\\
+\sum_i u^i_{2}\Big(A_{u^i}+\sum_j G^j_{u^i_{2}}A_{u^j_1}-A_{u^i_1 x}
-\sum_j u^j_1A_{u^i_1 u^j}-[A,A_{u^i_1}]\Big)+F^2(x,t,u^l,u^l_1),
\end{multline}
where $F^2$ is a power series in the variables $x$, $t$, $u^l$, $u^l_1$. 
Then equation~\er{gc} becomes
\begin{equation}\label{eqq.covering}
F^1_x+\sum_i u^i_1F^1_{u^i}+\sum_i u^i_{2}F^1_{u^i_1}-A_t-\sum_i G^i A_{u^i}
-\sum_{i,j}u^j_1G^i_{u^j}A_{u^i_1}-\sum_{i,j}u^j_{2}G^i_{u^j_1}A_{u^i_1}+[A,F^1]=0.
\end{equation}

Differentiating~\er{eqq.covering} with respect to~$u^i_{2},\,u^j_{2},\,u^h_{2}$ 
and taking into account~\er{eq.F1}, 
one gets $A_{u^i_1 u^j_1 u^h_1}=0$ for all $i,j,h$. That is, $A$ is of the form 
\begin{equation}
\label{eq.X.polynomial.second.order}
A=\frac{1}{2}\sum_{i,j}u^i_1 u^j_1 Y_{ij}+\sum_i u^i_1 Y_i+Y,\qquad Y_{ij}=Y_{ji},
\end{equation}
where $Y_{ij}$, $Y_i$, $Y$ are power series in the variables $x$, $t$, $u^1,\dots,u^{n-1}$. 

From the definition of $G^i$, for all $h,i,j=1,\dots,n-1$ one has 
\begin{equation}\label{eq.formule.Gpxx}
G^h_{u^i_{2}} = -6\delta^h_i\Delta^{-1}\langle u , u_1\rangle - 6\Delta^{-1} u^iu^h_1 + 6\Delta^{-1}u^i_1u^h,
\end{equation}
\begin{multline}
\label{guxjh}
G_{u_{1}^{j}}^{h}=  -6\Delta^{-1}u^{j}u_{2}^{h}-6\delta_{j}^{h}\Delta
^{-1}\langle u , u_{2}\rangle + 6\Delta^{-1}u_{2}^{j}u^{h} \\
 +48\Delta^{-2}\langle u , u_{1}\rangle u^{j}u_{1}^{h} + 24\delta_{j}^{h}\Delta^{-2}\langle u , u_{1}\rangle^{2}-12\Delta^{-2}\langle u , u\rangle u_{1}^{j}u_{1}^{h} - 6\delta_{j}^{h}\Delta^{-2}\langle u , u\rangle \langle u_{1}, u_{1}\rangle \\
 -12\Delta^{-2}u^{h}u^{j}\langle u_{1}, u_{1}\rangle - 24\Delta^{-2}\langle u , u_{1}\rangle u^{h}u_{1}^{j}+\frac{3}{2}\delta_{j}^{h}\left( r_{n}+
4\Delta^{-2}\sum_i(r_{i}-r_{n}){u^{i}}^{2}\right),
\end{multline}
where $u_k$ is the vector $(u^1_k,\dots,u^{n-1}_k)$ for $k\in\zp$.

Differentiating equation~\er{eqq.covering} with respect to $u^i_{2}$, $u^j_{2}$ 
and taking into account~\er{eq.F1}, \er{eq.formule.Gpxx}, \er{guxjh}, we obtain 
\begin{multline}\label{eq.coeff.pixx.pjxx.0}
A_{u^i_1 u^j_1 x}+\sum_h u^h_1 A_{u^i_1 u^j_1 u^h} + [A,A_{u^i_1 u^j_1}] 
+4\Delta^{-1}A_{u^i_1u^j_1}\langle u,u_1\rangle+2\sum_h u^iu^h_1\Delta^{-1}A_{u^h_1u^j_1}+
\\
+2\sum_h u^ju^h_1\Delta^{-1}A_{u^h_1u^i_1}
-2\sum_h u^i_1u^h\Delta^{-1}A_{u^h_1u^j_1}-2\sum_h u^j_1u^h\Delta^{-1}A_{u^h_1u^i_1}=0. 
\end{multline}
Substituting~\eqref{eq.X.polynomial.second.order} to~\er{eq.coeff.pixx.pjxx.0}, one gets
\begin{multline}\label{eq.coeff.pixx.pjxx}
Y_{ij,x}+\sum_h u^h_1 Y_{ij,h}+[A,Y_{ij}]+\\
+\Delta^{-1}\Big(4\langle u, u_1\rangle Y_{ij}+2\sum_h u^iu^h_1 Y_{hj}+
2\sum_h u^ju^h_1 Y_{hi}
-2\sum_h u^i_1u^h Y_{hj}-2\sum_h u^j_1u^h Y_{hi}\Big)=0.
\end{multline}
Here and below we use the following notation
\begin{gather*}
Y_{ij,x}=\dfrac{\partial Y_{ij}}{\partial x},\qquad
Y_{ij,h}=\dfrac{\partial Y_{ij}}{\partial u^h},\qquad
Y_{ij,hk}=\dfrac{\partial^2 Y_{ij}}{\partial u^h\partial u^k},
\qquad i,j,h,k=1,2,\dots,n-1,\\
Y_{h,x}=\frac{\partial Y_{h}}{\partial x},\qquad
Y_{i,j}=\frac{\partial Y_{i}}{\partial u^j},\qquad
Y_{,x}=\frac{\partial Y}{\partial x},\qquad
Y_{,h}=\frac{\partial Y}{\partial u^h}.
\end{gather*}

The left-hand side of~\er{eq.coeff.pixx.pjxx} is a polynomial of degree $\le 2$ 
with respect to the variables $u^1_1,\dots,u^{n-1}_1$.
Equating to zero the coefficients of this polynomial, we obtain
\begin{gather}\label{rel.1}
[Y_{hk},Y_{ij}]=0\qquad \forall\, h,k,i,j,\\
\label{system.defining.Yij}
Y_{ij,h}=  [Y_{ij},Y_h] - 2\Delta^{-1}\big(2u^h Y_{ij}+ u^i Y_{hj} 
+ u^j Y_{hi} - \delta^i_h \sum_m u^m Y_{mj} - \delta^j_h \sum_m u^m Y_{mi}\big)
\qquad\forall\,i,j,h,\\
\label{yijx}
Y_{ij,x} = [Y_{ij},Y]\qquad\forall\,i,j.
\end{gather}
Let us rewrite \er{system.defining.Yij} replacing $h$ by $k$
\beq
\lb{yijkyy}
Y_{ij,k}=  [Y_{ij},Y_k] - 2\Delta^{-1}\big(2u^k Y_{ij}+ u^i Y_{kj} 
+ u^j Y_{ki} - \delta^i_k \sum_m u^m Y_{mj} - \delta^j_k \sum_m u^m Y_{mi}\big)
\qquad\forall\,i,j,k.
\ee
Now we differentiate \er{system.defining.Yij} with respect to $u^k$ 
and differentiate \er{yijkyy} with respect to $u^h$.
Then the equality $Y_{ij,hk}=Y_{ij,kh}$ implies 
\begin{equation}\label{eq.comp.cond.system.defining.Yij}
[Y_{ij},T_{kh}] + 4\Delta^{-2}( -\delta^i_kY_{hj} + \delta^i_hY_{kj} - \delta^j_kY_{hi} + \delta^j_hY_{ki} )=0
\qquad\forall\, i,j,k,h,
\end{equation}
where 
\beq\label{eq.def.Tkh}
T_{kh}=[Y_k,Y_h] - Y_{k,h} + Y_{h,k}.
\ee
In other words, 
we regard~\er{system.defining.Yij} as an overdetermined system of PDEs for $Y_{ij}$, 
and equations~\er{eq.comp.cond.system.defining.Yij} have been obtained from the
compatibility condition of system~\er{system.defining.Yij}.

\begin{remark}
Another way to obtain~\er{eq.comp.cond.system.defining.Yij} 
from~\er{system.defining.Yij} is the following. Set
$$
V_k=Y_k+\frac{\partial}{\partial u^k},\qquad\quad k=1,\dots,n-1.
$$
Then~\er{system.defining.Yij} and~\er{eq.def.Tkh} can be rewritten as follows
\begin{gather}
\lb{vhyij}
[V_h,Y_{ij}]=- 2\Delta^{-1}\big(2u^h Y_{ij}+ u^i Y_{hj} 
+ u^j Y_{hi} - \delta^i_h \sum_m u^m Y_{mj} - \delta^j_h \sum_m u^m Y_{mi}\big),\\
\lb{tkhvkvh}
T_{kh}=[V_k,V_h].
\end{gather}
Using the Jacobi identity and equation~\er{tkhvkvh}, we get 
\beq
\lb{vvyt}
[V_k,[V_h,Y_{ij}]]=[[V_k,V_h],Y_{ij}]+[V_h,[V_k,Y_{ij}]]=
[T_{kh},Y_{ij}]+[V_h,[V_k,Y_{ij}]].
\ee

Let $P$ be a power series in the variables $x,t,u^1,\dots,u^{n-1}$ with 
coefficients in the Lie algebra $\fd^1(\CE,a)$. Then the operator 
$\ad V_k$ can be applied to $P$ in the standard way 
$$
(\ad V_k)(P)=[V_k,P]=\Big[Y_k+\frac{\partial}{\partial u^k},\,P\Big]=
[Y_k,P]+\frac{\partial}{\partial u^k}(P).
$$

Applying the operator $\ad V_k$ to equation~\er{vhyij}, 
for all $k$, $h$ we can express $[V_k,[V_h,Y_{ij}]]$ as a linear combination of 
$Y_{{\ai}{\bi}}$, ${\ai},{\bi}=1,\dots,n-1$. 
Then, exchanging $k$ and $h$, we can express $[V_h,[V_k,Y_{ij}]]$ 
as a linear combination of $Y_{{\ai}{\bi}}$, ${\ai},{\bi}=1,\dots,n-1$. 
Substituting the obtained expressions in~\er{vvyt}, 
one gets~\er{eq.comp.cond.system.defining.Yij}.
\end{remark}

Combining equations~\er{gd=0f1},~\er{gaukakf1} 
with formula~\er{eq.X.polynomial.second.order}, we obtain
\begin{gather}
\label{condyi}
\forall\,i=1,\dots,n-2,\qquad Y_i\,\Big|_{u^{i+1}=u^{i+2}=\dots=u^{n-1}=0}=0,\qquad Y_{n-1}=0,\\ 
\label{condy0}
Y\,\Big|_{u^{1}=u^{2}=\dots=u^{n-1}=0}=0.
\end{gather}

In order to study the obtained equations, 
we need the following lemmas on formal power series, 
which can be proved straightforwardly by induction on the degrees of the coefficients 
of these power series. 

\begin{lemma}
\label{zi}
Let $Z_1,\dots,Z_{n-1}$ be formal power series in some variables $v^1,\dots,v^{n-1}$ 
with coefficients in a Lie algebra $\mg$. 
Suppose that 
\beq 
\label{zivi1}
\forall\,i=1,\dots,n-2,\qquad Z_i\,\Big|_{v^{i+1}=v^{i+2}=\dots=v^{n-1}=0}=0,\qquad 
Z_{n-1}=0. 
\ee
Set $V_{ij}=[Z_i,Z_j] - \dfrac{\pd Z_i}{\pd v^j} + \dfrac{\pd Z_j}{\pd v^i}$. 

Then the Lie subalgebra generated 
by the coefficients of the power series $Z_1,\dots,Z_{n-1}$
coincides with the Lie subalgebra generated 
by the coefficients of the power series $V_{ij},\ i,j=1,\dots,n-1$.
\end{lemma}

\begin{lemma}
\label{zpsi}
Let $Z_1,\dots,Z_{n-1}$ and $\Psi_1,\dots,\Psi_m$ 
be formal power series in some variables $v^1,\dots,v^{n-1}$ 
with coefficients in a Lie algebra $\mg$. 

Suppose that~\er{zivi1} holds and 
\begin{equation} 
\label{psivi}
\frac{\pd \Psi_{l}}{\pd v^h}=  [\Psi_{l},Z_h]+\sum_{i=1}^m f^i_l(v^1,\dots,v^{n-1})\Psi_i 
\qquad\forall\,l,h,
\end{equation}
for some power series $f^i_l(v^1,\dots,v^{n-1})$ with coefficients in $\fik$. 
Consider the coefficients of the power series $\Psi_l$ 
\beq
\notag 
\Psi_l=\sum_{i_1,\dots,i_{n-1}\ge 0}\big(v^1\big)^{i_1}\dots\big(v^{n-1}\big)^{i_{n-1}}\psi^l_{i_1,\dots,i_{n-1}},
\ \quad \psi^l_{i_1,\dots,i_{n-1}}\in\mg,\quad l=1,\dots,n-1.
\ee
Then any coefficient $\psi^l_{i_1,\dots,i_{n-1}}$ belongs to the vector subspace spanned 
by the zero degree coefficients $\psi^1_{0,\dots,0},\dots,\psi^{n-1}_{0,\dots,0}$.
\end{lemma}

Recall that $Y_{ij}$ from formula~\er{eq.X.polynomial.second.order} 
are power series in the variables $x,\,t,\,u^1,\dots,u^{n-1}$ with coefficients in the Lie algebra $\fd^1(\CE,a)$. 
Let $\mg$ be the Lie algebra of formal power series in the variables $x,\,t$ with coefficients in $\fd^1(\CE,a)$. 
Then $Y_{ij}$ can be regarded as a power series in the variables $u^1,\dots,u^{n-1}$ with coefficients in $\mg$ 
\beq
\lb{yijyxt}
Y_{ij}=\sum_{i_1,\dots,i_{n-1}\ge 0}\big(u^1\big)^{i_1}\dots\big(u^{n-1}\big)^{i_{n-1}}y^{ij}_{i_1,\dots,i_{n-1}}(x,t), 
\ \quad y^{ij}_{i_1,\dots,i_{n-1}}(x,t)\in\mg.
\ee

Note that equations~\er{system.defining.Yij} are of the type~\er{psivi} for $v^h=u^h$ and $Z_h=Y_h$. 
Therefore, by Lemma~\ref{zpsi}, we obtain the following result. 
\begin{lemma}
\label{yyab}
For any $i,j,i_1,\dots,i_{n-1}$, 
the power series $y^{ij}_{i_1,\dots,i_{n-1}}(x,t)$ belongs to the vector subspace spanned 
by the power series $y^{{\ai}{\bi}}_{0,\dots,0}(x,t),\ {\ai},{\bi}=1,\dots,n-1$.
\end{lemma}
From~\er{yijx},~\er{condy0},~\er{yijyxt} one gets 
\begin{multline}
\label{yijx0}
\frac{\pd}{\pd x}\Big(y^{{\ai}{\bi}}_{0,\dots,0}(x,t)\Big)=
Y_{{\ai}{\bi},x}\,\Big|_{u^{1}=u^{2}=\dots=u^{n-1}=0} =\\
= \Big[Y_{{\ai}{\bi}}\,\Big|_{u^{1}=u^{2}=\dots=u^{n-1}=0},\,
Y\,\Big|_{u^{1}=u^{2}=\dots=u^{n-1}=0}\Big]=0\qquad\forall\,{\ai},{\bi}.
\end{multline}
Combining~\er{yijx0} with Lemma~\ref{yyab}, we obtain 
\beq
\lb{yabx0}
Y_{{\ai}{\bi},x}=0\qquad\quad\forall\,\ai,\bi
\ee 
and, by~\er{yijx}, 
\beq
\label{yijy0}
[Y_{ij},Y]=0\qquad\quad\forall\,i,j.
\ee

Before continuing the analysis of the obtained equations, 
we need to consider some special cases in the next subsections. 

\subsection{Some special cases}

Suppose that $A$ is of the form  
\beq
\label{atyi}
A=\sum_i u^i_1{\tilde{Y}}_i + \tilde{Y},
\ee
where ${\tilde{Y}}_i,\,\tilde{Y}$ 
are power series in the variables $x,\,t,\,u^1,\dots,u^{n-1}$ 
with coefficients in some Lie algebra. 

Substituting~\er{atyi} in~\eqref{eq.F1} we obtain
\begin{equation}\label{eq.F1.new}
F^1=\sum_i u^i_{2}H_i+F^2,
\end{equation}
where
\begin{gather}
\notag
H_i=\sum_k u^k_1{\tilde{T}}_{ik}
+{\tilde{Y}}_{,i}+[{\tilde{Y}}_i,\tilde{Y}]-{\tilde{Y}}_{i,x}+
\sum_j G^j_{u^i_{2}}{\tilde{Y}}_j,\\
\label{tildet}
{\tilde{T}}_{ik}=[{\tilde{Y}}_i,{\tilde{Y}}_k] - {\tilde{Y}}_{i,k} + {\tilde{Y}}_{k,i},\\
\notag
{\tilde{Y}}_{i,x}=\frac{\pd {\tilde{Y}}_{i}}{\partial x},\qquad
{\tilde{Y}}_{i,k}=\frac{\pd {\tilde{Y}}_{i}}{\partial u^k},\qquad
{\tilde{Y}}_{,i}=\frac{\pd \tilde{Y}}{\partial u^i}.
\end{gather}

Differentiating~\eqref{eqq.covering} with respect to $u^i_{2}$, one gets
\begin{multline}
\label{f2uix}
-\frac{\pd F^2}{\pd{u^i_1}}=
H_{i,x} +\sum_j u^j_{1}H_{i,j}-
\sum_{j,k}G^j_{u^i_{2}}u^k_1{\tilde{Y}}_{k,j}-
\sum_j G^j_{u^i_{2}}{\tilde{Y}}_{,j}-\sum_{j,k}u^j_1G^k_{u^i_{2}u^j}{\tilde{Y}}_k
-\sum_k G^k_{u^i_1}{\tilde{Y}}_k+\\
+\sum_{j,k}G^k_{u^j_{2}u^i_1}u^j_{2}{\tilde{Y}}_k+
\sum_k u^k_1[{\tilde{Y}}_k,H_i] + [\tilde{Y},H_i].
\end{multline}
Denote the right-hand side of~\er{f2uix} by~$\Omega_i$. 
Then equation~\er{f2uix} reads
\beq 
\label{f2omega}
-\frac{\pd F^2}{\pd{u^i_1}}=\Omega_i\qquad\forall\,i.
\ee
Differentiating~\eqref{f2omega} with respect to~$u^h_1,\,u^k_1$, we obtain
\begin{gather}
\label{omegaih}
\frac{\pd}{\pd u^h_1}\big(\Omega_i\big)=
\frac{\pd}{\pd u^i_1}\big(\Omega_h\big),\\
\label{omegaihk}
\frac{\pd^2}{\pd u^h_1\pd u^k_1}\big(\Omega_i\big)=
\frac{\pd^2}{\pd u^i_1\pd u^k_1}\big(\Omega_h\big).
\end{gather}
Using formulas~\er{eq.formule.Gpxx},~\er{guxjh}, 
it is straightforward to show that equations~\er{omegaihk},~\er{omegaih} reduce to 
\begin{gather}
\label{tihk}
{\tilde{T}}_{ih,k} = 2\Delta^{-1}\Big(2u^k{\tilde{T}}_{hi} + u^h{\tilde{T}}_{ki} + u^i{\tilde{T}}_{hk} + \delta^i_k\sum_m u^m{\tilde{T}}_{mh} + 
\delta^h_k\sum_m u^m{\tilde{T}}_{im}\Big) + [{\tilde{T}}_{ih},{\tilde{Y}}_k],\\
\label{tihx}
{\tilde{T}}_{ih,x} = [{\tilde{T}}_{ih},\tilde{Y}],
\end{gather}
where
\beq
\notag
{\tilde{T}}_{ih,k} = \frac{\pd {\tilde{T}}_{ih}}{\pd u^k},\qquad 
{\tilde{T}}_{ih,x} = \frac{\pd {\tilde{T}}_{ih}}{\pd x}.
\ee
Set $P_i={\tilde{Y}}_i+\dfrac{\partial}{\partial u^i}$. Then 
\begin{gather}
\notag
{\tilde{T}}_{ih}=[{\tilde{Y}}_i,{\tilde{Y}}_h]-{\tilde{Y}}_{i,h} + {\tilde{Y}}_{h,i}
=[P_i,P_h],\\
\label{tihqq}
\begin{split}
[{\tilde{T}}_{ih},{\tilde{T}}_{km}]=[{\tilde{T}}_{ih},[P_k,P_m]]
&=[[{\tilde{T}}_{ih},P_k],P_m]+[P_k,[{\tilde{T}}_{ih},P_m]]=\\
&=[[{\tilde{T}}_{ih},P_k],P_m]-[[{\tilde{T}}_{ih},P_m],P_k],
\end{split}
\end{gather}
and equation~\er{tihk} can be written as
\begin{equation}\label{tihqk}
[{\tilde{T}}_{ih},P_k] = 2\Delta^{-1}\Big(2u^k{\tilde{T}}_{ih} + u^h{\tilde{T}}_{ik} 
+ u^i{\tilde{T}}_{kh} + \delta^i_k\sum_m u^m{\tilde{T}}_{hm} 
+ \delta^h_k\sum_m u^m{\tilde{T}}_{mi}\Big).
\end{equation}
Using~\er{tihqq} and~\er{tihqk}, one obtains 
\begin{equation}
\label{eq.comm.Tih.Tkm}
[{\tilde{T}}_{ih},{\tilde{T}}_{km}]=4\Delta^{-2}(\delta^h_m{\tilde{T}}_{ki}+\delta^h_k{\tilde{T}}_{im}+\delta^i_m{\tilde{T}}_{hk}+\delta^i_k{\tilde{T}}_{mh}).
\end{equation}

\subsection{A zero-curvature representation}
\label{secszcr}

In order to analyze the structure of $\fd^1(\CE,a)$, 
we need to construct a ZCR with $A$ of the form 
\beq
\label{ahaty}
A=\sum_i u^i_1{\hat{Y}}_i,
\ee
where ${\hat{Y}}_i$ are power series in the variables $u^1,\dots,u^{n-1}$ with coefficients in some Lie algebra. 
In particular, we assume that ${\hat{Y}}_i$ do not depend on $x$, $t$. 
That is,
\beq
\label{yixt0}
\frac{\pd}{\pd x}\big({\hat{Y}}_i\big)=\frac{\pd}{\pd t}\big({\hat{Y}}_i\big)=0.
\ee
Also, similarly to~\er{condyi}, we assume 
\beq
\label{condhatyi}
\forall\,i=1,\dots,n-2,\qquad 
{\hat{Y}}_i\,\Big|_{u^{i+1}=u^{i+2}=\dots=u^{n-1}=0}=0,\qquad {\hat{Y}}_{n-1}=0.
\ee

By~\er{eq.T.general},~\er{eq.F1},~\er{ahaty}, one has 
\beq
\label{byi}
B=\sum_i u^i_{3} {\hat{Y}}_i + 
\sum_{i,j}u^i_{2}\left(u^j_1T_{ij}   + G^j_{u^i_{2}}{\hat{Y}}_j \right)+F^2.
\ee
Using~\er{f2uix} and~\er{yixt0}, it is straightforward to show that $F^2$ is of the form 
\begin{gather}
\label{f2mmm}
F^2 =\sum_{i,j,h}u^i_1u^j_1u^h_1 M_{ijh}+\sum_i u^i_1 M_{i} + M,\\
\notag
\begin{split}
M_{ihk}=2\Delta^{-2} \Big((&4(u^ku^i{\hat{Y}}_h+u^hu^k{\hat{Y}}_i+u^hu^i{\hat{Y}}_k)
-\langle u, u\rangle(\delta^i_k{\hat{Y}}_h+\delta^i_h{\hat{Y}}_k+\delta^h_k{\hat{Y}}_i)\\
&- 2\sum_m {u^m{\hat{Y}}_m}(\delta^h_ku^i+\delta^i_ku^h+\delta^i_hu^k)\Big)
-\frac23\sum_m\Delta^{-1}u^m(\delta^i_k{\hat{T}}_{hm}+\delta^i_h{\hat{T}}_{km}
+\delta^h_k{\hat{T}}_{im}),
\end{split}\\
\notag
M_i = \frac{3}{2}\left(r_n + 4\Delta^{-1}\sum_k(r_k-r_n){u^k}^2\right){\hat{Y}}_i,\\
\label{hattdef}
{\hat{T}}_{kh}=[{\hat{Y}}_k,{\hat{Y}}_h] - {\hat{Y}}_{k,h} + {\hat{Y}}_{h,k},
\end{gather}
where $M$ is a power series in the variables $x,t,u^1,\dots,u^{n-1}$.

Equations~\er{tihk},~\er{eq.comm.Tih.Tkm} remain valid if we replace~\er{atyi} 
by~\er{ahaty}, so one has
\begin{gather}
\label{tihkhat}
{\hat{T}}_{ih,k} = 2\Delta^{-1}\Big(2u^k{\hat{T}}_{hi} + u^h{\hat{T}}_{ki} + u^i{\hat{T}}_{hk} + \delta^i_k\sum_m u^m{\hat{T}}_{mh} + \delta^h_k\sum_m 
u^m{\hat{T}}_{im}\Big) + [{\hat{T}}_{ih},{\hat{Y}}_k],\\
\label{hattt}
[{\hat{T}}_{ih},{\hat{T}}_{km}]=4\Delta^{-2}(\delta^h_m{\hat{T}}_{ki}+\delta^h_k{\hat{T}}_{im}+\delta^i_m{\hat{T}}_{hk}+\delta^i_k{\hat{T}}_{mh}).
\end{gather}
Combining~\er{hattdef} with~\er{yixt0}, we obtain also 
\beq
\notag
\frac{\pd}{\pd x}\big({\hat{T}}_{kh}\big)=\frac{\pd}{\pd t}\big({\hat{T}}_{kh}\big)=0.
\ee
Thus ${\hat{T}}_{kh}$ is a power series in the variables $u^1,\dots,u^{n-1}$. 

Let $E_{i,j}\in\mathfrak{gl}_{n-1}(\fik)$ be the matrix with
$(i,j)$-th entry equal to 1 and all other entries equal to zero.
Then the matrices $C_{kh}=4(E_{k,h}-E_{h,k})$ 
span the Lie algebra $\mathfrak{so}_{n-1}\subset\mathfrak{gl}_{n-1}(\fik)$.
\begin{theorem}
\label{zcrson}
There is a ZCR with values in $\mathfrak{so}_{n-1}$ such that 
\begin{itemize}
\item $A$ is of the form~\er{ahaty}, 
where the power series~${\hat{Y}}_i(u^1,\dots,u^{n-1})$ satisfy~\er{condhatyi}, 
\item $B$ is defined by~\er{byi},~\er{f2mmm} with~$M=0$, 
\item for the power series~\er{hattdef}, one has 
\beq
\lb{hattekh}
{\hat{T}}_{kh}\Big|_{u^1=\dots=u^{n-1}=0}=4(E_{k,h}-E_{h,k}).
\ee
\end{itemize}
\end{theorem}
\begin{proof}
For $j=1,\dots,n-1$, let $N^j$ be the $\mathfrak{so}_{n-1}$-valued 
matrix-function whose $(h,i)$-entry is equal to 
$$
-2\delta^h_j\Delta^{-1}u^i + 2\delta^i_j\Delta^{-1}u^h.
$$
Let $A$ be given by~\er{ahaty} with~${\hat{Y}}_i=N^i$ 
and $B$ be defined by~\er{byi},~\er{f2mmm} with~$M=0$.
It is straightforward to check that such $A,\,B$ 
satisfy the zero-curvature condition~\er{gc}. 
In order to achieve condition~\er{condhatyi}, 
we can apply a gauge transformation depending on $u^1,\dots,u^{n-1}$, 
similarly to Theorem~\ref{evcov}.

According to Theorem~\ref{evcov}, the required gauge transformation 
is defined on a neighborhood of the point~\er{pointevfd1} and is equal 
to the identity transformation at this point.
Since we assume \er{x0t0=0}, \er{aik=0},
in our case the gauge transformation may depend nontrivially on 
$u^1,\dots,u^{n-1}$ 
and is equal to the identity transformation at the point $u^1=\dots=u^{n-1}=0$.

Computing ${\hat{T}}_{kh}$ by formula~\er{hattdef} for ${\hat{Y}}_i=N^i$, 
one obtains that~\er{hattekh} is valid.
It is easy to check that~\er{hattekh} remains valid 
after applying the gauge transformation, 
because the gauge transformation is 
equal to the identity transformation at the point $u^1=\dots=u^{n-1}=0$.
\end{proof}

\begin{remark}
\lb{zsoan}
Note that the ZCR described in Theorem~\ref{zcrson} 
is of order~$\le 1$ and is $a$-normal 
for the point $a\in\CE$ satisfying \er{pointevfd1}, \er{x0t0=0}, \er{aik=0}.
\end{remark}

\subsection{The ideal generated by $Y_{ij}$}
\lb{secyij}

Continue the analysis started in Section~\ref{prcomp}.
Recall that $A=A(x,t,u^j_0,u^j_1)$ and $B=B(x,t,u^j_0,u^j_1,u^j_2,u^j_3)$ 
are power series with coefficients in the algebra $\fd^1(\CE,a)$
and satisfy \er{gc}, \er{gd=0f1}, \er{gaukakf1}, \er{gbxx0f1}.
Recall that we have obtained formula~\er{eq.X.polynomial.second.order}.
\begin{lemma}
\lb{fd1gen}
The algebra $\fd^1(\CE,a)$
is generated by the coefficients of the power series 
\beq
\lb{yyyiji}
Y_{ij}\,\Big|_{t=0},\qquad Y_i\,\Big|_{t=0},\qquad Y\,\Big|_{t=0},\qquad
i,j=1,\dots,n-1.
\ee
\end{lemma}
\begin{proof}
Theorem~\ref{lemgenfdq} implies that the algebra $\fd^1(\CE,a)$
is generated by the coefficients of the power series $A\,\Big|_{t=0}$.
According to formula~\er{eq.X.polynomial.second.order}, 
the set of the coefficients of $A\,\Big|_{t=0}$ consists of 
the coefficients of the power series~\er{yyyiji}. 
\end{proof}

\begin{lemma}
\lb{yijzab}
For all $i,j$,
any coefficient of the power series $Y_{ij}\,\Big|_{t=0}$ belongs 
to the vector subspace spanned by the zero degree coefficients 
\beq
\lb{zdgz}
z^{{\ai}{\bi}}=Y_{{\ai}{\bi}}\,\Big|_{x=t=u^1=\dots=u^{n-1}=0},\qquad {\ai},{\bi}=1,\dots,n-1.
\ee
\end{lemma}
\begin{proof}
The statement follows from Lemma~\ref{yyab} and equation~\er{yijx0}.
\end{proof}

\begin{lemma}
\label{lemyijp}
Let $P$ be a power series in the variables $x,\,u^l$. 
Suppose that 
\beq
\label{yijp0}
\Big[Y_{ij}\,\Big|_{t=0},\,P\Big]=0\qquad\forall\,i,j. 
\ee
Then any coefficient of the power series $Y_{ij}\,\Big|_{t=0}$ commutes 
with any coefficient of $P$. 
\end{lemma}
\begin{proof}
Using Lemma~\ref{yijzab}, one can prove the statement by induction 
on the degree of coefficients of~$P$. 
\end{proof}

\begin{lemma}
\label{lemyijyy}
For all $i,j$, 
any coefficient of the power series $Y_{ij}\,\Big|_{t=0}$ commutes 
\begin{itemize}
\item with any coefficient of $Y\,\Big|_{t=0}$,
\item with any coefficient of $Y_{{\ai}{\bi}}\,\Big|_{t=0}$ 
for all ${\ai},{\bi}$.
\end{itemize}
\end{lemma}
\begin{proof}
Substituting $t=0$ to~\er{yijy0} and~\er{rel.1}, 
we obtain the required statement by Lemma~\ref{lemyijp}. 
\end{proof}

Consider the ZCR with values in $\mathfrak{so}_{n-1}$ 
constructed in Theorem~\ref{zcrson}. 
Recall that, according to~\er{ahaty},~\er{hattdef}, and Theorem~\ref{zcrson},   
this ZCR determines the power series ${\hat{Y}}_i,\,{\hat{T}}_{kh}$ 
in the variables $u^1,\dots,u^{n-1}$ with coefficients in $\mathfrak{so}_{n-1}$.

According to Remark~\ref{zsoan}, 
this $\mathfrak{so}_{n-1}$-valued ZCR is of order~$\le 1$ and is $a$-normal.
Therefore, by Remark~\ref{homfdhmg},
this ZCR determines a homomorphism $\hmm\cl\fd^1(\CE,a)\to\mathfrak{so}_{n-1}$.   

For a power series $P$ with coefficients in $\fd^1(\CE,a)$, 
we can apply $\hmm$ to each coefficient of~$P$ and obtain a power series $\hmm(P)$ 
with coefficients in $\mathfrak{so}_{n-1}$. 

Recall that $T_{kh}$ is defined by~\er{eq.def.Tkh}. 
By the definition of~$\hmm$, one has 
\beq
\label{vfyi}
\hmm(Y_i)={\hat{Y}}_i,\qquad \hmm(T_{kh})={\hat{T}}_{kh}.
\ee 
Combining~\er{vfyi} with~\er{tihkhat},~\er{hattt}, we obtain 
\begin{gather}
\label{vftihk}
\hmm\big({{T}}_{ih,k}\big)=\hmm\Big( 2\Delta^{-1}\left(2u^k{{T}}_{hi} + u^h{{T}}_{ki} + u^i{{T}}_{hk} + \delta^i_k \sum_m u^m{{T}}_{mh} + \delta^h_k \sum_m u^m{{T}}_{im}\right) + [{{T}}_{ih},{{Y}}_k]\Big),\\
\label{vftihkm}
\hmm\big([{{T}}_{ih},{{T}}_{km}]\big)=
\hmm\big(4\Delta^{-2}(\delta^h_m{{T}}_{ki}+\delta^h_k{{T}}_{im}+\delta^i_m{{T}}_{hk}+\delta^i_k{{T}}_{mh})\big).
\end{gather}
Set  
\begin{gather}
\label{gamihk}
\Gamma_{ihk}={{T}}_{ih,k}-\big( 2\Delta^{-1}\left(2u^k{{T}}_{hi} + u^h{{T}}_{ki} + u^i{{T}}_{hk} + \delta^i_k \sum_m u^m{{T}}_{mh} + \delta^h_k \sum_m u^m{{T}}_{im}\right) + [{{T}}_{ih},{{Y}}_k]\big),\\
\label{gamihkm}
\Gamma_{ihkm}=[{{T}}_{ih},{{T}}_{km}]-
4\Delta^{-2}(\delta^h_m{{T}}_{ki}+\delta^h_k{{T}}_{im}+\delta^i_m{{T}}_{hk}+\delta^i_k{{T}}_{mh}).
\end{gather}
\begin{lemma}
One has 
\beq
\label{ygamma0}
[Y_{{\ai}{\bi}},\Gamma_{ihk}]=0,\qquad [Y_{{\ai}{\bi}},\Gamma_{ihkm}]=0
\qquad\quad\forall\,\ai,\bi,i,h,k,m.
\ee
\end{lemma}
\begin{proof}
Set $P'_i={{Y}}_i+\dfrac{\partial}{\partial u^i}$, then 
\begin{gather}
\label{tihpn}
{{T}}_{ih}=[{{Y}}_i,{{Y}}_h]-{{Y}}_{i,h} + {{Y}}_{h,i}=[P'_i,P'_h],\\
\label{gamihkp}
\Gamma_{ihk}=[P'_k,{{T}}_{ih}]-2\Delta^{-1}\left(2u^k{{T}}_{hi} + u^h{{T}}_{ki} + u^i{{T}}_{hk} + \delta^i_k \sum_m u^m{{T}}_{mh} + \delta^h_k \sum_m u^m{{T}}_{im}\right).
\end{gather}
and equation~\er{tihk} can be written as
\beq
\label{yijph}
[Y_{ij},P'_h]=2\Delta^{-1}\big(2u^h Y_{ij}+ u^i Y_{hj} 
+ u^j Y_{hi} - \delta^i_h \sum_m u^m Y_{mj} - \delta^j_h \sum_m u^m Y_{mi}\big).
\ee
Using~\er{eq.comp.cond.system.defining.Yij},~\er{gamihkm},~\er{tihpn},~\er{gamihkp},~\er{yijph}, 
one can check relations~\er{ygamma0} by a straightforward computation.  
\end{proof}

Combining~\er{vftihk},~\er{vftihkm},~\er{gamihk},~\er{gamihkm} with~\er{yixt0}, 
we obtain 
\beq
\label{vfgam}
\hmm\big(\Gamma_{ihk}\big)=\hmm\Big(\Gamma_{ihk}\,\Big|_{t=0}\Big)=0,\qquad 
\hmm\big(\Gamma_{ihkm}\big)=\hmm\Big(\Gamma_{ihkm}\,\Big|_{t=0}\Big)=0.
\ee
From~\er{vfyi} and~\er{yixt0} it follows that
\beq
\label{vfyyy}
\hmm(Y_i)=\hmm\Big(Y_i\,\Big|_{t=0}\Big)=
\hmm\Big(Y_i\,\Big|_{x=t=0}\Big)={\hat{Y}}_i.
\ee 
Set 
\beq
\lb{giyy}
\Gamma_i=Y_i\,\Big|_{t=0}-Y_i\,\Big|_{x=t=0}, 
\ee
then equation~\er{vfyyy} implies 
\beq
\label{vfgami}
\hmm(\Gamma_i)=0.
\ee

Let $\mh\subset\fd^1(\CE,a)$ be the subalgebra generated by
the coefficients of the power series $Y_{i}\,\Big|_{t=0},\,i=1,\dots,n-1$. 
Then $Y_{i}\,\Big|_{t=0}$ 
and $T_{kh}\,\Big|_{t=0}$ are power series with coefficients in~$\mh$. 
Therefore, the coefficients of the power series
\beq
\lb{gggt0}
\Gamma_{ihk}\,\Big|_{t=0},\qquad\Gamma_{ihkm}\,\Big|_{t=0},\qquad 
\Gamma_i 
\ee 
belong to $\mh$ as well.
\begin{lemma}
\lb{idker}
Let $\mic\subset\mh$ be the ideal of $\mh$ generated by 
the coefficients of the power series~\er{gggt0}. Then 
\beq
\lb{micmhvf}
\mic=\mh\cap\ker\hmm.
\ee
\end{lemma}
\begin{proof}
From~\er{vfgam} and~\er{vfgami} it follows that 
\beq
\lb{micsker}
\mic\subset\ker\hmm. 
\ee
Consider the quotient Lie algebra $\mg=\mh/\mic$ 
and the natural projection $\psi\cl\mh\to\mg$. 
Set 
\beq
\lb{byitkh}
{\mathbb{Y}}_i=\psi\Big(Y_i\,\Big|_{t=0}\Big),\qquad 
{\mathbb{T}}_{kh}=\psi\Big(T_{kh}\,\Big|_{t=0}\Big),
\ee
which are power series with coefficients in~$\mg$. 
From~\er{eq.def.Tkh},~\er{condyi} it follows that
\begin{gather}
\label{btkh}
{\mathbb{T}}_{kh}=[{\mathbb{Y}}_k,{\mathbb{Y}}_h] - {\mathbb{Y}}_{k,h} + {\mathbb{Y}}_{h,k},\\
\label{bcondyi}
\forall\,i=1,\dots,n-2,\qquad {\mathbb{Y}}_i\,\Big|_{u^{i+1}=u^{i+2}=\dots=u^{n-1}=0}=0,\qquad {\mathbb{Y}}_{n-1}=0.
\end{gather}
By the definition of~$\mic$ and~$\psi$, 
\beq
\lb{psigam}
\psi\Big(\Gamma_{ihk}\,\Big|_{t=0}\Big)=0,\qquad
\psi\Big(\Gamma_{ihkm}\,\Big|_{t=0}\Big)=0,\qquad 
\psi\big(\Gamma_i\big)=0. 
\ee
According to \er{gamihk}, \er{gamihkm}, \er{giyy}, \er{byitkh}, 
equations~\er{psigam} say that 
\begin{gather}
\label{bgamihk}
{\mathbb{T}}_{ih,k}=2\Delta^{-1}\left(2u^k{\mathbb{T}}_{hi} + u^h{\mathbb{T}}_{ki} + u^i{\mathbb{T}}_{hk} + \delta^i_k \sum_m u^m{\mathbb{T}}_{mh} + \delta^h_k \sum_m u^m{\mathbb{T}}_{im}\right) + [{\mathbb{T}}_{ih},{\mathbb{Y}}_k],\\
\label{bgamihkm}
[{\mathbb{T}}_{ih},{\mathbb{T}}_{km}]=
4\Delta^{-2}(\delta^h_m{\mathbb{T}}_{ki}+\delta^h_k{\mathbb{T}}_{im}+\delta^i_m{\mathbb{T}}_{hk}+\delta^i_k{\mathbb{T}}_{mh}),\\
\lb{dxby0}
\frac{\pd}{\pd x}\big({\mathbb{Y}}_i\big)=0.
\end{gather}
From~\er{byitkh},~\er{btkh},~\er{dxby0} it follows that 
${\mathbb{Y}}_i,\,{\mathbb{T}}_{kh}$ are 
power series in the variables $u^1,\dots,u^{n-1}$ with coefficients in~$\mg$.

\begin{lemma}
\lb{tgeng}
The elements 
\beq
\label{mbtkh0}
{\mathbb{T}}_{kh}\,\Big|_{u^1=u^2=\dots=u^{n-1}=0}\,\in\mg,\qquad k,h=1,\dots,n-1,
\ee
generate the Lie algebra $\mg$. 
\end{lemma}
\begin{proof}
From the definition of $\mh$ and $\mg$ it follows that $\mg$ is generated by the coefficients 
of the power series ${\mathbb{Y}}_i$. 

Relations~\er{btkh},~\er{bcondyi},~\er{bgamihk} are of the type considered in Lemmas~\ref{zi},~\ref{zpsi} for $v^i=u^i$.
Therefore, by Lemmas~\ref{zi},~\ref{zpsi}, 
any coefficient of~${\mathbb{Y}}_i,\,{\mathbb{T}}_{kh}$ belongs 
to the Lie subalgebra generated by~\er{mbtkh0}.  
\end{proof}
From~\er{bgamihkm} and Lemma~\ref{tgeng} it follows that the map 
$$
\so_{n-1}\to\mg,\qquad (E_{k,h}-E_{h,k})\,\mapsto\,\frac14\Big({\mathbb{T}}_{kh}\,\Big|_{u^1=u^2=\dots=u^{n-1}=0}\Big), 
$$
is a surjective homomorphism. Therefore, 
\beq
\lb{dimle}
\dim\mg=\dim\big(\mh/\mic\big)\le\dim\so_{n-1}.
\ee
According to Theorem~\ref{zcrson}, 
the coefficients of the power series~${\hat{Y}}_i$ generate~$\so_{n-1}$. 
Combining this with~\er{vfyyy} and the definition of~$\mh$, we obtain 
\beq
\lb{vfmhso}
\hmm(\mh)=\so_{n-1}.
\ee 
Combining~\er{micsker},~\er{dimle},~\er{vfmhso}, one gets~\er{micmhvf}.
\end{proof}


Substituting~$t=0$ to~\er{system.defining.Yij}, we obtain
\begin{multline}
\label{yijht0}
Y_{ij,h}\,\Big|_{t=0}=\\
=\Big[Y_{ij}\,\Big|_{t=0},Y_h\,\Big|_{t=0}\Big] 
- 2\Delta^{-1}\Big(2u^h Y_{ij}+u^i Y_{hj}+u^j Y_{hi}
-\delta^i_h \sum_m u^m Y_{mj}-\delta^j_h \sum_m u^m Y_{mi}\Big)\,\Big|_{t=0}\quad
\forall\,i,j,h.
\end{multline}
From \er{yabx0} it follows that $Y_{ij}\,\Big|_{t=0}=Y_{ij}\,\Big|_{x=t=0}$ for all $i,j$.
Therefore, substituting~$x=t=0$ to~\er{system.defining.Yij}, we get
\begin{multline}
\label{yijhxt0}
Y_{ij,h}\,\Big|_{t=0}=\\
=\Big[Y_{ij}\,\Big|_{t=0},Y_h\,\Big|_{x=t=0}\Big] 
- 2\Delta^{-1}\Big(2u^h Y_{ij}+u^i Y_{hj}+u^j Y_{hi}
-\delta^i_h \sum_m u^m Y_{mj}-\delta^j_h \sum_m u^m Y_{mi}\Big)\,\Big|_{t=0}
\quad\forall\,i,j,h.
\end{multline}
Subtracting \er{yijhxt0} from \er{yijht0}, one obtains
\beq
\lb{yabgi0}
\Big[Y_{ij}\,\Big|_{t=0},\,\Gamma_h\Big]=0\quad\qquad\forall\,i,j,h,
\ee
where $\Gamma_h=Y_h\,\Big|_{t=0}-Y_h\,\Big|_{x=t=0}$.

Let $\za$ be the vector subspace spanned by the elements 
$z^{{\ai}{\bi}}$, ${\ai},{\bi}=1,\dots,n-1$, defined in~\er{zdgz}.
Lemma~\ref{lemyijyy} implies that $\za$ is a commutative subalgebra 
of the Lie algebra $\fd^1(\CE,a)$.

Recall that $Y_h\,\Big|_{t=0}$ is a power series in the variables $x$, $u^1,\dots,u^{n-1}$.
So one has 
\beq
\lb{yht0}
Y_h\,\Big|_{t=0}=\sum_{l,i_1,\dots,i_{n-1}\ge 0}
x^l\big(u^1\big)^{i_1}\dots\big(u^{n-1}\big)^{i_{n-1}}\beta^{h}_{l,i_1,\dots,i_{n-1}}
\ee
for some elements $\beta^{h}_{l,i_1,\dots,i_{n-1}}\in\fd^1(\CE,a)$.

Lemma~\ref{yijzab} says that for all $i,j$ 
any coefficient of the power series $Y_{ij}\,\Big|_{t=0}$ belongs to $\za$.
Then equation \er{yijht0} implies that for all $i,j,h$ 
the coefficients of the power series $\Big[Y_{ij}\,\Big|_{t=0},Y_h\,\Big|_{t=0}\Big]$ 
belong to $\za$. Using these facts and the definition of $\za$, 
by induction on $l+i_1+\dots+i_{n-1}$ one proves that 
$[\beta^{h}_{l,i_1,\dots,i_{n-1}},\za]\subset\za$.
Since the Lie algebra $\mh$ is generated by the elements $\beta^{h}_{l,i_1,\dots,i_{n-1}}$
and $\za$ is spanned by the elements $z^{{\ai}{\bi}}$, 
we see that $[z^{{\ai}{\bi}},\mh]\subset\za$ for all ${\ai},{\bi}$.

Substituting $t=0$ to~\er{ygamma0}, one gets
\beq
\label{ygammat0}
\Big[Y_{{\ai}{\bi}}\,\Big|_{t=0},\,\Gamma_{ihk}\,\Big|_{t=0}\Big]=0,
\qquad \Big[Y_{{\ai}{\bi}}\,\Big|_{t=0},\Gamma_{ihkm}\,\Big|_{t=0}\Big]=0
\qquad\quad\forall\,\ai,\bi,i,h,k,m.
\ee
By Lemma~\ref{lemyijp}, from \er{yabgi0} and \er{ygammat0} 
it follows that $z^{{\ai}{\bi}}$ defined in~\er{zdgz} commutes 
with any coefficient of the power series~\er{gggt0}. 

Thus $[z^{{\ai}{\bi}},\mh]\subset\za$ and $z^{{\ai}{\bi}}$ 
commutes with any coefficient of the power series~\er{gggt0}.
Combining this with Lemma~\ref{idker}, one gets 
\beq
\lb{zabhker}
\big[z^{{\ai}{\bi}},\,(\mh\cap\ker\hmm)\big]=0\qquad\quad\forall\,{\ai},{\bi}.
\ee

\begin{lemma}
\lb{videal}
The vector space $\za$ spanned by the elements 
$z^{{\ai}{\bi}}$, ${\ai},{\bi}=1,\dots,n-1$, is a commutative ideal 
of the Lie algebra $\fd^1(\CE,a)$.
\end{lemma}
\begin{proof}
We have shown above that $\za$ is a commutative subalgebra $\fd^1(\CE,a)$.
Let us show that $\za$ is an ideal of $\fd^1(\CE,a)$.

According to Lemma~\ref{fd1gen}, the Lie algebra $\fd^1(\CE,a)$
is generated by the coefficients of the power series \er{yyyiji}.
As we have shown above, $[\za,\mh]\subset\za$, 
where $\mh\subset\fd^1(\CE,a)$ is the subalgebra generated by
the coefficients of the power series $Y_{i}\,\Big|_{t=0}$, $i=1,\dots,n-1$.
Furthermore, Lemma~\ref{lemyijyy} implies that any element of $\za$ 
commutes with any coefficient of the power series 
$Y\,\Big|_{t=0}$ and $Y_{ij}\,\Big|_{t=0}$.
Therefore, $\za$ is an ideal of $\fd^1(\CE,a)$.
\end{proof}

\begin{lemma}
\lb{zabvker}
For any $v\in\ker\hmm\subset\fd^1(\CE,a)$, one has $[z^{{\ai}{\bi}},v]=0$ 
for all ${\ai},{\bi}$.
\end{lemma}
\begin{proof}
By Lemma~\ref{fd1gen}, the algebra $\fd^1(\CE,a)$
is generated by the coefficients of~\er{yyyiji}.
Since $\hmm(Y_{ij})=\hmm(Y)=0$ and the coefficients of~$Y_i\,\Big|_{t=0}$ generate $\mh$, 
the ideal~$\ker\hmm$ is generated by the subalgebra~$\mh\cap\ker\hmm$ and 
the coefficients of~$Y_{ij}\,\Big|_{t=0},\ Y\,\Big|_{t=0}$. 
Then the required statement follows from~\er{zabhker}, Lemma~\ref{lemyijyy}, 
and Lemma~\ref{videal}.
\end{proof}
From~\er{yijyxt} and~\er{yabx0} it follows that $y^{ij}_{i_1,\dots,i_{n-1}}$ 
is a power series in one variable $t$.
\begin{lemma}
\lb{yijyij}
For all $i,j,i',j'$, any coefficient of the power series $y^{ij}_{0,\dots,0}(t)$ 
commutes with any coefficient of~$y^{i'j'}_{0,\dots,0}(t)$. 
\end{lemma}
\begin{proof}
For each $l\in\zp$, set
\beq
\lb{zijly}
z^{ij}_l=\frac{\pd^l y^{ij}_{0,\dots,0}}{\pd t^l}\,\bigg|_{t=0}\,\in\,\fd^1(\CE,a).
\ee
Similarly to Lemma~\ref{zabvker}, by induction on $l$, one can prove that 
\beq
\lb{zijlker}
[z^{ij}_l,\ker\hmm]=0\qquad\quad\forall\,i,j,l.
\ee
Since $\hmm(Y_{i'j'})=0$, property~\er{zijlker} implies 
that $z^{ij}_l$ commutes with any coefficient of the power series $y^{i'j'}_{0,\dots,0}(t)$.
Since the elements $z^{ij}_l$ are the coefficients of the power series $y^{ij}_{0,\dots,0}$, 
we see that any coefficient of~$y^{ij}_{0,\dots,0}(t)$ 
commutes with any coefficient of~$y^{i'j'}_{0,\dots,0}(t)$.
\end{proof}

\begin{theorem}
\lb{idealyij}
The Lie subalgebra $\mathfrak{S}$ 
generated by the coefficients of~$Y_{ij}$ is abelian and satisfies
\beq
\lb{msmh}
[\mathfrak{S},\ker\hmm]=0.
\ee
Furthermore, the subalgebra $\mathfrak{S}$ is an ideal of $\fd^1(\CE,a)$. 
\end{theorem}
\begin{proof}
The fact that $\mathfrak{S}$ is abelian follows from 
Lemma~\ref{yyab}, property~\er{yabx0}, and Lemma~\ref{yijyij}.  
According to Lemma~\ref{yyab} and property~\er{yabx0}, 
the subalgebra $\mathfrak{S}$ is generated by the coefficients~\er{zijly}
of the power series $y^{ij}_{0,\dots,0}(t)$.
Then \er{zijlker} implies~\er{msmh}.

By Lemma~\ref{fd1gen}, the algebra $\fd^1(\CE,a)$
is generated by the coefficients of~\er{yyyiji}. 
Therefore, in order to show that $\mathfrak{S}$ is an ideal, 
one needs to prove that $[C,\mathfrak{S}]\subset\mathfrak{S}$ 
for any coefficient $C$ of $Y_i\,\Big|_{t=0}$ and $Y\,\Big|_{t=0}$. 
This can be easily deduced from 
equations~\er{system.defining.Yij} and~\er{yijy0}, 
using Lemma~\ref{yyab} and property~\er{yabx0}. 
\end{proof}

\subsection{The ideal generated by $\tilde{Y}_{i}$}
\lb{sidyi}

In this subsection we study the quotient Lie algebra $\bl=\fd^1(\CE,a)/\mathfrak{S}$, 
where $\mathfrak{S}$ is the ideal generated by the coefficients of~$Y_{ij}$. 
Therefore, we can assume that $A,\,B$ are power series with coefficients in $\bl$  
and formula~\er{atyi} holds.

Then ${\tilde{T}}_{ij}$ defined in~\er{tildet} 
are power series in the variables $x,\,t,\,u^1,\dots,u^{n-1}$ with coefficients in $\bl$. 
Let $\tilde\mg$ be the Lie algebra of formal power series in the variables $x,\,t$ with coefficients in $\bl$. 
Then ${\tilde{T}}_{ij}$ can be regarded as a power series in the variables $u^1,\dots,u^{n-1}$ with coefficients in $\tilde\mg$ 
\beq
\notag
{\tilde{T}}_{ij}=\sum_{i_1,\dots,i_{n-1}\ge 0}\big(u^1\big)^{i_1}\dots\big(u^{n-1}\big)^{i_{n-1}}\al^{ij}_{i_1,\dots,i_{n-1}}(x,t), 
\ \quad \al^{ij}_{i_1,\dots,i_{n-1}}(x,t)\in\tilde\mg.
\ee

From~\er{condyi},~\er{condy0} we get
\begin{gather}
\label{tcondyi}
\forall\,i=1,\dots,n-2,\qquad {\tilde{Y}}_i\,\Big|_{u^{i+1}=u^{i+2}=\dots=u^{n-1}=0}=0,
\qquad {\tilde{Y}}_{n-1}=0,\\ 
\label{tcondy0}
{\tilde{Y}}\,\Big|_{u^{1}=u^{2}=\dots=u^{n-1}=0}=0.
\end{gather}

Similarly to Lemma~\ref{yyab}, using Lemma~\ref{zpsi}, 
from equations~\er{tihk} we obtain the following. 
\begin{lemma}
\label{alalab}
For any $i,j,i_1,\dots,i_{n-1}$, 
the power series $\al^{ij}_{i_1,\dots,i_{n-1}}(x,t)$ belongs to the vector subspace spanned 
by the power series $\al^{{\ai}{\bi}}_{0,\dots,0}(x,t),\ {\ai},{\bi}=1,\dots,n-1$.
\end{lemma}
Similarly to~\er{yijx0}, from~\er{tihx} and \er{tcondy0} one gets 
\begin{multline}
\label{alijx0}
\frac{\pd}{\pd x}\Big(\al^{{\ai}{\bi}}_{0,\dots,0}(x,t)\Big)=
{\tilde{T}}_{{\ai}{\bi},x}\,\Big|_{u^{1}=u^{2}=\dots=u^{n-1}=0} =\\
= \Big[{\tilde{T}}_{{\ai}{\bi}}\,\Big|_{u^{1}=u^{2}=\dots=u^{n-1}=0},\,
\tilde{Y}\,\Big|_{u^{1}=u^{2}=\dots=u^{n-1}=0}\Big]=0\qquad\forall\,{\ai},{\bi}.
\end{multline}
Combining~\er{alijx0} with Lemma~\ref{alalab}, we obtain 
\beq
\lb{ttabx0}
{\tilde{T}}_{{\ai}{\bi},x}=0\qquad\quad\forall\,{\ai},{\bi}
\ee 
and, by~\er{tihx}, 
\beq
\label{ttijy0}
[{\tilde{T}}_{ij},\tilde{Y}]=0\qquad\quad\forall\,i,j.
\ee

\begin{lemma}
\lb{blgen}
The algebra $\bl$ is generated by the coefficients of the power series 
\beq
\lb{ttijy}
{\tilde{T}}_{ij}\,\Big|_{t=0},\qquad \tilde{Y}\,\Big|_{t=0},\qquad\quad
i,j=1,\dots,n-1.
\ee
\end{lemma}
\begin{proof}
Similarly to Lemma~\ref{fd1gen}, taking into account formula~\er{atyi}, 
one proves that $\bl$ is generated by the coefficients of 
$\tilde{Y}_i\,\Big|_{t=0},\,\ \tilde{Y}\,\Big|_{t=0}$. 
By Lemma~\ref{zi}, the Lie subalgebra generated by the coefficients 
of~$\tilde{Y}_i\,\Big|_{t=0}$ coincides with the Lie subalgebra generated 
by the coefficients of~${\tilde{T}}_{ij}\,\Big|_{t=0}$. 
\end{proof}

Similarly to Lemma~\ref{yijzab}, using equation~\er{tihk}, we obtain the following.
\begin{lemma}
\lb{tijzab}
Any coefficient of the power series ${\tilde{T}}_{ij}\,\Big|_{t=0}$ belongs 
to the vector subspace spanned by the zero degree coefficients 
\beq
\lb{bdgb}
\beta^{{\ai}{\bi}}={\tilde{T}}_{{\ai}{\bi}}\,\Big|_{x=t=u^1=\dots=u^{n-1}=0},
\qquad {\ai},{\bi}=1,\dots,n-1.
\ee
\end{lemma}

\begin{lemma}
\lb{tijso}
Let $\mathfrak{Z}\subset\fd^1(\CE,a)/\mathfrak{S}$ be the Lie subalgebra 
generated by the coefficients of the power series~${\tilde{T}}_{ij}\,\Big|_{t=0}$.
This subalgebra is isomorphic to $\so_{n-1}$.
\end{lemma}
\begin{proof}
Recall that $E_{i,j}\in\mathfrak{gl}_{n-1}(\fik)$ is the matrix with
$(i,j)$-th entry equal to 1 and all other entries equal to zero.
The matrices 
$$
C_{kh}=4(E_{k,h}-E_{h,k}),\qquad\quad
k,h=1,\dots,n-1,\qquad\quad 
k<h,
$$
form a basis of the Lie algebra $\mathfrak{so}_{n-1}\subset\mathfrak{gl}_{n-1}(\fik)$.

Substituting $x=t=u^1=\dots=u^{n-1}=0$ in equation~\er{eq.comm.Tih.Tkm}, 
we see that the elements~\er{bdgb} satisfy
\begin{equation}
\label{betaeq}
[\beta^{ih},\beta^{km}]=
4(\delta^h_m\beta^{ki}+\delta^h_k\beta^{im}+\delta^i_m\beta^{hk}+\delta^i_k\beta^{mh}).
\end{equation}
Since ${\tilde{T}}_{{\ai}{\bi}}=-{\tilde{T}}_{{\bi}{\ai}}$, 
we have $\beta^{{\ai}{\bi}}=-\beta^{{\bi}{\ai}}$.
Combining this with equation~\er{betaeq}, we see that the map
\beq
\lb{fszb}
\mathbf{f}\cl\mathfrak{so}_{n-1}\to\mathfrak{Z},\qquad\quad 
\mathbf{f}(C_{kh})=\beta^{kh},
\ee
is a homomorphism.

Consider the homomorphism $\hmm\cl\fd^1(\CE,a)\to\mathfrak{so}_{n-1}$
constructed in Section~\ref{secyij}.
Since $\hmm(Y_{ij})=0$ and the ideal $\mathfrak{S}\subset\fd^1(\CE,a)$
is generated by the coefficients of~$Y_{ij}$,
the homomorphism $\hmm$ determines a homomorphism 
$\hat{\hmm}\cl\fd^1(\CE,a)/\mathfrak{S}\to\mathfrak{so}_{n-1}$.

Consider the ZCR with coefficients in $\mathfrak{so}_{n-1}$ 
constructed in Theorem~\ref{zcrson}. 
Recall that, according to~\er{ahaty},~\er{hattdef}, and Theorem~\ref{zcrson},   
this ZCR determines the power series ${\hat{Y}}_i,\,{\hat{T}}_{kh}$ 
in the variables $u^1,\dots,u^{n-1}$ with coefficients in $\mathfrak{so}_{n-1}$
such that \er{hattekh}, \er{vfyi} hold.

For a power series $P$ with coefficients in $\fd^1(\CE,a)/\mathfrak{S}$, 
we can apply $\hat{\hmm}$ to each coefficient of~$P$ and obtain 
a power series $\hat{\hmm}(P)$ with coefficients in $\mathfrak{so}_{n-1}$. 

The projection $\fd^1(\CE,a)\to\fd^1(\CE,a)/\mathfrak{S}$ 
maps $Y_i$ to ${\tilde{Y}}_i$ and $T_{kh}$ to ${\tilde{T}}_{kh}$. 
Combining this with~\er{vfyi}, we obtain
\beq
\label{tivfyi}
\hat{\hmm}({\tilde{Y}}_i)={\hat{Y}}_i,\qquad 
\hat{\hmm}({\tilde{T}}_{kh})={\hat{T}}_{kh}.
\ee 
From \er{hattekh}, \er{bdgb}, \er{tivfyi} we get
\beq
\lb{hbkhe}
\hat{\hmm}(\beta^{kh})=4(E_{k,h}-E_{h,k})=C_{kh}.
\ee
Relations \er{fszb}, \er{hbkhe} imply that the Lie subalgebra 
generated by the elements $\beta^{kh}$ is isomorphic to 
$\mathfrak{so}_{n-1}$.
According to Lemma~\ref{tijzab}, the Lie subalgebra generated by $\beta^{kh}$ 
coincides with $\mathfrak{Z}$, so $\mathfrak{Z}$ is isomorphic to $\mathfrak{so}_{n-1}$.
\end{proof}

\begin{lemma}
\lb{anycty}
Any coefficient of~${\tilde{T}}_{ij}\,\Big|_{t=0}$ commutes 
with any coefficient of~$\tilde{Y}\,\Big|_{t=0}$.
\end{lemma}
\begin{proof}
The statement is proved similarly to Lemma~\ref{lemyijyy}, 
using~\er{ttijy0} and Lemma~\ref{tijzab}.
\end{proof}
Recall that $L(n)$ is the infinite-dimensional Lie algebra generated by~\er{qie}.

\begin{theorem}
\lb{tijnge4}
Suppose that $n\ge 4$. 
Then the Lie algebra $\bl=\fd^1(\CE,a)/\mathfrak{S}$ is isomorphic to the direct sum 
$L(n)\oplus\so_{n-1}$, where $\so_{n-1}$ is generated by the coefficients of~${\tilde{T}}_{ij}$. 
The ideal~$\so_{n-1}$ coincides also with the subalgebra generated 
by the coefficients of~$\tilde{Y}_i$.

The homomorphism $\fd^{1}(\CE,a)\to\fd^0(\CE,a)$ 
from~\er{fdnn-1} coincides with the composition of the homomorphisms
\beq
\lb{fd1lnso}
\fd^{1}(\CE,a)\to\fd^1(\CE,a)/\mathfrak{S}\cong
L(n)\oplus\mathfrak{so}_{n-1}
\to L(n)\cong \fd^0(\CE,a),
\ee
where $L(n)\cong \fd^0(\CE,a)$ is the isomorphism described in Theorem~\ref{thfd0}.
\end{theorem}
\begin{proof}
Let $S_1$ be the subalgebra generated by the coefficients of~$\tilde{Y}\,\Big|_{t=0}$
and $S_2$ be the subalgebra generated by the coefficients of~${\tilde{T}}_{ij}\,\Big|_{t=0}$.  
By Lemmas~\ref{blgen},~\ref{anycty}, one has $\bl=S_1+S_2$ and $[S_1,S_2]=0$. 
By Lemma~\ref{tijso}, $S_2\cong\so_{n-1}$. 
Since for $n\ge 4$ the center of the Lie algebra $\so_{n-1}$ is trivial, 
we obtain $S_1\cap S_2=0$ and, therefore, $\bl=S_1\oplus S_2$. In particular, $S_2$ is an ideal of $\bl$. 

For $n\ge 4$, from equations~\er{eq.comm.Tih.Tkm},~\er{ttabx0} and Lemma~\ref{alalab} 
it follows that for all $\ai,\bi=1,\dots,n-1$ any coefficient of~${\tilde{T}}_{\ai\bi}$ 
belongs to the ideal generated by
the coefficients of~${\tilde{T}}_{ij}\,\Big|_{t=0}$, $i,j=1,\dots,n-1$.
(To see this, one needs to differentiate~\er{eq.comm.Tih.Tkm} with respect to $t$ 
several times and substitute $t=0$.)

Since the subalgebra $S_2$ generated by the coefficients of~${\tilde{T}}_{ij}\,\Big|_{t=0}$ 
is an ideal, we obtain that all coefficients of~${\tilde{T}}_{\ai\bi}$ belong to~$S_2$. 

According to~\er{tildet} and~\er{tcondyi}, 
we can apply Lemma~\ref{zi} to the power series ${\tilde{Y}}_i$ and ${\tilde{T}}_{ij}$.
This implies that the subalgebra generated by the coefficients of 
${\tilde{Y}}_i$, $i=1,\dots,n-1$, coincides with 
the subalgebra $S_2$ generated by the coefficients of~${\tilde{T}}_{ij}$, $i,j=1,\dots,n-1$.
Since $S_2$ is an ideal, we see that
the ideal generated by the coefficients of $\tilde{Y}_i$ coincides with $S_2$. 
Therefore, taking into account formulas~\er{eq.T.general} and~\er{atyi}, 
one obtains that the quotient~$\bl/S_2$ is isomorphic to~$\fd^0(\CE,a)$. 
Combining this with Theorem~\ref{thfd0}, 
we get $S_1\cong\bl/S_2\cong\fd^0(\CE,a)\cong L(n)$.

According to formulas~\er{eq.X.polynomial.second.order} and \er{eq.T.general}, 
the kernel of the homomorphism $\fd^{1}(\CE,a)\to\fd^0(\CE,a)$ from~\er{fdnn-1} 
is generated by the coefficients of the power series $Y_{ij}$, $Y_i$.
Recall that $\mathfrak{S}\subset\fd^{1}(\CE,a)$ is the ideal generated 
by the coefficients of $Y_{ij}$. 
Applying the projection $\fd^{1}(\CE,a)\to\fd^1(\CE,a)/\mathfrak{S}$ to the coefficients 
of the power series $Y_i$, we get ${\tilde{Y}}_i$. As has been shown above, 
$$
\fd^1(\CE,a)/\mathfrak{S}\cong L(n)\oplus\mathfrak{so}_{n-1},
$$
where $\mathfrak{so}_{n-1}$ coincides with the ideal generated by the coefficients of ${\tilde{Y}}_i$.
According to Theorem~\ref{thfd0}, the algebra $\fd^0(\CE,a)$ 
is isomorphic to $L(n)$.
These results show that the homomorphism $\fd^{1}(\CE,a)\to\fd^0(\CE,a)$ 
from~\er{fdnn-1} 
coincides with the composition~\er{fd1lnso}.
\end{proof}

\begin{remark}
\lb{rzcrnge}
We have two different ZCRs for the same PDE~\er{main},
which can be transformed to the PDE~\er{pt} by the transformation~\er{insp}. 
Namely, we have the $\mathfrak{gl}_{n+1}$-valued ZCR \er{M}, \er{N} 
and the $\mathfrak{so}_{n-1}$-valued ZCR described in Theorem~\ref{zcrson}.

One can embed the Lie algebras $\gl_{n+1}$ and $\so_{n-1}$ 
into the Lie algebra $\gl_\sm$ for some $\sm\ge n+1$, and then 
one can regard these ZCRs as $\gl_\sm$-valued ZCRs.
One can ask whether these ZCRs can become gauge equivalent 
after suitable embeddings $\gl_{n+1}\hookrightarrow\gl_\sm$ and  
$\so_{n-1}\hookrightarrow\gl_\sm$.

Let us show that these ZCRs cannot become gauge equivalent.

The $\mathfrak{gl}_{n+1}$-valued ZCR \er{M}, \er{N} is of order~$\le 0$.
By Theorem~\ref{evcov} and Remark~\ref{ranorm}, 
there is an $a$-normal $\mathfrak{gl}_{n+1}$-valued ZCR $\mathfrak{Z}$ 
of order~$\le 0$ such that the ZCR \er{M}, \er{N} 
is gauge equivalent to the ZCR $\mathfrak{Z}$.

Theorem~\ref{thaca} and Remark~\ref{remaca} imply the following.
If two $a$-normal ZCRs $\mathfrak{C}_1$ and $\mathfrak{C}_2$ are gauge 
equivalent and the ZCR $\mathfrak{C}_1$ is of order~$\le 0$, then 
the ZCR $\mathfrak{C}_2$ is also of order~$\le 0$.

The $\mathfrak{so}_{n-1}$-valued ZCR described in Theorem~\ref{zcrson} 
is $a$-normal and is not of order~$\le 0$, 
because the function $A$ in this ZCR depends nontrivially on $u^j_1$.
Therefore, this $\mathfrak{so}_{n-1}$-valued ZCR cannot become gauge equivalent 
to the $\mathfrak{gl}_{n+1}$-valued ZCR $\mathfrak{Z}$, 
after any embeddings $\gl_{n+1}\hookrightarrow\gl_\sm$ and 
$\so_{n-1}\hookrightarrow\gl_\sm$.

Since the ZCR \er{M}, \er{N} is gauge equivalent to the ZCR $\mathfrak{Z}$, 
we see that the $\mathfrak{so}_{n-1}$-valued ZCR described in Theorem~\ref{zcrson}
cannot become gauge equivalent to the 
$\mathfrak{gl}_{n+1}$-valued ZCR \er{M}, \er{N},
after any embeddings $\gl_{n+1}\hookrightarrow\gl_\sm$ and 
$\so_{n-1}\hookrightarrow\gl_\sm$.
\end{remark}

\section{The algebras $\fd^\oc(\CE,a)$ for the multicomponent  Landau-Lifshitz system}
\label{secfdk}

\subsection{Preliminary computations}

Recall that the infinite prolongation $\CE$ of system~\er{pt} 
is an infinite-dimensional manifold with the coordinates~\eqref{xtuikn1},
where $u^i_m$ corresponds to $\dfrac{\pd^m u^i}{\pd x^m}$ for $m\in\zp$ 
and $i=1,\dots,n-1$. In particular, $u^i_0=u^i$. 

Consider an arbitrary point $a\in\CE$ given by~\er{pointevfd1}.
As has been said in Section~\ref{prcomp},
since the PDE~\er{pt} is invariant with respect to the change of variables 
$x\mapsto x-x_a,\ t\mapsto t-t_a$, 
it is sufficient to consider the case $x_a=t_a=0$.

For simplicity of exposition, we continue to assume~\er{aik=0}, 
so we assume that $a^i_k=0$ in~\er{pointevfd1}.
(In the case $a^i_k\neq 0$, the computations change very little, 
and the final result is the same.)

Fix an integer $k\ge 2$.
In this section we compute the algebra $\fd^k(\CE,a)$ of the PDE~\er{pt}. 

According to Remark~\ref{rem_fdpgen} and assumptions~\er{x0t0=0},~\er{aik=0}, 
in order to describe the Lie algebra $\fd^k(\CE,a)$ for the PDE~\er{pt}, 
we need to study the equations
\beq
\lb{gck}
D_x(B)-D_t(A)+[A,B]=0,
\ee
\begin{gather}
\label{gd=0fk}
\forall\,i_0=1,\dots,n-1,\qquad\forall\,k_0\ge 1,\qquad\quad
\frac{\pd A}{\pd u^{i_0}_{k_0}}
\,\,\bigg|_{u^i_k=0\ \forall\,(i,k)\succ(i_0,k_0-1)}=0,\\
\lb{gaukakfk}
A\,\Big|_{u^i_k=0\ \forall\,(i,k)}=0,\\
\lb{gbxx0fk}
B\,\Big|_{x=0,\ u^i_k=0\ \forall\,(i,k)}=0.
\end{gather}
where 
\begin{itemize}
\item $A=A(x,t,u^j_0,u^j_1,\dots,u^j_k)$ is a power series in the variables 
$x$, $t$, $u^j_0$, $u^j_1,\dots,u^j_k$ for $j=1,\dots,n-1$,
\item $B=B(x,t,u^j_0,u^j_1,\dots,u^j_{k+2})$ is a power series in the variables 
$x$, $t$, $u^j_0$, $u^j_1,\dots,u^j_{k+2}$ for $j=1,\dots,n-1$. 
\end{itemize}
The coefficients of the power series $A$, $B$ 
are generators of the Lie algebra $\fd^k(\CE,a)$.
Relations for these generators are provided by equations \er{gck}, \er{gd=0fk},
\er{gaukakfk}, \er{gbxx0fk}.

In this section, 
summations over repeated indices run from $1$ to $n-1$ 
when referred to the number of dependent variables and from $0$ to $k$ when referred to the order of derivatives, 
unless otherwise specified. For instance, $u^i_{m+1}F^1_{u^i_m}$ means $\sum_{i=1}^{n-1}\sum_{m=0}^k u^i_{m+1}F^1_{u^i_m}$. 
The integer $k\ge 2$ is fixed throughout this section, and there is no summation over $k$.


Recall that system~\er{pt} is of the form 
\begin{equation}\label{eq.landau.bis.bis}
u^i_t=u^i_{3}+G^i(u^j,u^j_1,u^j_{2})=u^i_{3}+G^i(u^j_0,u^j_1,u^j_2),
\quad\qquad i=1,\dots,n-1.
\end{equation}

Then equation~\er{gck} reads 
\begin{equation}\label{eq.covering.higher}
B_x + \sum_{j=0}^{k+2}u^i_{j+1}B_{u^i_j} - A_t -\sum_{j=0}^{k} u^i_{j+3}A_{u^i_j} - \sum_{j=0}^{k} D_x^j(G^i)A_{u^i_j} + [A,B]=0.
\end{equation}
%
%

If we put equal to zero the coefficient 
of $u^i_{k+3}$ in~\er{eq.covering.higher}, we obtain that
$
B_{u^i_{k+2}}-A_{u^i_k}=0,
$
and, by integrating this, we see that $B$ is of the form
\beq
\lb{bkf1}
B=u^i_{k+2}A_{u^i_k} + F^1(x,t,u^i,u^i_1,\dots,u^i_{k+1})
\ee
for some power series $F^1(x,t,u^i,u^i_1,\dots,u^i_{k+1})$.
Taking into account the form of $A$ and $B$ above, 
we rewrite equation \eqref{eq.covering.higher} as
\begin{multline}\label{eq.covering.higher.2}
u^i_{k+2}A_{u^i_k x} + F^1_x + u^h_{j+1}u^i_{k+2} A_{u^i_ku^h_j} 
+ u^i_{k+3}A_{u^i_k} + \sum_{m=0}^{k+1}u^i_{m+1}F^1_{u^i_m}\\
 - A_t - u^i_{j+3}A_{u^i_j}  - D_x^j(G^i)A_{u^i_j}
+ u^i_{k+2}[A,A_{u^i_k}] + [A,F^1]=0.
\end{multline}
Now we compute the coefficient of $u^i_{k+2}$ 
of the left-hand side of \eqref{eq.covering.higher.2}. It is
$$
A_{u^i_k x} + u^h_{j+1}A_{u^i_ku^h_j} - A_{u^i_{k-1}} + 
F^1_{u^i_{k+1}} - G^j_{u^i_{2}}A_{u^j_k} + [A,A_{u^i_k}],
$$
and, by putting it equal to zero, one obtains the following system of PDEs
$$
F^1_{u^i_{k+1}} = 
A_{u^i_{k-1}} + A_{u^j_k}G^j_{u^i_{2}} - 
A_{u^i_k x} - A_{u^i_ku^h_j}u^h_{j+1} - [A,A_{u^i_k}],\qquad i=1,\dots, n-1.
$$
Integrating this system, we see that $F^1$ is of the form 
\begin{equation}\label{eq.F1.higher.2}
F^1=- \frac{1}{2}u^i_{k+1}u^h_{k+1}A_{u^i_ku^h_k} + u^i_{k+1}\Big(A_{u^i_{k-1}} + G^j_{u^i_{2}}A_{u^j_{k}} -\sum_{j=0}^{k-1}A_{u^i_ku^h_j}u^h_{j+1} - A_{u^i_k x} - [A,A_{u^i_k}]\Big) + F^2
\end{equation}
for some power series $F^2=F^2(x,t,u^i,u^i_1,\dots,u^i_{k})$.

In view of \eqref{eq.F1.higher.2},
the coefficient of $u^i_{k+1}u^j_{k+1}u^{\ai}_{k+1}$ in the left-hand side of 
\eqref{eq.covering.higher.2} is $A_{u^i_ku^j_ku^{\ai}_k}$ multiplied 
by a nonzero scalar.
Therefore, $A_{u^i_ku^j_ku^{\ai}_k}=0$ for all $i$, $j$, $\ai$. Hence $A$ is of the form
\begin{equation}\label{eq.X.higher}
A=\frac{1}{2}u^i_ku^j_kY^{k}_{ij} + u^i_kY^k_i + Y^k,\quad\qquad 
Y^{k}_{ij}=Y^{k}_{ji},
\end{equation}
where $Y^{k}_{ij}$, $Y^k_i$, $Y^k$ are power series in the variables 
$x$, $t$, $u^i_m$, $m=0,\dots, k-1$.

Similarly to~\er{condyi}, 
combining condition~\er{gd=0fk} with formula~\er{eq.X.higher}, we obtain
\beq
\label{condyik}
\forall\,i=1,\dots,n-2,\qquad Y^k_i\,\Big|_{u^{i+1}_{k-1}=u^{i+2}_{k-1}=\dots=u^{n-1}_{k-1}=0}=0,\qquad Y^k_{n-1}=0.
\ee

\subsection{The ideal generated by $Y^{k}_{ij}$} 
\lb{secykij}

The coefficient of $u^i_{k+1}u^h_{k+1}$ of the left-hand side of \eqref{eq.covering.higher.2}, 
taking into account that $D_x^j(G^i)$ does not contribute to this coefficient because $k\geq 2$, is
\begin{equation}\label{eq.nonloso.higher}
-3 Y^{k}_{ih,x} -3\sum_{j=0}^{k-1} u^{\ai}_{j+1}Y^{k}_{ih,u^{\ai}_j} + G^j_{u^i_{2}}Y^{k}_{jh} + G^j_{u^h_{2}}Y^{k}_{ji} - 3[A,Y^{k}_{ih}].
\end{equation}
Note that \eqref{eq.nonloso.higher} is a polynomial of degree $2$ in $u^i_k$. 
Equating the coefficient of $u^{\ai}_ku^{\bi}_k$ of \eqref{eq.nonloso.higher} to zero, we get
\begin{equation}\label{eq.rel.kkab.kkih}
[Y^{k}_{{\ai}{\bi}},Y^{k}_{ih}]=0\qquad\quad\forall\,{\ai},{\bi},i,h.
\end{equation}
Now, to compute the coefficient of $u^i_{k+1}u^h_{k+1}u^{\ai}_k$  of \eqref{eq.covering.higher.2}, 
it is enough to look at the coefficient of $u^{\ai}_k$ of \eqref{eq.nonloso.higher}. 
Note that the quantity $G^j_{u^i_{2}}Y^{k}_{jh}+G^j_{u^h_{2}}Y^{k}_{ji}$ does not depend on $u^{\ai}_k$, 
since we suppose $k\geq 2$. Then we have only the system
\begin{equation}\label{eq.Ykk.ih.pa.k.meno.1}
Y^{k}_{ih,u^{\ai}_{k-1}}=[Y^{k}_{ih},Y^k_{\ai}]\qquad\quad\forall\,{\ai},i,h.
\end{equation}
We regard~\er{eq.Ykk.ih.pa.k.meno.1} as an overdetermined system of PDEs for $Y^{k}_{ih}$.
The compatibility conditions of this system give the relations
\begin{equation*}
[T^{k}_{{\ai}{\bi}},Y^{k}_{ih}]=0\qquad\quad\forall\,{\ai},{\bi},i,h,
\end{equation*}
where
\begin{equation}\label{equ.def.Tkkab}
T^{k}_{{\ai}{\bi}}:=[Y^k_{\ai},Y^k_{\bi}]-Y^k_{{\ai},u^{\bi}_{k-1}} + Y^k_{{\bi},u^{\ai}_{k-1}}.
\end{equation}
\begin{lemma}
\lb{cykitab}
The Lie subalgebra generated by the coefficients of the power series~$Y^k_i$ 
coincides with the Lie subalgebra generated by the coefficients of~$T^{k}_{{\ai}{\bi}}$.
\end{lemma}
\begin{proof}
The statement follows from Lemma~\ref{zi} applied to~\er{condyik},~\er{equ.def.Tkkab}.
\end{proof}

The coefficient of $u^i_{k+1}u^h_{k+1}$ of \eqref{eq.covering.higher.2} 
is the coefficient of degree zero in $u^{\ai}_k$ of \eqref{eq.nonloso.higher}.
Equating this coefficient to zero, we obtain
\begin{equation}\label{eq.nonloso.higher.2}
-3 Y^{k}_{ih,x} -3\sum_{j=0}^{k-2} u^{\ai}_{j+1} Y^{k}_{ih,u^{\ai}_j} 
+ G^j_{u^i_{2}}Y^{k}_{jh} + G^j_{u^h_{2}}Y^{k}_{ji} - 3[Y^k,Y^{k}_{ih}]=0.
\end{equation}

\begin{lemma}
\lb{ykihyka}
The power series $Y^{k}_{ih}$ does not depend on $u^{\ai}_{\bi}$ for ${\bi}>0$. 
So $Y^{k}_{ih}$ depends only on $u^{\ai}_0=u^{\ai}$ and $x,t$.
\end{lemma} 
\begin{proof}
This is proved by induction on~$l=k-{\bi}$, 
using conditions~\er{gd=0fk},~\er{gaukakfk} 
and equations~\er{eq.Ykk.ih.pa.k.meno.1},~\er{eq.nonloso.higher.2}.
\end{proof}

In view of Lemma~\ref{ykihyka}, equation \eqref{eq.nonloso.higher.2} becomes
\begin{equation}\label{eq.nonloso.higher.3}
-3 Y^{k}_{ih,x} -3 u^{\ai}_{1} Y^{k}_{ih,{\ai}} + 
G^j_{u^i_{2}}Y^{k}_{jh} + G^j_{u^h_{2}}Y^{k}_{ji} - 3[Y^k,Y^{k}_{ih}]=0,
\end{equation}
where $Y^{k}_{ih,{\ai}}=Y^{k}_{ih,u^{\ai}}$.

Furthermore, by differentiating \eqref{eq.nonloso.higher.3}  
with respect to $u^{\ai}_l$ with $l\geq 2$ and with respect to $u^{\ai}_1$, $u^{\bi}_1$ we get, respectively, the following relations
$$
[Y^k_{,u^{\ai}_l},Y^{k}_{ij}]=0\qquad\forall\,l\geq 2, 
\qquad [Y^k_{,u^{\ai}_1u^{\bi}_1},Y^{k}_{ih}]=0.
$$
Therefore, 
the quantity $[Y^k,Y^{k}_{ih}]$ is a polynomial of degree~$\le 1$ in $u^{\ai}_1$.
Its coefficient in $u^{\ai}_1$ is $[Y^k_{,u^{\ai}_1},Y^{k}_{ih}]$ 
(recall that $Y^{k}_{ih}$ depends only on $x$, $t$, $u^{\ai}$) 
and its zero degree term is $[Y^k-Y^k_{,u^{\ai}_1}u^{\ai}_1,Y^{k}_{ih}]$. 
By putting equal to zero the coefficient of $u^{\ai}_1$ and 
that of degree zero (in $u^{\ai}_1$) of the left-hand side of~\eqref{eq.nonloso.higher.3}, 
we get the system
\begin{gather}
\label{Ykijhy}
Y^{k}_{ij,h}=  [Y^{k}_{ij},Y^k_{,u^h_1}] - 2\Delta^{-1}\big(2u^h Y^{k}_{ij}+ u^i Y^{k}_{hj} + u^j Y^{k}_{hi} - \delta^i_h u^m Y^{k}_{mj} - \delta^j_h u^m Y^{k}_{mi}\big)\qquad
\forall\,i,j,h,\\
\notag
Y^{k}_{ij,x} = [Y^{k}_{ij},Z],\qquad\quad
Z=Y^k-u^{\ai}_1Y^k_{,u^{\ai}_1}.
\end{gather}
Similarly to~\er{eq.comp.cond.system.defining.Yij}, 
the compatibility conditions of system~\er{Ykijhy} imply
\begin{equation}\label{eq.comp.cond.system.defining.Ykkij.higher}
[Y^{k}_{ij},W_{mh}] + 
4\Delta^{-2}( -\delta^i_mY^{k}_{hj} + 
\delta^i_hY^{k}_{mj} - \delta^j_mY^{k}_{hi} + \delta^j_hY^{k}_{mi} )=0\qquad\quad
\forall\,i,j,m,h,
\end{equation}
where
\begin{equation}\label{eq.def.Wmh.Rh}
W_{mh}:=[Y^k_{,u^m_1},Y^k_{,u^h_1}] - Y^k_{,u^m_1u^h} + Y^k_{,u^h_1u^m}.
\end{equation}
%
%

Recall that in Section~\ref{secyij} we have defined the homomorphism 
$\hmm\cl\fd^1(\CE,a)\to\mathfrak{so}_{n-1}$.
Let $\hmm_k\cl\fd^k(\CE,a)\to\mathfrak{so}_{n-1}$ be the composition 
of the homomorphisms $\fd^k(\CE,a)\to\fd^1(\CE,a)\to\mathfrak{so}_{n-1}$, 
where the homomorphism $\fd^k(\CE,a)\to\fd^1(\CE,a)$ arises from~\er{fdnn-1}.

\begin{theorem}
\lb{idealyijk}
The Lie subalgebra $\mathfrak{S}_k$ generated by 
the coefficients of~$Y^{k}_{ij}$ is abelian and satisfies
\beq
\notag
[\mathfrak{S}_k,\ker\hmm_k]=0.
\ee
Furthermore, this subalgebra is an ideal of $\fd^k(\CE,a)$. 
\end{theorem}
\begin{proof}
This is proved similarly to Theorem~\ref{idealyij}, 
using the results of this section. 
\end{proof}

\subsection{The ideal generated by $Y^{k}_{i}$} 
\lb{secyki}

In this subsection we study the quotient Lie algebra $\bl_k=\fd^k(\CE,a)/\mathfrak{S}_k$, 
where $\mathfrak{S}_k$ is the ideal generated by the coefficients of~$Y^{k}_{ij}$. 
Therefore, we can assume that $A$, $B$ are power series with coefficients in $\bl_k$ 
and 
\beq
\label{auikyy}
A=u^i_kY^k_i+Y^k,
\ee
where $Y^k_i$ and $Y^k$ depend on $x,t,u^i,\dots, u^i_{k-1}$. 

Then from~\er{bkf1},~\er{eq.covering.higher.2},~\er{eq.F1.higher.2} one gets 
\begin{gather}
\notag
B=u^i_{k+2}Y^k + F^1,\\
\label{f1uik1}
F^1=u^i_{k+1}\Big(A_{u^i_{k-1}} + G^j_{u^i_{2}}A_{u^j_{k}} 
-\sum_{j=0}^{k-1}A_{u^i_ku^h_j}u^h_{j+1} - A_{u^i_k x} - [A,A_{u^i_k}]\Big) + F^2,\\
\label{eq.covering.higher.3}
F^1_x + \sum_{m=0}^{k}u^i_{m+1}F^1_{u^i_m} - A_{t} 
-\sum_{j=0}^{k-2}u^i_{j+3}A_{u^i_j} - D_x^j(G^i)A_{u^i_j} + [A,F^1]=0.
\end{gather}
Substituting \er{auikyy} and \er{f1uik1} in~\er{eq.covering.higher.3}, 
we see that $F^2$ is a polynomial of degree $\le 3$ in $u^1_k,\dots,u^{n-1}_k$.
So $F^2$ is of the form 
$$
F^2=u^i_ku^j_ku^h_k M^{k}_{ijh} + u^i_ku^j_k M^{k}_{ij}+ u^i_kM^{k}_{i} + M^k,
$$
where all the $M^{k}_{...}$ are fully symmetric with respect to the lower indices 
and depend on $x$, $t$, $u^i_{\bi}$, ${\bi}=0,\dots,k-1$.

Differentiating \eqref{eq.covering.higher.3} 
with respect to $u^i_{k+1}$, $u^{\ai}_{k}$, $u^{\bi}_{k}$, we get
\begin{equation}
\label{ttytyt}
T^{k}_{i{\ai},u^{\bi}_{k-1}} + T^{k}_{i{\bi},u^{\ai}_{k-1}} + 
[Y^k_{\ai},T^{k}_{i{\bi}}] + [Y^k_{\bi},T^{k}_{i{\ai}}] 
+ 6M^{k}_{i{\ai}{\bi}}=0
\qquad\quad\forall\,i,{\ai},{\bi}.
\end{equation}
Interchanging the indices $i$ and ${\ai}$ in~\er{ttytyt}, one obtains
\begin{equation}
\label{ttytytii}
T^{k}_{{\ai}{i},u^{\bi}_{k-1}} + T^{k}_{{\ai}{\bi},u^{i}_{k-1}} + 
[Y^k_{i},T^{k}_{{\ai}{\bi}}] + [Y^k_{\bi},T^{k}_{{\ai}{i}}] 
+ 6M^{k}_{{\ai}{i}{\bi}}=0
\qquad\quad\forall\,{\ai},{i},{\bi}.
\end{equation}
Subtracting~\er{ttytytii} from~\er{ttytyt} and using the relations
$M^{k}_{i{\ai}{\bi}}=M^{k}_{{\ai}{i}{\bi}}$, 
$T^{k}_{{\ai}{i}}=-T^{k}_{i{\ai}}$, $T^{k}_{{\ai}{\bi}}=-T^{k}_{{\bi}{\ai}}$, 
we get 
\begin{equation}\label{equ.Tkkibuakmeno1}
T^{k}_{i{\bi},u^{\ai}_{k-1}} + T^{k}_{{\bi}{\ai},u^i_{k-1}} + 
2T^{k}_{i{\ai},u^{\bi}_{k-1}} + [Y^k_{\ai},T^{k}_{i{\bi}}] + 
[Y^k_i,T^{k}_{{\bi}{\ai}}] + 2[Y^k_{\bi},T^{k}_{i{\ai}}]=0.
\end{equation}
Set
$$
V^k_{\ai}=Y^k_{\ai}+\frac{\partial}{\partial u^{\ai}_{k-1}}.
$$
Then from \eqref{equ.def.Tkkab} we have 
$T^{k}_{{\ai}{\bi}}=[V^k_{\ai},V^k_{\bi}]$, 
and equation \eqref{equ.Tkkibuakmeno1} 
can be written in terms of commutators of $V^k_{\ai}$ 
\beq
\lb{vvvvvv}
[V^k_{\ai},[V^k_i,V^k_{\bi}]] + [V^k_i,[V^k_{\bi},V^k_{\ai}]] 
+ 2[V^k_{\bi},[V^k_i,V^k_{\ai}]]=0.
\ee
By the Jacobi identity, equation~\er{vvvvvv} can be rewritten as
\beq
\lb{vvv0}
3[V^k_{\bi},[V^k_i,V^k_{\ai}]]=0.
\ee
Since $[V^k_i,V^k_{\ai}]=T^{k}_{i{\ai}}$ and 
$V^k_{\bi}=Y^k_{\bi}+\dfrac{\partial}{\partial u^{\bi}_{k-1}}$, 
equation~\er{vvv0} says that
\begin{equation}
\lb{tkiak1}
T^{k}_{i{\ai},u^{\bi}_{k-1}} + [Y^k_{\bi},T^{k}_{i{\ai}}]=0\qquad\quad\forall\,i,{\ai},{\bi}.
\end{equation}

Differentiating \eqref{eq.covering.higher.3} with respect to $u^i_{k+1}$, $u^{\ai}_{k}$,
we obtain 
\begin{multline}
\lb{ykmkia}
Y^k_{{\ai},u^i_{k-1}x} - 2Y^k_{i,u^{\ai}_{k-1}x} + [Y^k_i,Y^k_{\ai}]_{,x} + [Y^k_{i,x},Y^k_{\ai}]
-Y^k_{{\ai},u^i_{k-2}} +\sum_{m=0}^{k-2}T^{k}_{i{\ai},u^{\bi}_m}u^{\bi}_{m+1} 
+ G^j_{u^i_{2}}T^{k}_{{\ai}j} + \\
+ [Y^k,T^{k}_{i{\ai}}]+ Y^k_{,u^i_{k-1}u^{\ai}_{k-1}} - \sum_{m=0}^{k-2}Y^k_{i,u^h_mu^{\ai}_{k-1}}u^h_{m+1}
- Y^k_{i,u^{\ai}_{k-2}} - [Y^k_{,u^{\ai}_{k-1}},Y^k_i] - [Y^k,Y^k_{i,u^{\ai}_{k-1}}]+\\
+[Y^k_{\ai},Y^k_{,u^i_{k-1}}] - \sum_{m=0}^{k-2}[Y^k_{\ai},Y^k_{i,u^h_m}]u^h_{m+1} - [Y^k_{\ai},[Y^k,Y^k_i]] + 2M^{k}_{i{\ai}}=0\qquad\quad\forall\,i,{\ai}.
\end{multline}
Denote the left-hand side of~\er{ykmkia} by $Z^k_{i{\ai}}$. 
So equation~\er{ykmkia} reads $Z^k_{i{\ai}}=0$ for all $i$, ${\ai}$.
Interchanging the indices $i$ and ${\ai}$, one gets also $Z^k_{{\ai}i}=0$.

Subtracting the equation $Z^k_{{\ai}i}=0$ from the equation $Z^k_{i{\ai}}=0$
and using the relation $M^{k}_{i{\ai}}=M^{k}_{{\ai}i}$, we obtain
\begin{equation}
\label{eq.comp.Akkia}
3T^{k}_{i{\ai},x} + 3\sum_{m=0}^{k-2}T^{k}_{i{\ai},u^{\bi}_m}u^{\bi}_{m+1} 
+ G^h_{u^i_{2}}T^{k}_{{\ai}h} - G^h_{u^{\ai}_{2}}T^{k}_{ih} + 3[Y^k,T^{k}_{i{\ai}}]=0.
\end{equation}

\begin{lemma}
\lb{tkiaujm}
The power series $T^{k}_{i{\ai}}$ does not depend on $u^j_m$ for $m>0$.
So $T^{k}_{i{\ai}}$ depends only on $u^j_0=u^j$ and $x, t$.
\end{lemma} 
\begin{proof}
This is proved by induction on~$l=k-m$, 
using conditions~\er{gd=0fk},~\er{gaukakfk} 
and equations~\er{tkiak1},~\er{eq.comp.Akkia}.
\end{proof}

Differentiating \eqref{eq.comp.Akkia} with respect to $u^{\ai}_l$ with $l\geq 2$ 
and with respect to $u^{\ai}_1,u^{\bi}_1$ 
we get, respectively, the following relations 
$$
[Y^k_{,u^{\ai}_l},T^{k}_{ij}]=0\qquad
l\geq 2, 
\quad\qquad [Y^k_{,u^{\ai}_1u^{\bi}_1},T^{k}_{ih}]=0.
$$
Therefore, 
the left-hand side of \eqref{eq.comp.Akkia} is a polynomial of degree $1$ in $u^{\ai}_1$. 
Equating to zero the coefficients of this polynomial, taking into account \eqref{eq.formule.Gpxx}, we obtain the system
\begin{gather}
\label{system.defining.Tkkij.bis}
T^{k}_{i{\ai},u^j} = 
2\Delta^{-1}\left(2u^jT^{k}_{{\ai}i} + u^{\ai}T^{k}_{ji} + u^iT^{k}_{{\ai}j} + 
\delta^i_j u^mT^{k}_{m{\ai}} + \delta^{\ai}_j u^mT^{k}_{im}\right) + [T^{k}_{i{\ai}},Y^k_{,u^j_1}],\qquad\quad
\forall\,i,{\ai},j,\\
\notag
T^{k}_{i{\ai},x} = [T^{k}_{i{\ai}},Z],\qquad\quad 
Z=Y^k-u^{\bi}_1Y^k_{,u^{\bi}_1}.
\end{gather}
The compatibility conditions of  system~\eqref{system.defining.Tkkij.bis} imply 
\begin{equation}\label{eq.comp.cond.system.defining.Tkkia.higher}
[T^{k}_{i{\ai}},W_{jm}] = 4\Delta^{-2}
(\delta^{\ai}_mT^{k}_{ji}+\delta^{\ai}_jT^{k}_{im}+\delta^i_mT^{k}_{{\ai}j}+\delta^i_jT^{k}_{m{\ai}})
\end{equation}
with $W_{jm}$ defined by~\eqref{eq.def.Wmh.Rh}.

Let $\psi_k\cl\fd^k(\CE,a)\to\fd^1(\CE,a)$ be the 
homomorphism that arises from~\er{fdnn-1}.
Formulas~\er{eq.X.higher} and~\er{bkf1} imply that $\psi_k$ 
maps each coefficient of~$Y^{k}_{ij}$ to zero.
Since $\mathfrak{S}_k$ is generated by the coefficients of~$Y^{k}_{ij}$, 
we get $\psi_k(\mathfrak{S}_k)=0$.
Therefore, $\psi_k\cl\fd^k(\CE,a)\to\fd^1(\CE,a)$ induces a homomorphism
from $\fd^k(\CE,a)/\mathfrak{S}_k$ to $\fd^1(\CE,a)$, which we denote 
by $\hat{\psi}_k\cl\fd^k(\CE,a)/\mathfrak{S}_k\to\fd^1(\CE,a)$

Recall that in Section~\ref{secyij} we have defined the homomorphism 
$\hmm\cl\fd^1(\CE,a)\to\mathfrak{so}_{n-1}$.
Let 
$$
\hat{\hmm}_k\cl\fd^k(\CE,a)/\mathfrak{S}_k\to\mathfrak{so}_{n-1}
$$ 
be the composition of the homomorphisms $\hmm\cl\fd^1(\CE,a)\to\mathfrak{so}_{n-1}$ 
and $\hat{\psi}_k\cl\fd^k(\CE,a)/\mathfrak{S}_k\to\fd^1(\CE,a)$.

\begin{theorem}
\lb{idealyik}
Recall that $k\ge 2$.
Let $\mathfrak{H}_k\subset\fd^k(\CE,a)/\mathfrak{S}_k$ 
be the Lie subalgebra generated by the coefficients of~$Y^{k}_{i}$.
Then $\mathfrak{H}_k$ is abelian and satisfies
\beq
\lb{hkkervfk}
[\mathfrak{H}_k,\ker\hat{\hmm}_k]=0.
\ee
Furthermore, this subalgebra is an ideal of the Lie algebra $\fd^k(\CE,a)/\mathfrak{S}_k$. 
\end{theorem}
\begin{proof}
According to Lemma~\ref{cykitab},  
the Lie subalgebra generated by the coefficients of the power series~$Y^k_i$ 
coincides with the Lie subalgebra generated by the coefficients of~$T^{k}_{{\ai}{\bi}}$.

Therefore, it remains to prove that 
the subalgebra generated by the coefficients of~$T^{k}_{{\ai}{\bi}}$ 
is an abelian ideal and satisfies~\er{hkkervfk}. 
This is proved similarly to Theorem~\ref{idealyij}, 
using the results of the present section. 
\end{proof}

\subsection{The structure of the algebras $\fd^\oc(\CE,a)$}
\lb{subsfd}

\begin{theorem}
\label{fadesc}
Let $n\ge 4$.
Recall that the $n$-component Landau-Lifshitz system~\er{main} 
is transformed to the PDE~\er{pt} by means of the change of variables~\er{sp}.

The Lie algebras $\fd^\oc(\CE,a)$, $\oc\in\zp$, for the PDE~\er{pt} have the following structure.

The algebra $\fd^0(\CE,a)$ is isomorphic to $L(n)$.

There is an abelian ideal $\mathfrak{S}$ of $\fd^{1}(\CE,a)$ 
such that $\fd^{1}(\CE,a)/\mathfrak{S}\cong L(n)\oplus\mathfrak{so}_{n-1}$, 
where $\mathfrak{so}_{n-1}$ is the Lie algebra 
of skew-symmetric $(n-1)\times(n-1)$ matrices. 
The homomorphism $\fd^{1}(\CE,a)\to\fd^0(\CE,a)$ from~\er{fdnn-1} coincides 
with the composition of the homomorphisms
\beq
\lb{fd1fd1s}
\notag
\fd^{1}(\CE,a)\to\fd^{1}(\CE,a)/\mathfrak{S}\cong L(n)\oplus\mathfrak{so}_{n-1}
\to L(n)\cong \fd^0(\CE,a). 
\ee

Let $\tau_k\cl\fd^{k}(\CE,a)\to\fd^{k-1}(\CE,a)$ 
be the surjective homomorphism from~\er{fdnn-1}.
Then for any $k\ge 2$ we have
\beq
\lb{vvv123}
[v_1,[v_2,v_3]]=0\qquad\quad
\forall\,v_1,v_2,v_3\in\ker\tau_k.
\ee
In particular, the kernel of $\tau_k$ is nilpotent.

For each $k\ge 1$, 
let $\vf_k\colon\fd^k(\CE,a)\,\to\,L(n)\oplus\mathfrak{so}_{n-1}$ 
be the composition of the homomorphisms 
\beq
\lb{fdkso}
\fd^{k}(\CE,a)\to\fd^{1}(\CE,a)\to\fd^{1}(\CE,a)/\mathfrak{S}\cong L(n)\oplus\mathfrak{so}_{n-1},
\ee
where $\fd^k(\CE,a)\to\fd^1(\CE,a)$ arises from~\er{fdnn-1}.
Then 
\beq
\lb{mllhk}
[h_1,[h_2,\dots,[h_{2k-2},[h_{2k-1},h_{2k}]]\dots]]=0\qquad\qquad
\forall\,h_1,h_2,\dots,h_{2k}\in\ker\vf_k.
\ee
In particular, the kernel of $\vf_k$ is nilpotent.
\end{theorem}
\begin{proof}
The isomorphism $\fd^0(\CE,a)\cong L(n)$ is proved in Theorem~\ref{thfd0}. 

In particular, the abelian ideal $\mathfrak{S}\subset\fd^{1}(\CE,a)$ is defined 
in Theorem~\ref{idealyij}, and the isomorphism 
$\fd^{1}(\CE,a)/\mathfrak{S}\cong L(n)\oplus\mathfrak{so}_{n-1}$ 
is described in Theorem~\ref{tijnge4}.
Also, in Theorem~\ref{tijnge4} it is shown that the homomorphism 
$\fd^{1}(\CE,a)\to\fd^0(\CE,a)$ from~\er{fdnn-1} coincides with the composition~\er{fd1fd1s}.



Recall that in Section~\ref{secyij} we have defined the homomorphism 
$\hmm\cl\fd^1(\CE,a)\to\mathfrak{so}_{n-1}$.
For any $r\ge 1$, let $\hmm_r\cl\fd^r(\CE,a)\to\mathfrak{so}_{n-1}$ be the composition 
of the homomorphisms $\fd^r(\CE,a)\to\fd^1(\CE,a)\to\mathfrak{so}_{n-1}$,
where $\fd^r(\CE,a)\to\fd^1(\CE,a)$ arises from~\er{fdnn-1}.

Since for any $k\ge 2$ the composition 
$$
\fd^k(\CE,a)\xrightarrow{\tau_k}\fd^{k-1}(\CE,a)
\xrightarrow{\hmm_{k-1}}\mathfrak{so}_{n-1}
$$
coincides with $\hmm_k\cl\fd^k(\CE,a)\to\mathfrak{so}_{n-1}$, we have 
\beq
\lb{ktkhm}
\ker\tau_k\subset\ker\hmm_k\qquad\quad\forall\,k\ge 2.
\ee
Formulas~\er{eq.X.higher},~\er{bkf1} imply that $\ker\tau_k$ 
is generated by the coefficients of the power series $Y^{k}_{ij}$, $Y^{k}_{i}$.

Recall that $\mathfrak{S}_k\subset\fd^k(\CE,a)$ is the ideal 
generated by the coefficients of~$Y^{k}_{ij}$ for any $k\ge 2$.
According to Theorem~\ref{idealyijk}, 
\beq
\lb{skkh0}
[\mathfrak{S}_k,\ker\hmm_k]=0\qquad\quad\forall\,k\ge 2.
\ee
Since $\ker\tau_k\subset\ker\hmm_k$, we get 
\beq
\lb{skktk0}
[\mathfrak{S}_k,\ker\tau_k]=0\qquad\quad\forall\,k\ge 2.
\ee
Theorem~\ref{idealyik} implies 
\beq
\lb{ktkhs}
[\ker\tau_k,\ker\hmm_k]\subset\mathfrak{S}_k\qquad\quad\forall\,k\ge 2.
\ee
Combining \er{ktkhm}, \er{skktk0}, \er{ktkhs}, one obtains
$[\ker\tau_k,\ker\tau_k]\subset\mathfrak{S}_k$ and $[\mathfrak{S}_k,\ker\tau_k]=0$.
This yields~\er{vvv123} and implies that $\ker\tau_k$ is nilpotent.

Since $\vf_k\colon\fd^k(\CE,a)\,\to\,L(n)\oplus\mathfrak{so}_{n-1}$ 
is the composition~\er{fdkso}, we have
\beq
\lb{kvfkhm}
\ker\vf_k\subset\ker\hmm_k\qquad\quad\forall\,k\ge 2.
\ee

For $k=1$, the homomorphism 
$\vf_1\colon\fd^1(\CE,a)\,\to\,L(n)\oplus\mathfrak{so}_{n-1}$ 
is the composition of the homomorphisms 
\beq
\lb{vf1com}
\fd^{1}(\CE,a)\to\fd^{1}(\CE,a)/\mathfrak{S}\cong L(n)\oplus\mathfrak{so}_{n-1},
\ee
hence $\ker\vf_1=\mathfrak{S}\subset\fd^{1}(\CE,a)$.
According to Theorem~\ref{idealyij}, the ideal $\mathfrak{S}$ is abelian,
so $[h_1,h_2]=0$ for all $h_1,h_2\in\ker\vf_1=\mathfrak{S}$. 
This proves~\er{mllhk} for $k=1$.

Now let us prove~\er{mllhk} by induction on $k\ge 2$.

In the case $k=2$, consider any elements
\beq
\lb{hhvf2}
h_1,h_2,h_3,h_4\in\ker\vf_2\subset\fd^2(\CE,a).
\ee
According to the definition of $\vf_k$,
the homomorphism $\vf_2$ is the composition of the homomorphisms 
$$
\fd^2(\CE,a)\xrightarrow{\tau_2}\fd^{1}(\CE,a)
\to\fd^{1}(\CE,a)/\mathfrak{S}\cong L(n)\oplus\mathfrak{so}_{n-1}.
$$
Therefore, the condition $h_3,h_4\in\ker\vf_2$ means that 
$\tau_2(h_3),\tau_2(h_4)\in\mathfrak{S}$.
Since $[\mathfrak{S},\mathfrak{S}]=0$, we have $\tau_2([h_3,h_4])=0$, 
so $[h_3,h_4]\in\ker\tau_2$.

From~\er{kvfkhm}, \er{hhvf2} it follows that $h_1,h_2\in\ker\hmm_2$.
According to~\er{ktkhs}, since $h_2\in\ker\hmm_2$ and $[h_3,h_4]\in\ker\tau_2$, 
we have $[h_2,[h_3,h_4]]\in\mathfrak{S}_2$.
According to~\er{skkh0}, since $h_1\in\ker\hmm_2$ and $[h_2,[h_3,h_4]]\in\mathfrak{S}_2$, 
we get $[h_1,[h_2,[h_3,h_4]]]=0$, which means that we have proved~\er{mllhk} for $k=2$.

Let $r\ge 2$ be such that~\er{mllhk} is valid for $k=r$. 
Then for any elements 
$$
h_1,h_2,h_3,\dots,h_{2r+2}\in\ker\vf_{r+1}\subset\fd^{r+1}(\CE,a)
$$ 
we have 
\beq
\lb{tr1h}
\big[\tau_{r+1}(h_3),\big[\tau_{r+1}(h_4),\dots,\big[\tau_{r+1}(h_{2r}),
\big[\tau_{r+1}(h_{2r+1}),\tau_{r+1}(h_{2r+2})\big]\big]\dots\big]\big]=0,
\ee
because $\tau_{r+1}(h_i)\in\ker\vf_r$ for $i=3,4,\dots,2r+2$.
Equation~\er{tr1h} says that 
\beq
\lb{h3h2r1}
[h_3,[h_4,\dots,[h_{2r},[h_{2r+1},h_{2r+2}]]\dots]]\in\ker\tau_{r+1}.
\ee
According to~\er{kvfkhm}, since $h_2\in\ker\vf_{r+1}$, we have
$h_2\in\ker\hmm_{r+1}$. Combining this with \er{ktkhs} and \er{h3h2r1},
one obtains
\beq
\lb{h23h2r1}
[h_2,[h_3,[h_4,\dots,[h_{2r},[h_{2r+1},h_{2r+2}]]\dots]]]\in\mathfrak{S}_{r+1}.
\ee
According to~\er{kvfkhm}, since $h_1\in\ker\vf_{r+1}$, we have
$h_1\in\ker\hmm_{r+1}$. Combining this with~\er{skkh0} and \er{h23h2r1}, one gets
\beq
\notag
[h_1,[h_2,[h_3,\dots,[h_{2r},[h_{2r+1},h_{2r+2}]]\dots]]]=0.
\ee
Thus we have proved~\er{mllhk} for $k=r+1$. 
Clearly, property~\er{mllhk} implies that $\ker\vf_k$ is nilpotent.
\end{proof}

\section{On zero-curvature representations 
for the multicomponent Landau-Lifshitz system}
\lb{secczcr}

In this section we use the notions of reductions and direct sums of ZCRs.
These notions have been introduced in Section~\ref{sbdszcr}.

Let $n\ge 4$.
Recall that the $n$-component Landau-Lifshitz system~\er{main} 
is transformed to the PDE~\er{pt} by means of the change of variables~\er{sp}.
In what follows we study ZCRs for this PDE.

Let $\mathfrak{Z}_1$ be the $L(n)$-valued ZCR defined in Remark~\ref{mllzcr}. 
So the ZCR $\mathfrak{Z}_1$ is given by the functions~\er{M},~\er{N}, 
which take values in the infinite-dimensional Lie algebra $L(n)$ 
introduced in Remark~\ref{mllzcr}.

Let $\mathfrak{Z}_2$ be the $\mathfrak{so}_{n-1}$-valued ZCR defined in 
Theorem~\ref{zcrson}.
We can consider also the direct sum $\mathfrak{Z}_1\oplus\mathfrak{Z}_2$ 
of the ZCRs $\mathfrak{Z}_1$ and $\mathfrak{Z}_2$. 
The ZCR $\mathfrak{Z}_1\oplus\mathfrak{Z}_2$ takes values 
in the Lie algebra $L(n)\oplus\mathfrak{so}_{n-1}$.

For any given ZCR $\mathfrak{R}$ 
of the PDE~\er{pt}, we are going to show that, 
after suitable gauge transformations and after killing 
some nilpotent ideal in the corresponding Lie algebra, 
the ZCR $\mathfrak{R}$ becomes isomorphic to a reduction 
of the ZCR $\mathfrak{Z}_1\oplus\mathfrak{Z}_2$.

Let $\CE$ be the infinite prolongation of the PDE~\er{pt}.
Fix a point $a\in\CE$ given by~\er{pointevfd1}.
We are going to study ZCRs defined on a neighborhood of $a\in\CE$.

Fix $\oc\in\zp$. In what follows, we study ZCRs of order~$\le\oc$.
Without loss of generality, we can assume $\oc\ge 1$, 
because any ZCR of order~$\le 0$ is at the same time of order~$\le 1$.

Recall that the theory of Section~\ref{csev} has been developed
for an arbitrary $\nv$-component PDE~\er{uitfi} such that 
the right-hand side of~\er{uitfi} may depend 
on $x$, $t$, $u^j$, $u^j_k$ for $k\le\eo$.
Since we study now the $(n-1)$-component PDE~\er{pt} of order~$3$, 
we will use the theory of Section~\ref{csev} for $\nv=n-1$ and $\eo=3$.

According to~\er{gasumxt},~\er{gbsumxt}, one has the power series 
$\ga=\ga(x,t,u^j_0,u^j_1,\dots,u^j_\oc)$, 
$\gb=\gb(x,t,u^j_0,u^j_1,\dots,u^j_{\oc+2})$ with coefficients in $\fd^\oc(\CE,a)$
satisfying~\er{xgbtga}, so $\ga$, $\gb$ constitute a formal ZCR with coefficients 
in $\fd^\oc(\CE,a)$.

According to Theorem~\ref{fadesc}, we have the surjective homomorphism
\beq
\lb{vflnso}
\vf_\oc\colon\fd^\oc(\CE,a)\,\to\,L(n)\oplus\mathfrak{so}_{n-1}
\ee
such that $\ker\vf_\oc$ is a nilpotent ideal of $\fd^\oc(\CE,a)$.
Then 
\beq
\lb{chgagb}
\check{\ga}=\vf_\oc(\ga),\qquad\quad 
\check{\gb}=\vf_\oc(\gb)
\ee
are power series with coefficients in $L(n)\oplus\mathfrak{so}_{n-1}$ 
and constitute a formal ZCR with coefficients in $L(n)\oplus\mathfrak{so}_{n-1}$.

According to Remark~\ref{rzfz}, 
the $\big(L(n)\oplus\mathfrak{so}_{n-1}\big)$-valued ZCR 
$\mathfrak{Z}_1\oplus\mathfrak{Z}_2$ can be regarded also 
as a formal ZCR with coefficients in $L(n)\oplus\mathfrak{so}_{n-1}$.

Recall that in Remark~\ref{remunxi} we have defined formal gauge transformations 
and the notion of gauge equivalence for formal ZCRs with coefficients 
in arbitrary (possibly infinite-dimensional) Lie algebras.
The results of Sections~\ref{secfd1},~\ref{secfdk} imply that 
the formal ZCR $\check{\ga}$, $\check{\gb}$ 
with coefficients in $L(n)\oplus\mathfrak{so}_{n-1}$ 
is gauge equivalent to the formal ZCR $\mathfrak{Z}_1\oplus\mathfrak{Z}_2$.

Let $\mg\subset\gl_\sm$ be a matrix Lie algebra.
Let $\mathfrak{R}$ be a $\mg$-valued ZCR of order~$\le\oc$.
So the ZCR $\mathfrak{R}$ is given by $\mg$-valued functions 
$A=A(x,t,u^j_0,u^j_1,\dots,u^j_\oc)$, $B=B(x,t,u^j_0,u^j_1,\dots,u^j_{\oc+2})$ 
satisfying $D_x(B)-D_t(A)+[A,B]=0$.

According to Theorem~\ref{thzcrfd}, 
there is a homomorphism $\hrf\cl\fds^\oc(\CE,a)\to\mg$ such that
the ZCR $A$, $B$ is gauge equivalent to the 
$\hrf\big(\fds^\oc(\CE,a)\big)$-valued ZCR 
\beq
\lb{tabhrf}
\tilde{A}=\hrf(\ga),\qquad\quad
\tilde{B}=\hrf(\gb).
\ee
Here $\tilde{A}$, $\tilde{B}$ are $\hrf\big(\fds^\oc(\CE,a)\big)$-valued 
functions, and formulas~\er{tabhrf} mean that $\hrf$ maps the coefficients 
of the power series $\ga$, $\gb$ to the corresponding 
coefficients of the Taylor series of the functions $\tilde{A}$, $\tilde{B}$.
The ZCR $\tilde{A}$, $\tilde{B}$ is $a$-normal.

Recall that $\ker\vf_\oc$ is a nilpotent ideal of $\fd^\oc(\CE,a)$.
This implies that $\hrf\big(\ker\vf_\oc\big)$ is a nilpotent ideal 
of the Lie subalgebra $\hrf\big(\fds^\oc(\CE,a)\big)\subset\mg$.

Since the homomorphism~\er{vflnso} is surjective, 
the algebra $L(n)\oplus\mathfrak{so}_{n-1}$ is isomorphic to 
$\fd^\oc(\CE,a)/\ker\vf_\oc$, and the surjective homomorphism 
$\hrf\cl\fds^\oc(\CE,a)\to\hrf\big(\fds^\oc(\CE,a)\big)$ induces 
a surjective homomorphism 
\beq
\lb{hhrf}
\hat{\hrf}\cl 
L(n)\oplus\mathfrak{so}_{n-1}\to
\hrf\big(\fds^\oc(\CE,a)\big)/\hrf\big(\ker\vf_\oc\big).
\ee

Consider the natural surjective homomorphism 
\beq
\lb{psihrf}
\psi\colon \hrf\big(\fds^\oc(\CE,a)\big)\to
\hrf\big(\fds^\oc(\CE,a)\big)/\hrf\big(\ker\vf_\oc\big).
\ee
Using this homomorphism and 
the $\hrf\big(\fds^\oc(\CE,a)\big)$-valued ZCR~\er{tabhrf},
we see that the functions $\psi(\tilde{A})$, $\psi(\tilde{B})$
form a ZCR with values in the Lie algebra 
$\hrf\big(\fds^\oc(\CE,a)\big)/\hrf\big(\ker\vf_\oc\big)$.

Formulas~\er{chgagb},~\er{tabhrf} and the definition 
of the homomorphism~\er{hhrf} imply the following
\beq
\lb{psiabh}
\psi(\tilde{A})=\hat{\hrf}\big(\check{\ga}\big),\qquad\quad
\psi(\tilde{B})=\hat{\hrf}\big(\check{\gb}\big).
\ee
Equations~\er{psiabh} say that 
the ZCR $\psi(\tilde{A})$, $\psi(\tilde{B})$ 
is a reduction of the ZCR $\check{\ga}$, $\check{\gb}$, 
which is gauge equivalent to the ZCR $\mathfrak{Z}_1\oplus\mathfrak{Z}_2$, 
as has been discussed above.

Recall that two ZCRs are called gauge equivalent 
if one is obtained from the other by means of a gauge transformation.
Thus, starting from an arbitrary ZCR $\mathfrak{R}$ 
given by functions $A$, $B$, we have done the following steps:
\begin{enumerate}
\item Using Theorem~\ref{thzcrfd}, 
we have obtained a homomorphism $\hrf\cl\fds^\oc(\CE,a)\to\mg$ 
and a gauge equivalent ZCR $\tilde{A}$, $\tilde{B}$ satisfying~\er{tabhrf}.
The ZCR $\tilde{A}$, $\tilde{B}$ takes values in the Lie subalgebra 
$\hrf\big(\fds^\oc(\CE,a)\big)\subset\mg$ and is obtained from the ZCR $A$, $B$ 
by means of a gauge transformation.
\item Killing the nilpotent ideal 
$\hrf\big(\ker\vf_\oc\big)\subset\hrf\big(\fds^\oc(\CE,a)\big)$, 
we have obtained the ZCR $\psi(\tilde{A})$, $\psi(\tilde{B})$, where $\psi$ 
is defined by~\er{psihrf}.
\item We have shown that the ZCR $\psi(\tilde{A})$, $\psi(\tilde{B})$ 
is a reduction of the ZCR $\check{\ga}$, $\check{\gb}$, 
which is gauge equivalent to the ZCR $\mathfrak{Z}_1\oplus\mathfrak{Z}_2$.
The ZCR $\check{\ga}$, $\check{\gb}$ is obtained from 
the ZCR $\mathfrak{Z}_1\oplus\mathfrak{Z}_2$ by means of a gauge transformation.
\end{enumerate}

Thus, for any given ZCR $\mathfrak{R}$, we have shown that, 
after suitable gauge transformations and after killing 
some nilpotent ideal in the corresponding Lie algebra, 
the ZCR $\mathfrak{R}$ becomes isomorphic to a reduction 
of the ZCR $\mathfrak{Z}_1\oplus\mathfrak{Z}_2$.

\begin{remark}
One has also a similar result in the case when a ZCR $\mathfrak{R}$ 
takes values in a matrix Lie algebra $\mg$ and depends on a parameter $\la$.

Suppose that the ZCR $\mathfrak{R}$ is given by functions
\beq
\lb{abla2}
A=A(\la,x,t,u^j_0,u^j_1,\dots,u^j_\oc),\quad 
B=B(\la,x,t,u^j_0,u^j_1,\dots,u^j_{\oc+2}),\quad
D_x(B)-D_t(A)+[A,B]=0,
\ee
where $\mg$-valued functions $A$, $B$ depend on $x$, $t$, $u^i_k$ 
and a parameter $\la$.

Let $\tilde{\mg}$ be the infinite-dimensional Lie algebra 
of functions $h(\la)$ with values in $\mg$. 
(Depending on the problem under study, one can consider analytic or meromorphic functions $h(\la)$.
Or one can assume that $\la$ runs through an open subset of some algebraic curve 
and consider $\mg$-valued functions $h(\la)$ on this algebraic curve.)

Then~\er{abla2} can be regarded as a ZCR with values in $\tilde{\mg}$.  
Using parameter-dependent versions of Theorems~\ref{evcov},~\ref{thzcrfd},
one can show that the ZCR~\er{abla2} is gauge equivalent to an $a$-normal ZCR 
\beq
\notag
\tilde{A}=\tilde{A}(\la,x,t,u^j_0,u^j_1,\dots,u^j_\oc),\quad 
\tilde{B}=\tilde{B}(\la,x,t,u^j_0,u^j_1,\dots,u^j_{\oc+2}),\quad
D_x(\tilde{B})-D_t(\tilde{A})+[\tilde{A},\tilde{B}]=0,
\ee
satisfying $\tilde{A}=\hrf(\ga)$, $\tilde{B}=\hrf(\gb)$ for some homomorphism 
$\hrf\cl\fds^\oc(\CE,a)\to\tilde{\mg}$. 
The functions $\tilde{A}$, $\tilde{B}$ take values in the Lie subalgebra 
$\hrf\big(\fds^\oc(\CE,a)\big)\subset\tilde{\mg}$.

Then one can kill the nilpotent ideal 
$\hrf\big(\ker\vf_\oc\big)\subset\hrf\big(\fds^\oc(\CE,a)\big)$
and proceed similarly to the steps described above.
\end{remark}

\section{The algebras $\fd^\oc(\CE,a)$ for the 
classical Landau-Lifshitz and nonlinear Schr\"odinger equations}
\label{sllnls}

In this section we assume $\fik=\Com$ and study the algebras $\fd^\oc(\CE,a)$ 
for the classical Landau-Lifshitz and nonlinear Schr\"odinger equations.

The classical Landau-Lifshitz equation reads
\beq
\lb{lle}
S_t=S\times S_{xx}+S\times JS,\qquad S=\big(S^1(x,t),S^2(x,t),S^3(x,t)\big),\qquad 
(S^1)^2+(S^2)^2+(S^3)^2=1,
\ee
where $J=\mathrm{diag}(j_1,j_2,j_3)$ is a constant diagonal $(3\times 3)$-matrix 
with $j_1,j_2,j_3\in\Com$, and the symbol $\times$ denotes the vector product. 
We consider the \emph{fully anisotropic} case $j_1\neq j_2\neq j_3\neq j_1$.

Let $\CE$ be the infinite prolongation of the PDE~\er{lle}.
Let $a\in\CE$.
We are going to describe the algebras $\fd^\oc(\CE,a)$ for this PDE.
\begin{remark}
Strictly speaking, 
in order to define $\CE$ and $\fd^\oc(\CE,a)$ for the PDE~\er{lle}, 
we need to resolve the constraint $(S^1)^2+(S^2)^2+(S^3)^2=1$ 
and to rewrite~\er{lle} as a $2$-component evolution PDE.
For example, Roelofs and Martini~\cite{ll} use the spherical coordinates
\beq
\lb{sssuv}
S^1=\cos v\sin u,\qquad
S^2=\sin v\sin u,\qquad
S^3=\cos u.
\ee
According to~\cite{ll},
using the transformation~\er{sssuv}, one can rewrite the PDE~\er{lle} 
as the $2$-component evolution PDE
\beq
\lb{uvt}
\begin{aligned}
u_t&=-(\sin u)v_{xx}-2(\cos u)u_xv_x+(j_1-j_2)\sin u\cos v\sin v,\\
v_t&=\frac{1}{\sin u}u_{xx}-(\cos u)v^2_{x}+(\cos u)(j_1\cos^2 v+j_2\sin^2 v-j_3).
\end{aligned}
\ee

One can use also some other way to resolve 
the constraint $(S^1)^2+(S^2)^2+(S^3)^2=1$ 
and to rewrite~\er{lle} as a $2$-component evolution PDE, 
then $\CE$ and $\fd^\oc(\CE,a)$ will be the same (up to isomorphism).
\end{remark}

We need some auxiliary constructions.
Let $\Com[v_1,v_2,v_3]$ be the algebra of polynomials in the variables $v_1$, $v_2$, $v_3$.
Recall that, by our assumption, the constants 
$j_1,j_2,j_3\in\Com$ satisfy $j_1\neq j_2\neq j_3\neq j_1$. 
Consider the ideal $\mathcal{I}_{j_1,j_2,j_3}\subset\Com[v_1,v_2,v_3]$ generated by the
polynomials
\begin{equation}
  \label{elc}
  v_\al^2-v_\beta^2+j_\al-j_\beta,\qquad\qquad \al,\,\beta=1,2,3.
\end{equation}

Set $E_{j_1,j_2,j_3}=\Com[v_1,v_2,v_3]/\mathcal{I}_{j_1,j_2,j_3}$. 
In other words, $E_{j_1,j_2,j_3}$ 
is the commutative associative algebra of polynomial 
functions on the algebraic curve 
in $\Com^3$ defined by the polynomials~\eqref{elc}.

Since we assume $j_1\neq j_2\neq j_3\neq j_1$, 
this curve is nonsingular and is of genus~$1$, so this is an elliptic curve. 
It is well known that the Landau-Lifshitz equation~\er{lle} possesses 
an $\mathfrak{so}_3(\Com)$-valued ZCR parametrized by points of this 
curve~\cite{sklyanin,ft}.


We have the natural surjective homomorphism 
$\mu\cl\Com[v_1,v_2,v_3]\to
\Com[v_1,v_2,v_3]/\mathcal{I}_{j_1,j_2,j_3}=E_{j_1,j_2,j_3}$.
Set $\bar v_i=\mu(v_i)\in E_{j_1,j_2,j_3}$ for $i=1,2,3$.

Consider also a basis $x_1$, $x_2$, $x_3$ of the Lie algebra
$\mathfrak{so}_3(\Com)$ such that $[x_1,x_2]=x_3$, $[x_2,x_3]=x_1$, $[x_3,x_1]=x_2$. 
We endow the space $\mathfrak{so}_3(\Com)\otimes_\Com E_{j_1,j_2,j_3}$ with 
the following Lie algebra structure 
$$
[z_1\otimes h_1,\,z_2\otimes h_2]=[z_1,z_2]\otimes h_1h_2,\qquad\quad 
z_1,z_2\in\mathfrak{so}_3(\Com),\qquad\quad h_1,h_2\in E_{j_1,j_2,j_3}.
$$

Denote by $\mR_{j_1,j_2,j_3}$ the Lie subalgebra of 
$\mathfrak{so}_3(\Com)\otimes_\Com E_{j_1,j_2,j_3}$ generated by the elements
$$
x_i\otimes\bar v_i\,\in\,\mathfrak{so}_3(\Com)\otimes_\Com E_{j_1,j_2,j_3},
\qquad\qquad i=1,2,3.
$$
Since $\mR_{j_1,j_2,j_3}\subset\mathfrak{so}_3(\Com)\otimes_\Com E_{j_1,j_2,j_3}$, 
we can regard elements of~$\mR_{j_1,j_2,j_3}$ 
as $\mathfrak{so}_3(\Com)$-valued functions on the elliptic curve in~$\Com^3$ 
determined by the polynomials~\eqref{elc}.
The paper~\cite{ll} describes a basis for $\mR_{j_1,j_2,j_3}$, 
which implies that the algebra $\mR_{j_1,j_2,j_3}$ is infinite-dimensional. 

\begin{theorem}
\label{thlla}
Recall that $J=\mathrm{diag}(j_1,j_2,j_3)$ in~\er{lle} 
is a constant diagonal $(3\times 3)$-matrix with $j_1,j_2,j_3\in\Com$.
We consider the case $j_1\neq j_2\neq j_3\neq j_1$.

Let $\CE$ be the infinite prolongation of the Landau-Lifshitz equation~\er{lle}.
Let $a\in\CE$.

The Lie algebras $\fd^\oc(\CE,a)$, $\oc\in\zp$, for this PDE have the following structure.

The algebra $\fd^0(\CE,a)$ is isomorphic to $\mR_{j_1,j_2,j_3}$.

For each $\oc\in\zp$, 
there are an abelian ideal $\mathcal{I}_{\oc+1}$ 
of the Lie algebra $\fd^{\oc+1}(\CE,a)$
and an abelian ideal $\mathcal{J}_{\oc+1}$ of the quotient Lie algebra 
$\fd^{\oc+1}(\CE,a)/\mathcal{I}_{\oc+1}$ such that
\begin{itemize}
\item $\big[\mathcal{I}_{\oc+1},\,\fd^{\oc+1}(\CE,a)\big]=0$, 
so $\mathcal{I}_{\oc+1}$ lies in the center of the Lie algebra $\fd^{\oc+1}(\CE,a)$,
\item
$\big[\mathcal{J}_{\oc+1},\,\fd^{\oc+1}(\CE,a)/\mathcal{I}_{\oc+1}\big]=0$, 
so $\mathcal{J}_{\oc+1}$ lies in the center of the Lie algebra 
$\fd^{\oc+1}(\CE,a)/\mathcal{I}_{\oc+1}$,
\item the algebra $\fd^{\oc}(\CE,a)$ is isomorphic to the quotient algebra
$\big(\fd^{\oc+1}(\CE,a)/\mathcal{I}_{\oc+1}\big)/\mathcal{J}_{\oc+1}$,
and the surjective homomorphism $\fd^{\oc+1}(\CE,a)\to\fd^{\oc}(\CE,a)$ 
from~\er{fdnn-1} coincides with the composition of the natural homomorphisms 
$$
\fd^{\oc+1}(\CE,a)\to\fd^{\oc+1}(\CE,a)/\mathcal{I}_{\oc+1}\to
\big(\fd^{\oc+1}(\CE,a)/\mathcal{I}_{\oc+1}\big)/\mathcal{J}_{\oc+1}\cong
\fd^{\oc}(\CE,a),
$$
which means that the Lie algebra $\fd^{\oc+1}(\CE,a)$ is obtained from 
the Lie algebra $\fd^{\oc}(\CE,a)$ 
by applying two times the operation of central extension.
\end{itemize}

For each $q\in\zsp$, the Lie algebra $\fd^q(\CE,a)$ is obtained from 
the Lie algebra $\fd^0(\CE,a)$ 
by applying several times the operation of central extension.
The kernel of the surjective homomorphism 
$\fd^q(\CE,a)\to\fd^{0}(\CE,a)$ from~\er{fdnn-1} is nilpotent.
\end{theorem}
\begin{proof}
Recall that, using the transformation~\er{sssuv}, 
one can rewrite the PDE~\er{lle} in the form~\er{uvt}.
In what follows, when we speak about the WE algebra 
and the algebras $\fd^\oc(\CE,a)$ for the PDE~\er{lle},
we assume that the PDE~\er{lle} is written in the form~\er{uvt}.

Let $\wea$ be the Wahlquist-Estabrook prolongation algebra 
(WE algebra) for the PDE~\er{lle}.
It is shown in~\cite{ll} that this algebra is isomorphic to the direct sum 
of~$\mR_{j_1,j_2,j_3}$ and the $2$-dimensional abelian Lie algebra $\Com^2$.
So $\wea\cong\mR_{j_1,j_2,j_3}\oplus\Com^2$.

In Section~\ref{fd0we} for any evolution PDE of the form~\er{gevxt} 
we have defined the notion of formal Wahlquist-Estabrook ZCR, 
which is given by formulas~\er{wgawgbser},~\er{wxgbtga}.
According to~\cite{ll}, in the formal Wahlquist-Estabrook ZCR 
with coefficients in~$\wea$ for the PDE~\er{lle} one has 
\begin{equation}
\notag
\wga=p_1S^1+p_2S^2+p_3S^3+p_4,\qquad\quad 
p_1,p_2,p_3,p_4\in\wea\cong\mR_{j_1,j_2,j_3}\oplus\Com^2,
\end{equation}
where $S^1$, $S^2$, $S^3$ are given by~\er{sssuv}, 
the elements $p_1$, $p_2$, $p_3$ generate the Lie subalgebra 
$\mR_{j_1,j_2,j_3}\subset\wea\cong\mR_{j_1,j_2,j_3}\oplus\Com^2$, 
and one has $[p_4,p_l]=0$ for all $l=1,2,3$.

This implies that the subalgebra $\swe\subset\wea$ defined in Theorem~\ref{thmhfd0} 
in Section~\ref{fd0we} is equal to $\mR_{j_1,j_2,j_3}\subset\wea$.
According to Theorem~\ref{thmhfd0}, one has $\fd^0(\CE,a)\cong\swe$. 
Since in the considered case we have $\swe=\mR_{j_1,j_2,j_3}$, 
we see that $\fd^0(\CE,a)$ is isomorphic to $\mR_{j_1,j_2,j_3}$.

It remains to prove the statements about 
$\fd^{\oc+1}(\CE,a)$, $\mathcal{I}_{\oc+1}$, $\mathcal{J}_{\oc+1}$, 
$\fd^{\oc}(\CE,a)$, $\fd^q(\CE,a)$, $\fd^q(\CE,a)\to\fd^{0}(\CE,a)$
in Theorem~\ref{thlla}.

The statements about $\fd^{\oc+1}(\CE,a)$, $\mathcal{I}_{\oc+1}$, 
$\mathcal{J}_{\oc+1}$, $\fd^{\oc}(\CE,a)$ can be proved similarly to the results 
of Sections~\ref{secfd1},~\ref{secfdk}.

The homomorphism $\fd^q(\CE,a)\to\fd^{0}(\CE,a)$ is equal to the composition 
of the surjective homomorphisms
$$
\fd^q(\CE,a)\to\fd^{q-1}(\CE,a)\to\dots\to\fd^1(\CE,a)\to\fd^{0}(\CE,a)
$$
from~\er{fdnn-1}. It is easily seen that the statements about 
$\fd^{\oc+1}(\CE,a)$, $\mathcal{I}_{\oc+1}$, $\mathcal{J}_{\oc+1}$, $\fd^{\oc}(\CE,a)$
imply the statements about $\fd^q(\CE,a)$ and $\fd^q(\CE,a)\to\fd^{0}(\CE,a)$.
\end{proof}


The nonlinear Schr\"odinger (NLS) equation 
is another well-known PDE from mathematical physics 
(see, e.g,~\cite{ft}). It can be written as follows
\beq
\lb{psinls}
i\psi_t+\psi_{xx}-\kappa\bar{\psi}\psi^2=0,\qquad\quad
i=\sqrt{-1}\in\Com,
\ee
where $\kappa\in\mathbb{R}$ is a nonzero constant.
The function $\psi=\psi(x,t)$ takes values in $\Com$, and one has 
\beq
\lb{psiuu}
\psi=u^1+iu^2,\qquad\quad
\bar{\psi}=u^1-iu^2,\qquad\quad i=\sqrt{-1}\in\Com,
\ee
for some $\mathbb{R}$-valued functions $u^1=u^1(x,t)$ and $u^2=u^2(x,t)$.

Taking into account~\er{psiuu}, 
we see that the NLS equation~\er{psinls}
is equivalent to the $2$-component evolution PDE
\beq
\lb{nlseq}
\begin{aligned}
u^1_t&=-u^2_{xx}+\kappa u^2\big((u^1)^2+(u^2)^2\big),\\
u^2_t&=u^1_{xx}-\kappa u^1\big((u^1)^2+(u^2)^2\big).
\end{aligned}
\ee
In what follows, when we speak about the NLS equation, we mean the PDE~\er{nlseq}.

So we will study the evolution PDE~\er{nlseq}, where $\kappa$ is a nonzero constant.
Since in this section we work over $\Com$, we assume that $u^1$, $u^2$ in~\er{nlseq} 
take values in $\Com$.

Let $\Com[\la]$ be the algebra of polynomials in the variable $\la$.
Consider the infinite-dimensional Lie algebra
$$
\sl_2(\Com[\la])=\sl_2(\Com)\otimes_\Com\Com[\la].
$$
\begin{theorem}
\lb{thnls}
Let $\CE$ be the infinite prolongation of the NLS equation~\er{nlseq}.
Let $a\in\CE$.

The Lie algebras $\fd^\oc(\CE,a)$, $\oc\in\zp$, for this PDE have the following structure.

The algebra $\fd^0(\CE,a)$ is isomorphic to 
the direct sum of $\sl_2(\Com[\la])$ and a one-dimensional abelian Lie algebra. 

For each $\oc\in\zp$, 
there are an abelian ideal $\mathcal{I}_{\oc+1}$ 
of the Lie algebra $\fd^{\oc+1}(\CE,a)$
and an abelian ideal $\mathcal{J}_{\oc+1}$ of the quotient Lie algebra 
$\fd^{\oc+1}(\CE,a)/\mathcal{I}_{\oc+1}$ such that
\begin{itemize}
\item $\big[\mathcal{I}_{\oc+1},\,\fd^{\oc+1}(\CE,a)\big]=0$, 
so $\mathcal{I}_{\oc+1}$ lies in the center of the Lie algebra $\fd^{\oc+1}(\CE,a)$,
\item
$\big[\mathcal{J}_{\oc+1},\,\fd^{\oc+1}(\CE,a)/\mathcal{I}_{\oc+1}\big]=0$, 
so $\mathcal{J}_{\oc+1}$ lies in the center of the Lie algebra 
$\fd^{\oc+1}(\CE,a)/\mathcal{I}_{\oc+1}$,
\item the algebra $\fd^{\oc}(\CE,a)$ is isomorphic to the quotient algebra
$\big(\fd^{\oc+1}(\CE,a)/\mathcal{I}_{\oc+1}\big)/\mathcal{J}_{\oc+1}$,
and the surjective homomorphism $\fd^{\oc+1}(\CE,a)\to\fd^{\oc}(\CE,a)$ 
from~\er{fdnn-1} coincides with the composition of the natural homomorphisms 
$$
\fd^{\oc+1}(\CE,a)\to\fd^{\oc+1}(\CE,a)/\mathcal{I}_{\oc+1}\to
\big(\fd^{\oc+1}(\CE,a)/\mathcal{I}_{\oc+1}\big)/\mathcal{J}_{\oc+1}\cong
\fd^{\oc}(\CE,a),
$$
which means that the Lie algebra $\fd^{\oc+1}(\CE,a)$ is obtained from 
the Lie algebra $\fd^{\oc}(\CE,a)$ 
by applying two times the operation of central extension.
\end{itemize}

For each $q\in\zsp$, the Lie algebra $\fd^q(\CE,a)$ is obtained from 
the Lie algebra $\fd^0(\CE,a)$ 
by applying several times the operation of central extension.
The kernel of the surjective homomorphism 
$\fd^q(\CE,a)\to\fd^{0}(\CE,a)$ from~\er{fdnn-1} is nilpotent.
\end{theorem}
\begin{proof}
Let $\wea$ be the Wahlquist-Estabrook prolongation algebra 
(WE algebra) for the NLS equation~\er{nlseq}.
It is shown in~\cite{we-nls,eck-lmp86} that this algebra 
is isomorphic to the direct sum of $\sl_2(\Com[\la])$ 
and a $3$-dimensional abelian Lie algebra. 

To describe $\fd^0(\CE,a)$, we can use again Theorem~\ref{thmhfd0}, which
allows us to describe $\fd^0(\CE,a)$ as a certain subalgebra of $\wea$.
Applying Theorem~\ref{thmhfd0} to the description of the 
Wahlquist-Estabrook prolongation algebra $\wea$ in~\cite{we-nls,eck-lmp86}, 
one obtains that $\fd^0(\CE,a)$ is isomorphic to 
the direct sum of $\sl_2(\Com[\la])$ and a one-dimensional abelian Lie algebra. 

It remains to prove the statements about 
$\fd^{\oc+1}(\CE,a)$, $\mathcal{I}_{\oc+1}$, $\mathcal{J}_{\oc+1}$, 
$\fd^{\oc}(\CE,a)$, $\fd^q(\CE,a)$, $\fd^q(\CE,a)\to\fd^{0}(\CE,a)$
in Theorem~\ref{thnls}.

The statements about $\fd^{\oc+1}(\CE,a)$, $\mathcal{I}_{\oc+1}$, 
$\mathcal{J}_{\oc+1}$, $\fd^{\oc}(\CE,a)$ can be proved similarly to the results 
of Sections~\ref{secfd1},~\ref{secfdk}.

The homomorphism $\fd^q(\CE,a)\to\fd^{0}(\CE,a)$ is equal to the composition 
of the surjective homomorphisms
$$
\fd^q(\CE,a)\to\fd^{q-1}(\CE,a)\to\dots\to\fd^1(\CE,a)\to\fd^{0}(\CE,a)
$$
from~\er{fdnn-1}. It is easily seen that the statements about 
$\fd^{\oc+1}(\CE,a)$, $\mathcal{I}_{\oc+1}$, $\mathcal{J}_{\oc+1}$, $\fd^{\oc}(\CE,a)$
imply the statements about $\fd^q(\CE,a)$ and $\fd^q(\CE,a)\to\fd^{0}(\CE,a)$.
\end{proof}



\section*{Acknowledgments}

Sergei Igonin is a research fellow of 
Istituto Nazionale di Alta Matematica (INdAM), Italy. 
Sergei Igonin and Gianni Manno are members of GNSAGA of INdAM.

The authors thank T.~Skrypnyk and V.~V.~Sokolov for useful discussions. 

Sergei Igonin is grateful to the Max Planck Institute for Mathematics (Bonn, Germany) 
for its hospitality and excellent working conditions 
during 02.2006--01.2007 and 06.2010--09.2010, when part of this research was done.

\end{document}